\documentclass[final,3p, review,11pt]{elsarticle}

\usepackage{amssymb}
\usepackage{natbib}
\usepackage{graphicx}
\usepackage[usenames,dvipsnames]{color}
\usepackage{amsmath}
\usepackage{bm}
\usepackage{multirow}
\usepackage{hyperref}
\newcommand{\Bezier}{B\'{e}zier\ }
\newcommand{\D}{\displaystyle}

\newlength\figwidth
\usepackage{amsthm}

\renewcommand\figurename{Fig. }
\makeatletter
\renewcommand{\fnum@figure}[1]{\textbf{\figurename\thefigure.} }
\makeatother

\hypersetup{colorlinks=true}

\newtheorem{lemma}{Lemma}[section]

\newtheorem{proposition}[lemma]{Proposition}

\def\pb{\textbf{\emph{P}}}

\usepackage[switch]{lineno}

\journal{CMAME}

\begin{document}

\begin{frontmatter}



\title{Constructing IGA-suitable planar  parameterization from complex
  CAD boundary by domain partition and global/local optimization}


\author[label1,B]{Gang Xu}
\ead{gxu@hdu.edu.cn,xugangzju@gmail.com}
\author[label2]{Ming Li}
\author[label3]{Bernard Mourrain}
\author[label4]{Timon Rabczuk}
\author[label1]{Jinlan Xu}
\author[label5]{St\'ephane P.A. Bordas}
\address[label1]{School of Computer Science and Technology, Hangzhou
  Dianzi University, Hangzhou 310018, China}
 \address[B]{Key Laboratory of Complex Systems Modeling and Simulation,
      Ministry of Education, Hangzhou 310018, China}
\address[label2]{State Key Laboratory of CAD $\&$CG, Zhejiang
  University, Hangzhou 310058, P.R China}
\address[label3]{INRIA Sophia-Antipolis, 2004 Route des Lucioles,
  06902 Cedex, France}
\address[label4]{Institute of Structural Mechanics, Bauhaus-University Weimar, Marienstr. 15, D-99423 Weimar, Germany}
\address[label5]{Research Unit in Engineering, University of Luxembourg,
   Luxembourg}

\begin{abstract}
In this paper, we propose a general framework for constructing IGA-suitable planar B-spline parameterizations from given complex CAD boundaries consisting of a set of B-spline curves. Instead of forming the computational domain by a simple boundary,  planar domains with high genus and more complex boundary curves are considered. Firstly, some pre-processing operations including \Bezier extraction and subdivision are performed on each boundary
curve in order to generate a high-quality planar parameterization; then a robust planar domain partition framework is proposed to construct
high-quality patch-meshing results with few singularities from the discrete boundary formed by connecting the end points of the resulting boundary segments. After the topology information generation of quadrilateral decomposition, the optimal placement of interior \Bezier curves corresponding to the interior edges of the
quadrangulation is constructed by a global optimization method to achieve a patch-partition with high quality. Finally, after the imposition of
$C^1/G^1$-continuity constraints on the interface of neighboring
\Bezier patches with respect to each quad in the quadrangulation, the high-quality \Bezier
patch parameterization is obtained by a $C^1$-constrained local optimization method to
achieve uniform and orthogonal iso-parametric structures while keeping
the continuity conditions between patches. The efficiency and robustness of the
proposed method are demonstrated by several examples which are compared to results obtained by the skeleton-based parameterization approach.
\end{abstract}


\end{frontmatter}


\section{Introduction}
\label{sec:introduction}

In isogeometric analysis (IGA) \cite{hughes:CMAME 2005}, the parameterization of the computational domain corresponds to the mesh
generation in finite element analysis. Consequently, it has a great effect on
the subsequent analysis accuracy and efficiency \cite{cohen:analysis
  aware,xu:gmp10,Pilger:2013Bounding}. As proposed in
\cite{xu:cmame2011}, an analysis-suitable parameterization of the
computational domain should satisfy three
requirements: 1) it has no
self-intersections, i.e, the mapping from the parametric domain to the
physical domain should be injective; 2) the iso-parametric
elements should be as uniform as possible; 3) the iso-parametric
structure should be as orthogonal as possible. From the given
boundary information with spline representations,  several approaches such
as constrained optimization methods \cite{xu:spm2012}, variational
harmonic methods \cite{xu:jcp2013}, divide-and-conquer techniques
\cite{qian:spm2013}, the Teichm\"{u}ller
mapping method \cite{Nian:cmame2016},
skeleton-based decomposition method \cite{xu:cmame2015},
multi-patch parameterization method \cite{BucheCAD2016} , and
parameterization with non-standard B-splines (i.e, T-splines \cite{Zhang:cm2013},
THB-splines \cite{falini:THBparameter}, PHT/RHT-splines \cite{Timon:cmame2017,Timon:cm2014,Timon:cmame2017a,Timon:cmame2017b}, Powell-Sabin splines \cite{Speleers:CAM2015}
and subdivision surfaces \cite{Pan:JCP2017,Pan:JCP2015}),  have been proposed to
address the parameterization problem of the
computational domain. However, most of the above methods only focus on  the
computational domain with simple boundaries. The
construction of analysis-suitable parameterizations of complex computational
domains with standard B-spline boundaries remains one of the most significant
challenges in IGA.

In this paper, a general framework  for constructing IGA-suitable planar
parameterization from complex CAD boundaries consisting of
standard B-spline curves is proposed.
Instead of the computational domain formed by four
boundary curves, the planar domains with high genus and more complex boundary curves
are considered. Our main contributions can be summarized as follows:

\begin{itemize}
\item We propose a general framework for IGA-suitable planar  parameterizations from complex
CAD boundaries by domain partition and global/local optimization.
\item Given a complex planar domain bounded by B-spline curves, a
  novel method is proposed to construct  four-sided patch partitions
  with high quality by global optimization and  quad-meshing with few singularities.
\item A $C^1/G^1$-constrained local optimization approach is proposed to construct the
  \Bezier patch with respect to each quad in the quadrangulation achieving high-quality iso-parametric structures while
  keeping the continuity constraints between patches.
\end{itemize}

The rest of the paper is structured as follows. Some related work on
parameterization of the computational domain are reviewed in Section
\ref{sec:related}. Preliminary properties of Bernstein
polynomials and  an overview of the proposed framework is given in
Section \ref{sec:preliminary}. Several pre-processing operations, including \Bezier extraction
  and \Bezier subdivision, are presented in Section
  \ref{sec:preprocessing}. Section \ref{sec:partition} describes the
topology information generation and
  interior \Bezier boundary construction for the quadrilateral high-quality patch
  partition by the global optimization approach. A local optimization method for
high-quality \Bezier patch parameterizations with $C^1/G^1$ continuity
constraints is proposed in Section
\ref{sec:local}. Some parameterization examples are presented
in Section \ref{sec:example}. To demsonstrate the effectiveness of the
proposed method, the results are compared to results obtained by the  skeleton-based approach.  Finally, we conclude this paper and
outline future work in Section \ref{sec:conclude}.

\section{Related work}
\label{sec:related}

Currently, the related work on parameterization of the computational domain
in IGA can be classified into four categories: (1)
analysis-aware optimal parameterization; (2) volumetric spline
parameterization from boundary triangulation; (3) analysis-suitable
planar parameterization;  (4) analysis-suitable
volumetric parameterization from spline boundaries .

\noindent \textbf{Analysis-aware optimal parameterization}:
E. Cohen et al. \cite{cohen:analysis aware} proposed the concept of
\emph{ analysis-aware modeling},  in which the parameters of CAD
models  are selected to facilitate isogeometric analysis. Xu et
al.  \cite{xu:gmp10}  showed that the
quality of parameterization has a great impact on the analysis results and the
efficiency. Pilgerstorfer and J\"uttler \cite{Pilger:2013Bounding} showed that the condition number of the stiffness matrix, which is a key
factor for the stability of the linear system, depends strongly on the
quality of the domain parameterization.

\noindent \textbf{Volumetric spline parameterization from boundary
triangulation}: Using volumetric harmonic functions,
Martin et al. \cite{martin:CAGD2009}  proposed a fitting method for triangular meshes by B-spline parametric volumes. In \cite{Escobara:CMAME2011},  a
method is proposed to construct
trivariate T-spline volumetric parameterizations for genus-zero solids
based on an adaptive tetrahedral meshing and mesh untangling
technique. Zhang et al. proposed a robust and efficient approach to
construct injective solid T-splines for genus-zero geometries
from a boundary triangulation \cite{Zhang:cmame2012}. Chan et al. proposed a volumetric
parameterization method with PHT-splines from the level-set boundary
representation \cite{chan:CAD2017}. For meshes with
arbitrary topology,  volumetric parameterization methods are proposed
from the Morse theory \cite{Wang:spm12} and Boolean operations
\cite{zhang:mr2013}.

\noindent \textbf{Analysis-suitable planar parameterization}:
The boundary in CAD is usually provided in spline form.  Xu et al. proposed a
constrained optimization method to construct injective
planar parameterizations \cite{xu:cmame2011}. Gravessen et
al. \cite{gravessen:parameterization} investigated the planar
parameterization  problem within a general non-linear
optimization framework, in which several objective functions related
to the parameterization quality are introduced.
Xu et al. proposed a skeleton-guided
domain decomposition method for planar parameterization with
$C^0$-continuity between patches \cite{xu:cmame2015}. Speleers and Manni
proposed a parameterization method with $C^1$ Powell-Sabin splines defined
on triangulation \cite{Speleers:CAM2015};  truncated hierarchical
B-splines are used for the planar parameterization problems in
\cite{falini:THBparameter}.   Nian and Chen \cite{Nian:cmame2016}
proposed an approach for planar domain parameterization based on
Teichm\"{u}ller mapping, which can generate a bijective high-quality
parameterization from four specified boundary curves. For the multi-patch
parameterization problem, Buchegger and J\"uttler \cite{BucheCAD2016}
proposed a systematic method  to explore the different possible
parameterizations of a planar domain by collections of quadrilateral
patches. For the geometrically continuous spline space over multi-patch domains, Kapl et al. proposed the construction of geometrically
continuous isogeometric functions which are defined on two-patch domains \cite{KaplCMA2015}. Bases for bicubic and biquartic geometrically
continuous isogeometric functions on bilinearly parameterized
multi-patch domains are constructed in \cite{KaplCMAME2016} ; Mourrain
et al. analyzed the space of geometrically continuous piecewise
polynomial splines for rectangular and triangular patches with arbitrary topology and general
rational transition maps \cite{Mourrain:CAGD2016};  Buchegger et al. constructed
the  truncated hierarchical B-spline basis
for the space of adaptively refined spline functions on multi-patch
domains with enhanced smoothness across interfaces
\cite{BucheAMC2016}.
Overall, it is still
an open problem how to construct
analysis-suitable planar parameterization from complex CAD boundary.

\noindent \textbf{Analysis-suitable volumetric parameterizations from
  spline boundary}:  A variational approach for constructing  NURBS parameterizations of
swept volumes is proposed by M. Aigner et al \cite{Aigner:swept}.  Xu et al. proposed a constrained optimization
framework to construct analysis-suitable volume
parameterizations \cite{xu:spm2012}.
Spline volume faring is proposed by Pettersen and Skytt
to obtain high-quality volume parameterization  \cite{Pettersen:2012splinefaring}.
The construction of conformal solid T-splines from boundary T-spline
representations is studied by using octree structure and boundary
offset \cite{Zhang:cm2013}. In \cite{xu:jcp2013}, a variational harmonic
method is proposed to construct analysis-suitable parameterizations of  computational domains
from given CAD boundary information. Wang and Qian proposed an
efficient method by combining divide-and-conquer, constraint aggregation and the hierarchical
optimization technique to obtain valid trivariate B-spline solids from
six boundary B-spline surfaces \cite{qian:spm2013}.
Analysis-suitable trivariate NURBS
representations of composite panels  is constructed with a new
curve/surface offset algorithm \cite{Nguyen:CAD2014}.
Xu et al. proposed  a two-stage
scheme to construct the analysis-suitable NURBS volumetric parameterization
by a uniformity-improved boundary reparameterization method \cite{xu:cm2014}.
Recently, given a template domain, B-spline based consistent volumetric parameterization is proposed  for a set of models with similar semantic features \cite{xu:cad17}.

\section{Preliminary properties of Bernstein polynomials
and framework overview}

\label{sec:preliminary}

\subsection{Preliminary properties of  Bernstein polynomials}

\label{subsec:prebernstein}

Since some  properties of Bernstein
polynomials \cite{Farin:CAGD} will
be applied in our framework, they are reviewed subsequently.
\begin{lemma}\label{lemma:product1} Product of Bernstein polynomials
\begin{equation}
B_{i}^{m}(t) B_{j}^{n}(t) = \frac{{m \choose i}{n \choose j}}{{m+n \choose i+j}} B_{{i+j}}^{m+n}(t)
\end{equation}
\end{lemma}
\begin{lemma} Integration of Bernstein polynomials
\begin{equation}\label{eq:integration}
\int_{0}^{1}B_{i}^{m}(t) dt= \frac{1}{m+1}
\end{equation}
\end{lemma}
\begin{lemma} Degree elevation of Bernstein polynomials
\begin{equation}
B_{i}^{n-1}(t) = \frac{n-i}{n} B_{i}^{n}(t) + \frac{i+1}{n}
B_{i+1}^{n}(t)
\end{equation}
\end{lemma}

From Lemma \ref{lemma:product1}, we have the following proposition
\cite{Farin:CAGD}

\begin{proposition}  \label{prop:prop}
Let $ R(t)$ and $S(t)$ be a parametric function defined by
\[
R(t)=\sum_{i=0}^{{\ell}_1}  a_{i} B_i^{{\ell}_1}(t) ,  \qquad S(t)=\sum_{i=0}^{{\ell}_2}  b_{i} B_i^{{\ell}_2}(t),
\]
where $r_{i}$ and $s_{i}$ are scale values. Then the product of  $
R(t)$ and $S(t)$ can be defined as
\begin{equation}\label{eq:product}
R(t)  S(t)=\sum_{i=0}^{{\ell}_1+{\ell}_2} c_{i} B_i^{{\ell}_1+{\ell}_2}(t),
\end{equation}
where
\[
c_{i}=\sum_{r=\text{max}(0,i-{\ell}_1)}^{\text{min}(i,{\ell}_2)} \frac{{{\ell}_1 \choose
    r}{{\ell}_2 \choose i-r}}{{{\ell}_1+{\ell}_2 \choose i}}a_i b_{i-r}
\]
\end{proposition}

\begin{figure*}[!htb]
\centering
\begin{minipage}[t]{2.1in}
\centering
\includegraphics[width=2.1in,height=2.1in]{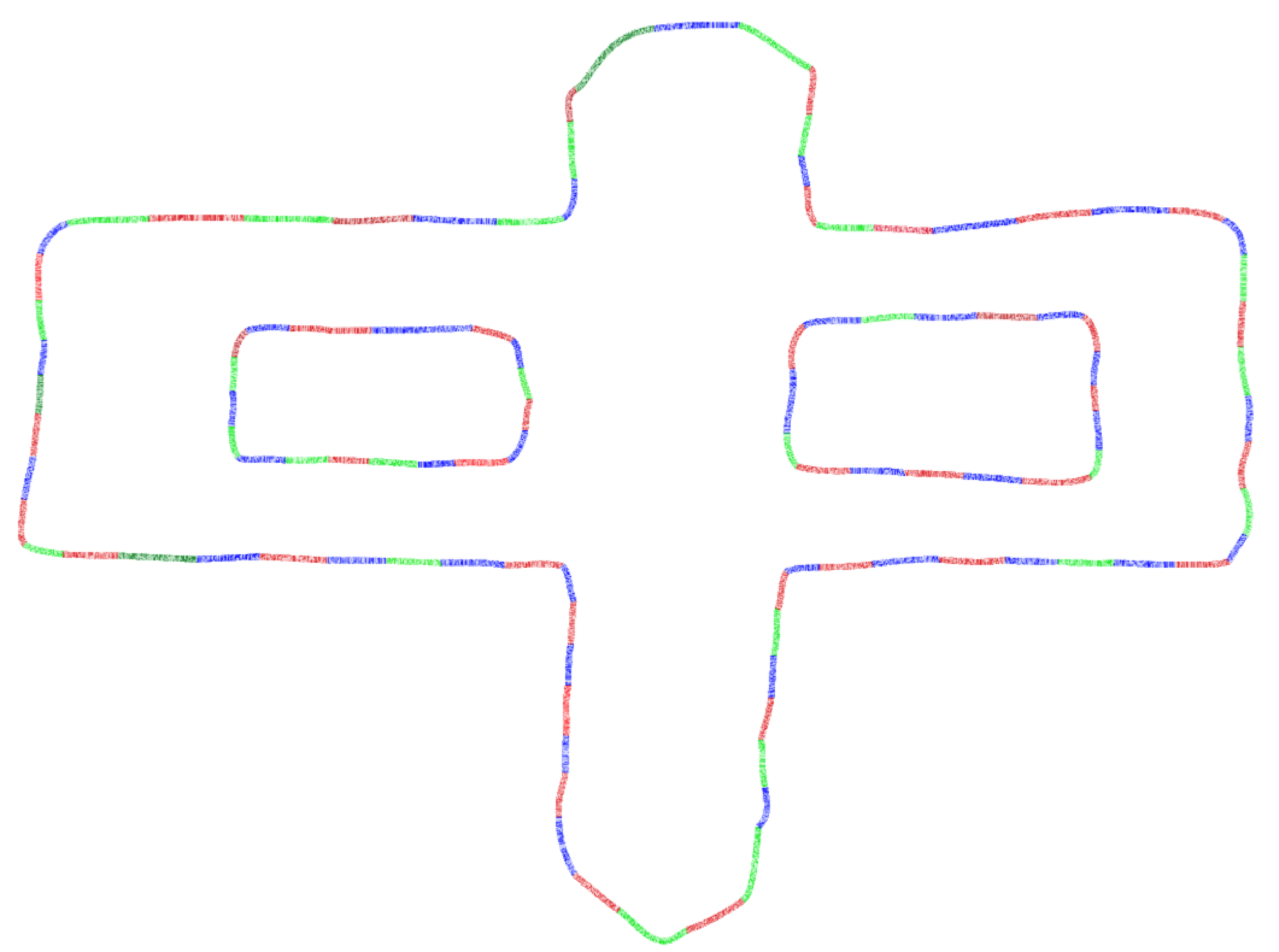}
\\ (a) boundary \Bezier curves
\end{minipage}
\begin{minipage}[t]{2.1in}
\centering
\includegraphics[width=2.25in]{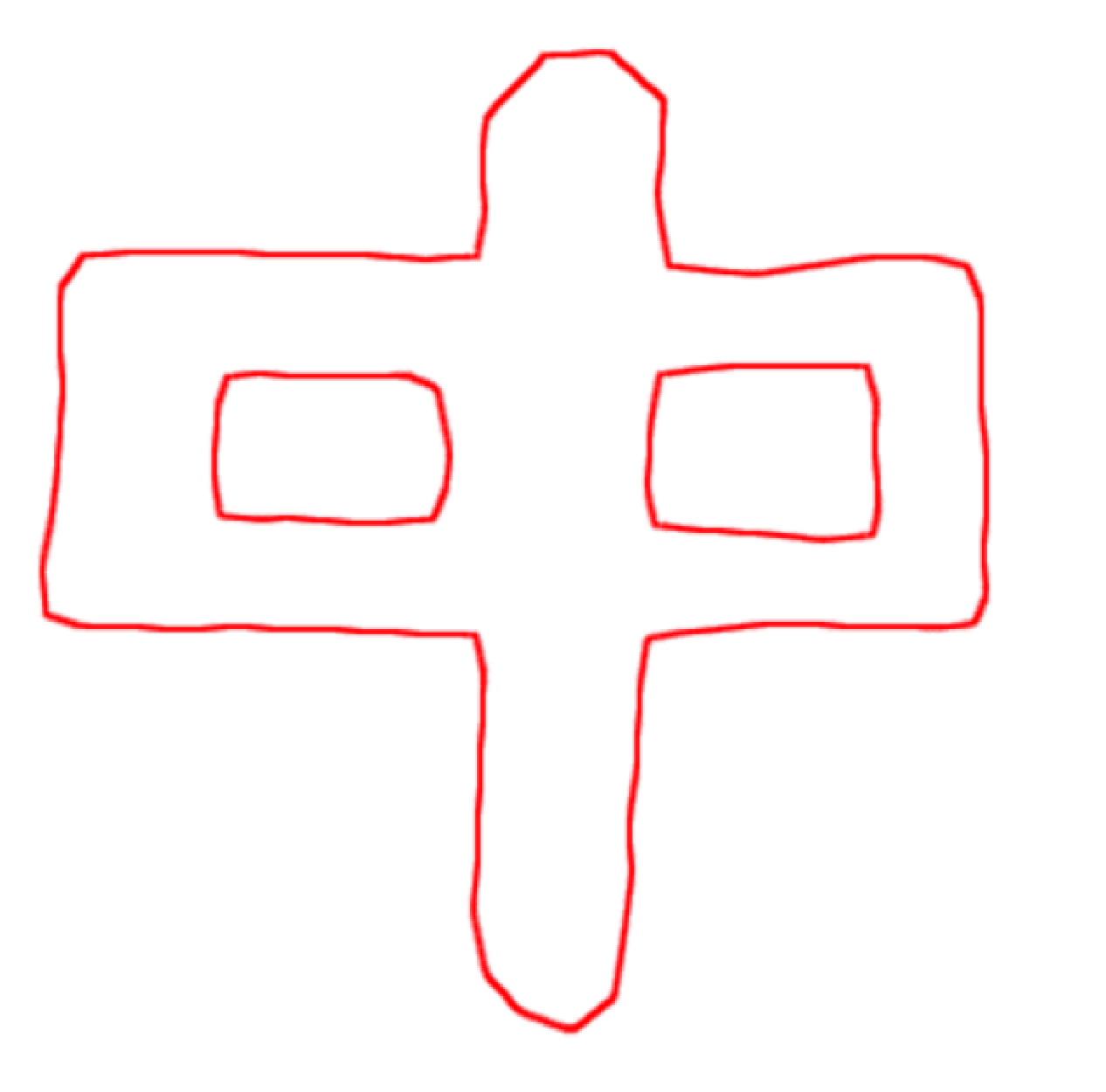}
\\ (b) discrete boundary
\end{minipage}
\begin{minipage}[t]{2.1in}
\centering
\includegraphics[width=2.25in]{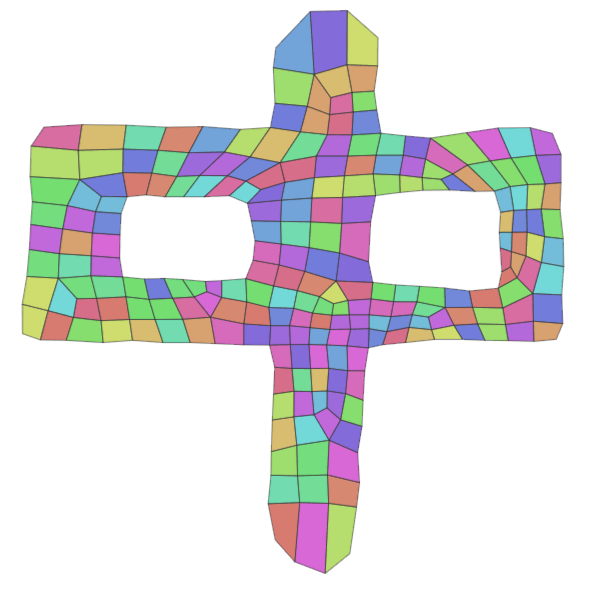}
\\ (c) quad meshing result
\end{minipage}\\
\begin{minipage}[t]{2.1in}
\centering
\includegraphics[width=2.08in]{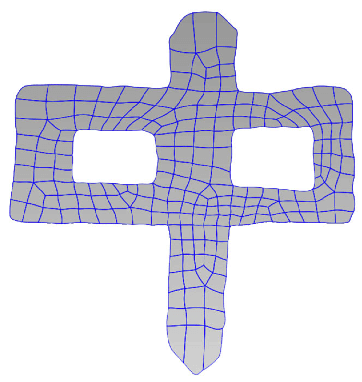}
\\ (d) segmentation curves
\end{minipage}
\begin{minipage}[t]{2.1in}
\centering
\includegraphics[width=2.1in]{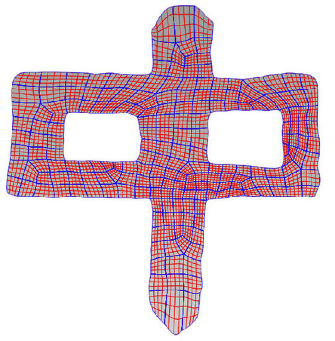}
\\ (e) parameterization result
\end{minipage}
\begin{minipage}[t]{2.1in}
\centering
\includegraphics[width=2.75in]{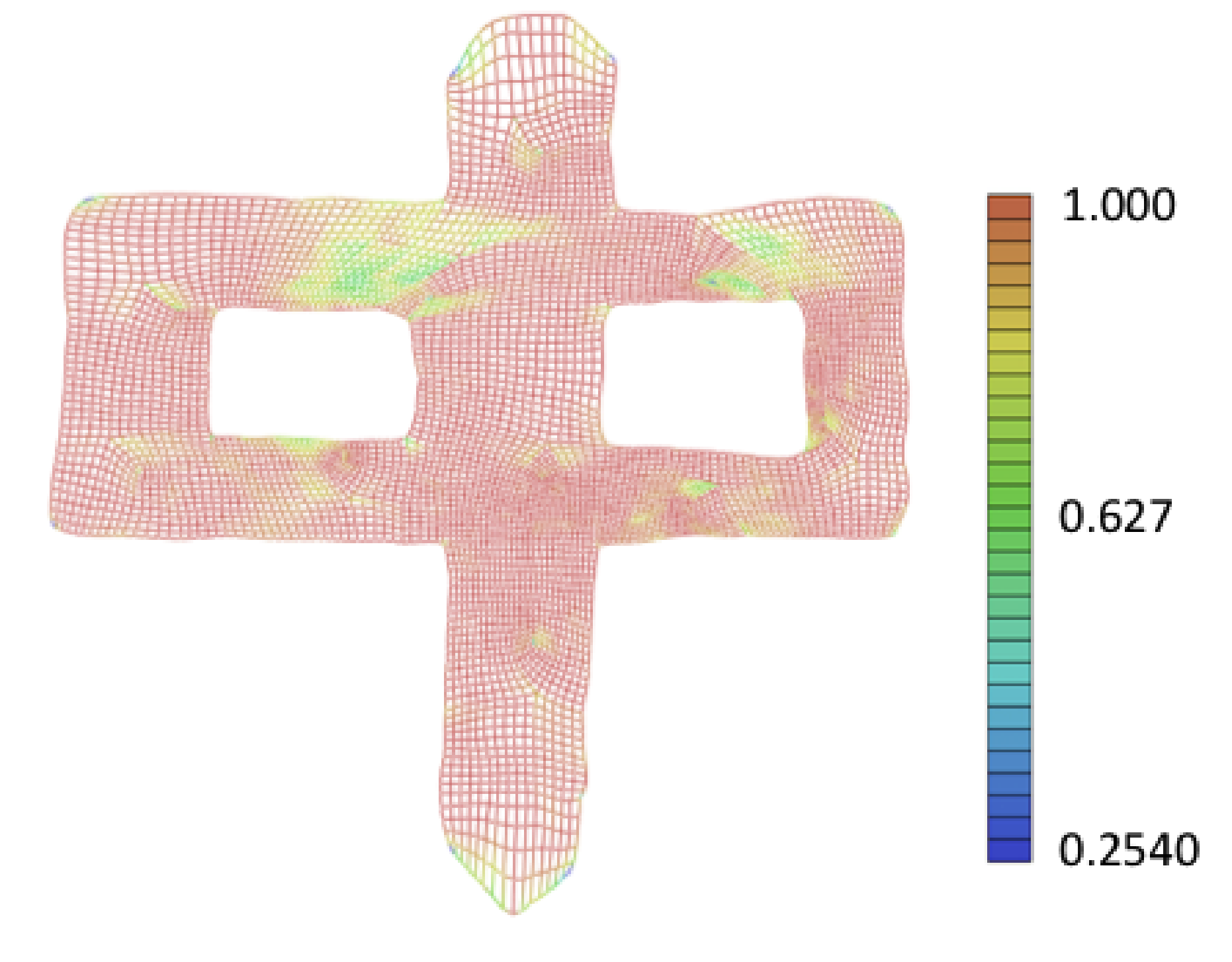}
\\ (f) Jacobian colormap
\end{minipage}
\caption{Illustration example I of the proposed parameterization
  framework.   (a) shows the boundary \Bezier curves after
the pre-processing of input boundary B-spline curves;
  (b) presents the discrete boundary obtained by
connecting the ending control points of each \Bezier curve, and the
corresponding quad meshing result is shown in
  (c).   (d) shows the
domain partition results with the construction of segmentation curves
by global optimization method. After local optimization process for
each sub-patch,  the final planar parameterization result is illustrated in
  (e) with the iso-parametric curves; (f)
presents the scaled Jacobian colormap of the parameterization to
illustrate the quality of planar parameterization.}
\label{fig:example1}
\end{figure*}

\subsection{Framework overview}

The problem investigated in this paper can be stated as follows:
Given a planar region bounded by a set of B-spline curves, construct
an analysis-suitable planar B-spline parameterization. In other words: How can several B-spline
patches be constructued to fill the computational domain bounded by B-spline curves? In order to address the above problem,  a general framework is
proposed  which consists of the following four steps
illustrated in Fig. \ref{fig:example1}:


\begin{enumerate}
\item \emph{Pre-processing for high-quality parameterization} (Section
  \ref{sec:preprocessing}):  In order to generate a high-quality planar
  parameterization, several pre-processing operations, including
  \Bezier extraction and \Bezier subdivision are performed on
  boundary curves (see Fig. \ref{fig:example1}(a)) ;

\item \emph{Topology information generation of quadrilateral decomposition} (Section \ref{sec:topology}):
From the discrete boundary obtained by connecting the ending control points of each \Bezier
curve as shown in Fig. \ref{fig:example1}(b), approximate convex decomposition and pattern based quad mesh generation are employed to generate the topology information of the
quadrilateral decomposition of the computational
domain;

\item \emph{Construction of the quadrilateral patch partition by global
    optimization } (Section \ref{sec:topology1} and Section
  \ref{sec:topology2}): After the topology partition is obtained, the Laplacian smoothing method
is used to improve the quad mesh quality ( Fig. \ref{fig:example1}(c)); then the optimal geometry of the interior B-spline curves corresponding to the interior edges of the
quadrangulation are obtained by a global optimization method to achieve a high-quality
patch-partition  ( Fig. \ref{fig:example1}(d)).

\item \emph{High-quality patch parameterization by local optimization} (Section \ref{sec:local}):  The \Bezier  patch with respect to each quad in the
quadrangulation is obtained by local optimization   to
achieve uniform and orthogonal iso-parametric structures while keeping
the continuity conditions between patches; an approach for detecting and recovering of
invalid \Bezier patches is also proposed to guarantee the injectivity
of the resulting parameterization ( Fig. \ref{fig:example1}(e) and Fig. \ref{fig:example1}(f)).
\end{enumerate}

\section{Pre-processing of input boundary curves for high-quality parameterization}
\label{sec:preprocessing}

In order to achieve a high-quality parameterization of the computational
domain, special treatment of input boundary curves are needed.

In IGA, the standard computational element is the
sub-patch corresponding to the knot interval in the definition of the
B-spline planar surface. In our parameterization framework,
the input boundary B-spline curves will be segmented
firstly with the \Bezier extraction technique \cite{Borden,Farin:CAGD}, in
which  the
piecewise B-spline representation is converted into \Bezier form. The B-spline basis defined on a knot vector
can be written as a linear combination of the
Bernstein polynomials, that is,
\begin{equation}\label{eqn:extraction1}
\mathbf {N}(\mathbf{t}) = \mathbf{C} \mathbf{B}(\mathbf{t})
\end{equation}
where $\mathbf{C}$ denotes the \Bezier extraction operator and $ \mathbf{B}(\bf
t)$ are the Bernstein polynomials which are defined on $[0, 1]$.
The conversion matrix $\mathbf{C}$ is sparse and its entries can be
obtained by multiple knot insertion, which can be performed by Boehm's
algorithm. Details on the \Bezier extraction can be found in \cite{Borden,Farin:CAGD}.

With the conversion matrix $\mathbf{C}$, the \Bezier extraction of the B-spline
curves can be represented by
\begin{equation}
\mathbf{P} = \mathbf{C} \mathbf{Q}
\end{equation}
$\mathbf{Q}$ denoting the control points of the B-spline curve, and
$\mathbf{P}$ are the control points of the extracted \Bezier curve.

After \Bezier extraction, most of the resulting \Bezier
curves have an ideal shape for the parameterization. However, in some
cases, there might be some \Bezier curves with complex shapes,
which require further division. In order to obtain
satisfactory curve segments, we will
perform the \Bezier subdivision process based on the distance
estimation between \Bezier curves and its control
polygon \cite{Wang84,PeterCAGD1999}. Let  $L_{ave}$ be the average
length of the lines connecting the starting control points and ending control
points of all the extracted \Bezier curves from the input boundary.
The corresponding termination criterion for this \Bezier subdivision process can be described as follows:

\noindent \textbf{Termination criterion.} The maximal distance $D^{k}_{max}$  between the \Bezier curve
$S_k(t)$ and the straight line connecting its starting
control point $(s_{0,k}^x,s_{0,k}^y) $ and
ending control point $(s_{n,k}^x,s_{n,k}^y) $ is smaller than $L_{ave}$ .

It should be mentioned that $L_{ave}$ is determined by the input
boundary  and is kept unchanged during the subdivision process. From
the above termination criterion, the \Bezier curve will be subdivided into two \Bezier segments at $t=0.5$ if
 $D^{k}_{max}$  is larger than $L_{ave}$. If $(s_{i,k}^x,s_{i,k}^y) $, $i = 0,
 ..., n$ and
\[
\eta =\mathop{\max}\limits_{0\leq i \leq n-2} \{|s_{i,k}^x-2
s_{i+1,k}^x+s_{i+2,k}^x|, |s_{i,k}^y-2 s_{i+1,k}^y+s_{i+2,k}^y|\},
\]
then
\begin{equation}
\Gamma \geq \log_4\frac{\sqrt{3}n(n-1) \eta }{8 L_{ave} },
\end{equation}
where $\Gamma$ is the number of times the \Bezier curve must be subdivided in order to satisfy the termination
criterion  \cite{Wang84,PeterCAGD1999}. Fig. \ref{fig:curvesubdivision} shows an example of the \Bezier
subdivision, in which a \Bezier curve with concave shape is
subdivided into two \Bezier segments for the subsequent high-quality parameterization
process.

\begin{figure}
\centering
\begin{minipage}[t]{2.5in}
\centering
\includegraphics[width=2.2in]{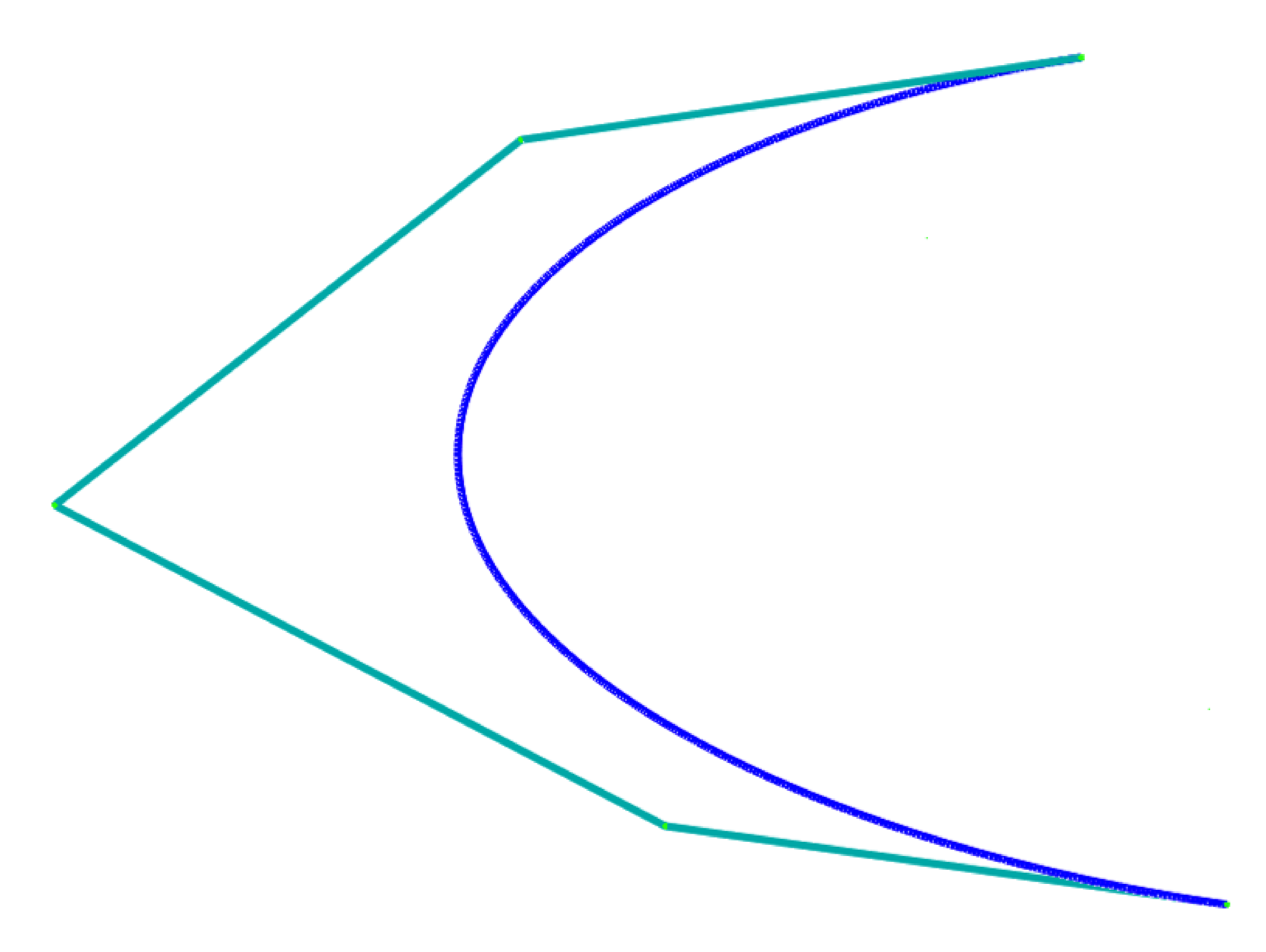}
\\ (a) original \Bezier curve
\end{minipage}
\begin{minipage}[t]{2.5in}
\centering
\includegraphics[width=2.3in]{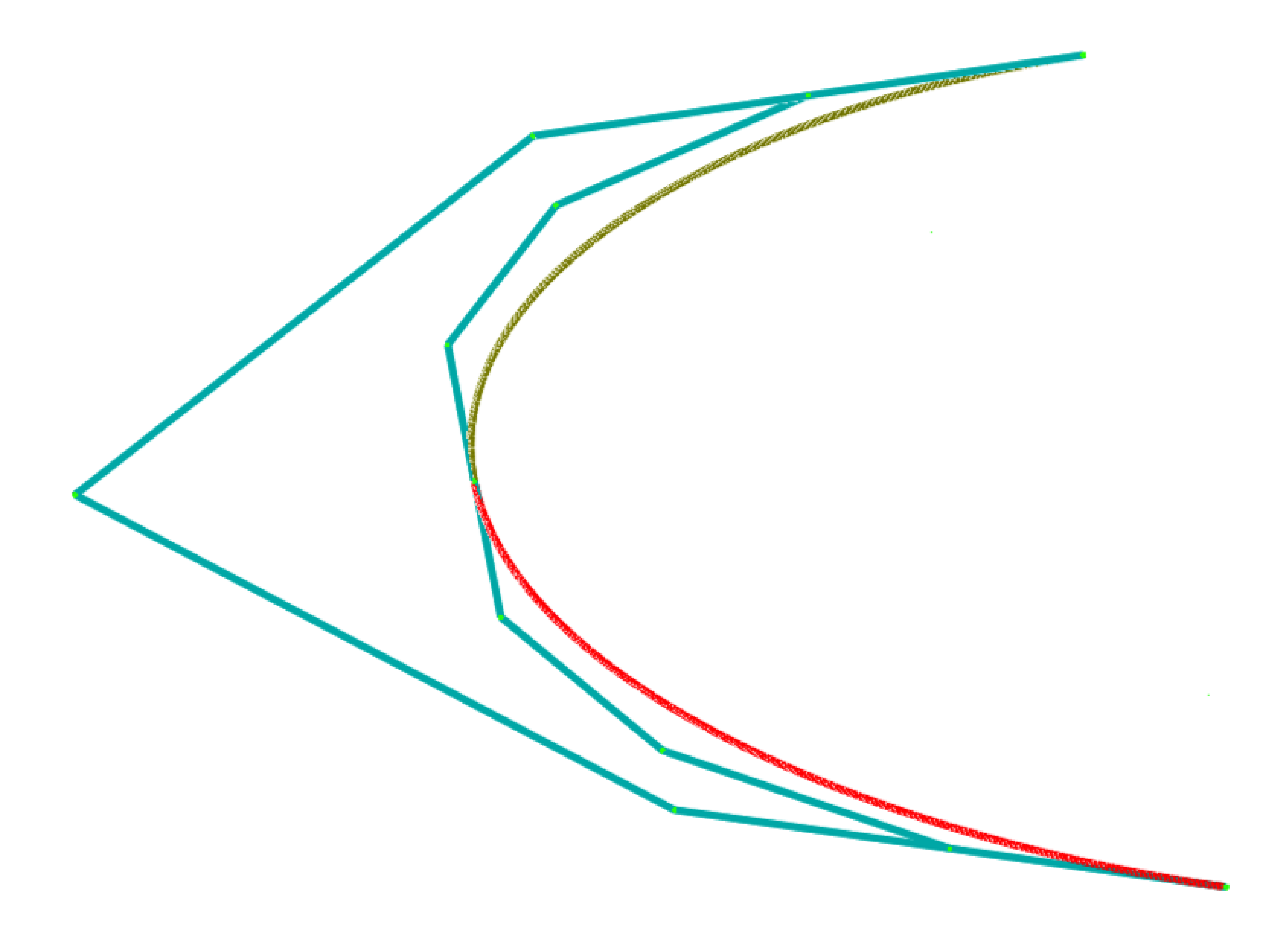}
\\ (b) \Bezier subdivision
\end{minipage}\\
\caption{Subdivision of a \Bezier curve with concave shape. }
\label{fig:curvesubdivision}
\end{figure}

\section{Four-sided partition of the computational domain by global
  optimization}

\label{sec:partition}

After the pre-processing operation for the given boundary curves,  we will now propose a global optimization method to construct the four-sided curved partition of the computational
domain. Three steps are required to address this
problem: Firstly, we generate the topology information of the quadrilateral
decompositions. Afterwards, the Laplacian mesh smoothing method is employed to
improve the quad mesh quality. Finally, the optimal shape of the interior
B-spline curves corresponding to the interior edges of the quadrangulation are obtained by a global optimization method.

\subsection{Topology generation of quadrilateral decomposition}
\label{sec:topology}

\begin{figure*}[!htb]
\centering
\begin{minipage}[t]{3.2in}
\centering
\includegraphics[width=2.3in]{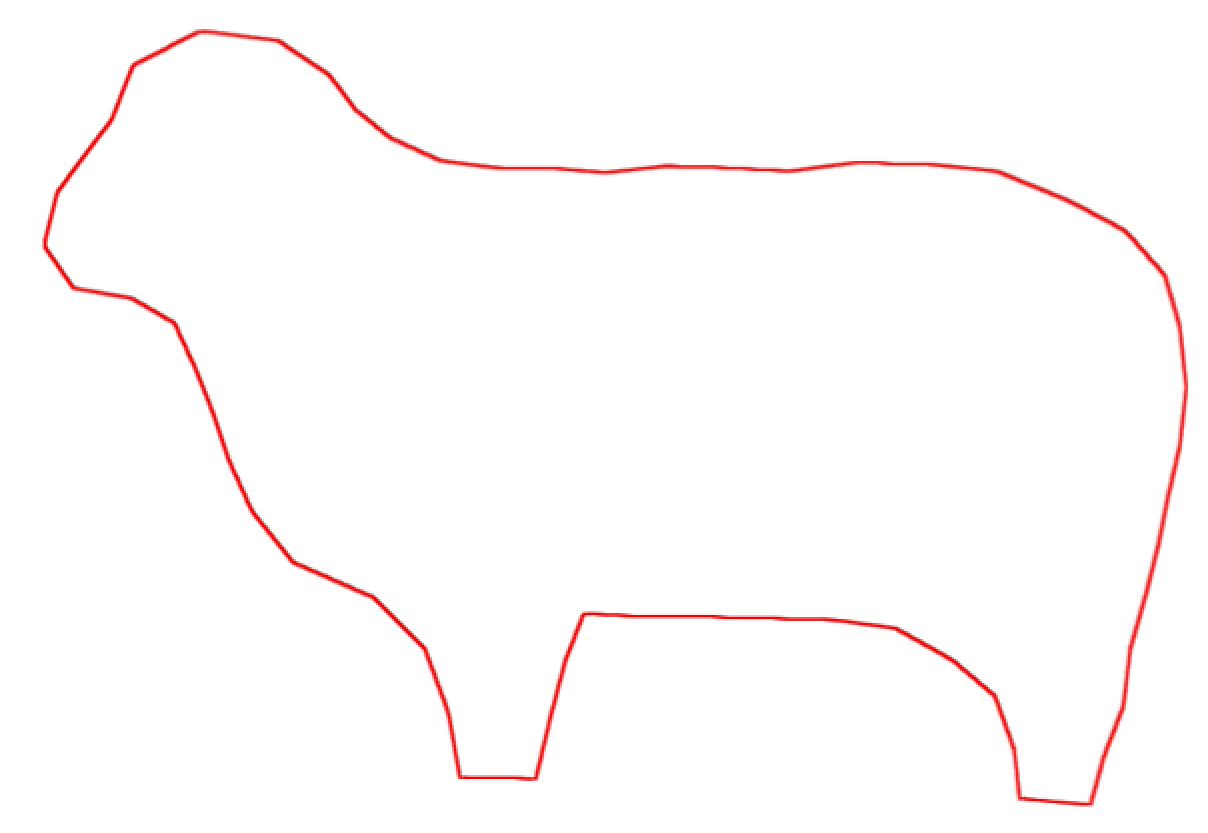}
\\ (a) input discrete boundary
\end{minipage}
\begin{minipage}[t]{3.2in}
\centering
\includegraphics[width=2.3in]{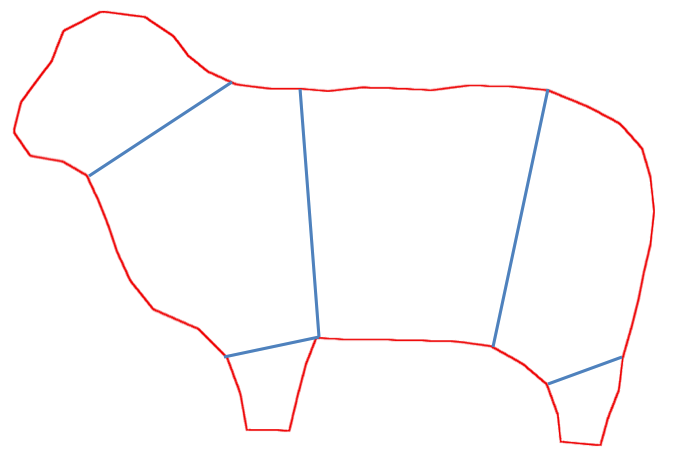}
\\ (b) quasi-convex polygon decomposition
\end{minipage}\\
\begin{minipage}[t]{3.2in}
\centering
\includegraphics[width=2.3in]{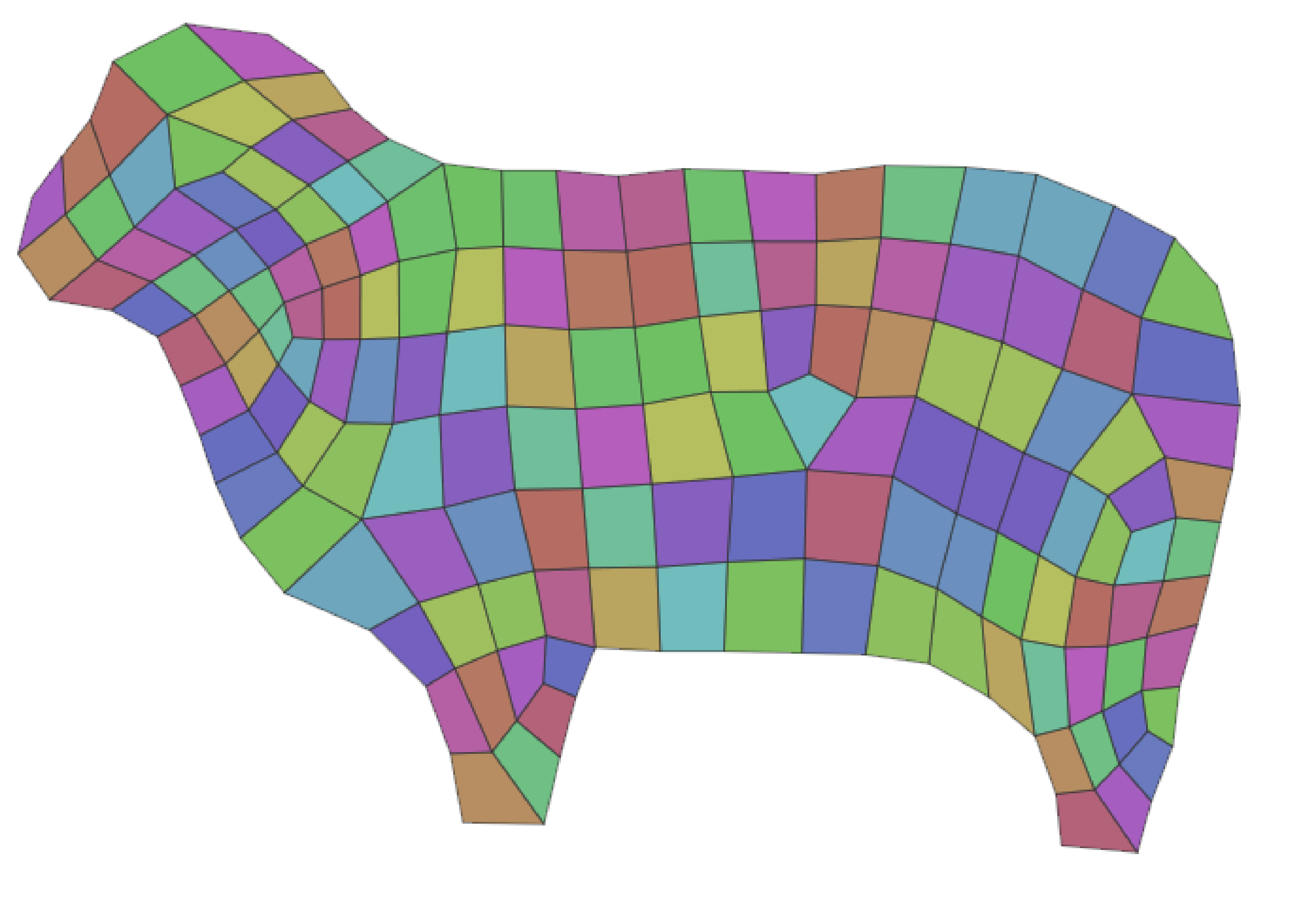}
\\ (c) quad-meshing result by our method with $147$ elements and $14$ irregular
vertices
\end{minipage}
\begin{minipage}[t]{3.2in}
\centering
\includegraphics[width=2.27in]{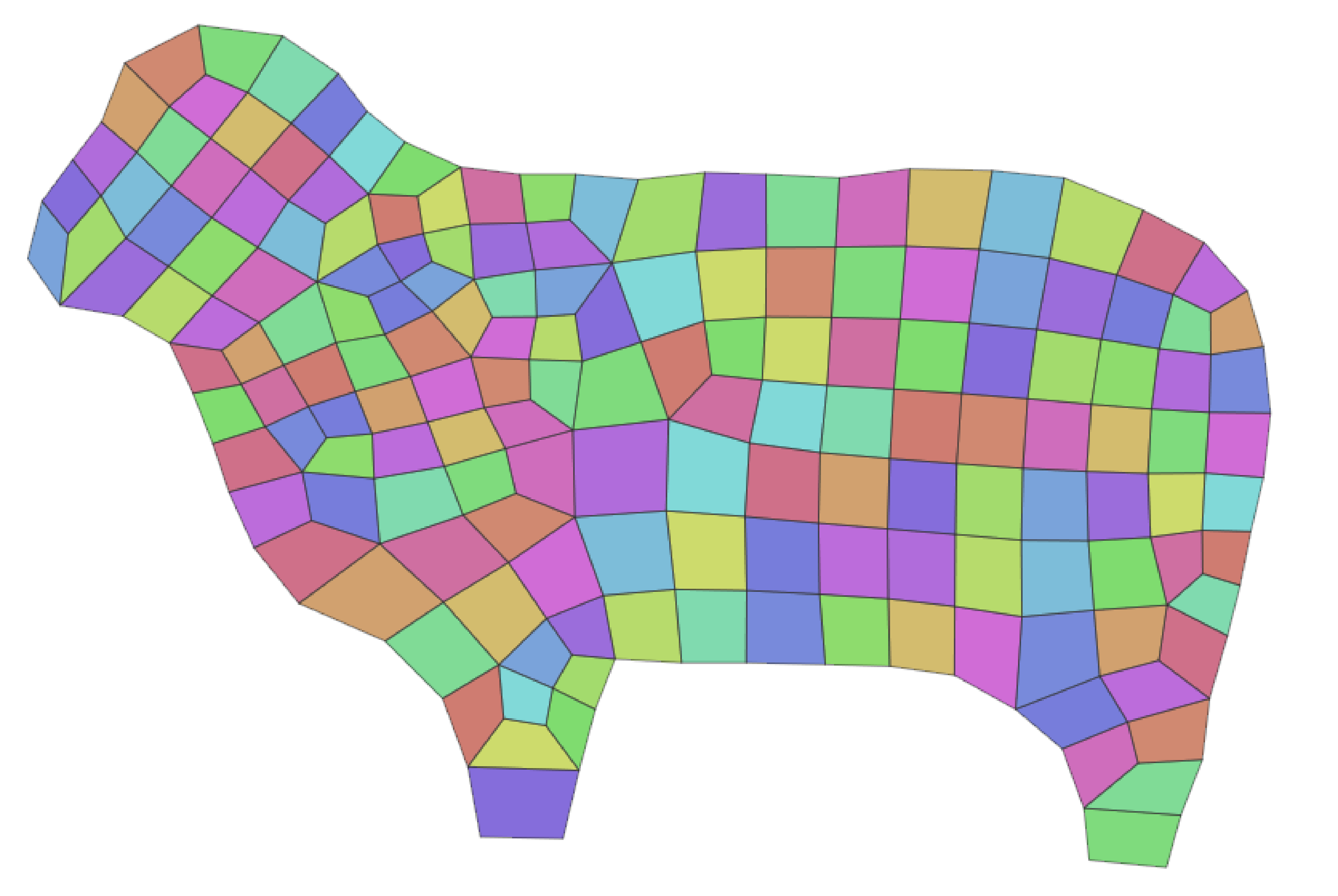}
\\ (d) meshing result by the method in \cite{ma:automeshSDU} with $160$ elements and $45$ irregular
vertices
\end{minipage}
\caption{Quad-meshing example. }
\label{fig:quadgeneration}
\end{figure*}

In this subsection, we will describe how to generate the topology
information of quadrilateral decompositions from the boundary curves of the
computational domain. The main steps include:
\begin{enumerate}[{Step}~1:]
\item Construct the  discrete boundary by connecting the endpoints of the
 extracted \Bezier curves obtained in Section
 \ref{sec:preprocessing} (see Fig. \ref{fig:quadgeneration}(a));
\item Convert the multiply-connected regions into simply-connected
  regions as presented in Appendix I (see Fig. \ref{fig:multiplyConnect});
\item Approximate the convex decomposition of the simply-connected regions
  by using the approach proposed in \cite{acd} (see Fig. \ref{fig:quadgeneration}(b));
\item For each quasi-convex polygon obtained in Step. 3,  generate the
  quadrangulation topology information by the patterns proposed in
  \cite{pattern} (see Fig. \ref{fig:quadgeneration}(c)).
\end{enumerate}

The proposed method in \cite{pattern} can produce patterns with minimal number of
irregular vertices, which are constructed  by solving a set of small
integer linear programs. As shown in
Fig. \ref{fig:quadgeneration},  our framework can generate
high-quality quad-meshing results with fewer irregular vertices
compared to the method in \cite{ma:automeshSDU}. Furthermore, the patterns proposed in
\cite{pattern} only introduce irregular vertices with valence $3$ or
$5$, which guarantees the solution existence for $G^1$
planar parameterization around the irregular vertex as shown in Section \ref{sec:g1}.

\begin{figure}
\centering
\begin{minipage}[t]{2.4in}
\centering
\includegraphics[width=2.in]{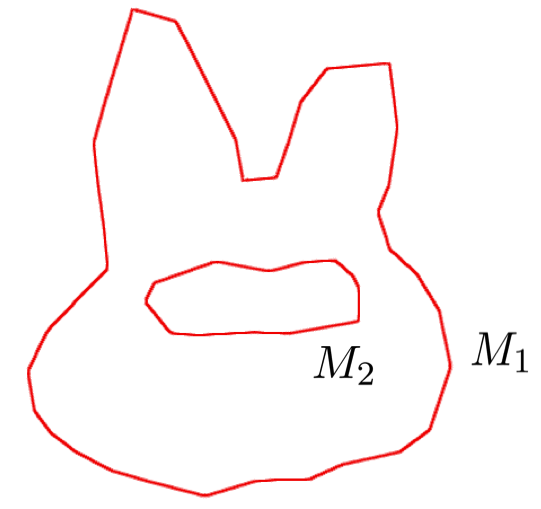}
\\ (a) multiply-connected domain
\end{minipage}
\begin{minipage}[t]{2.4in}
\centering
\includegraphics[width=2.in]{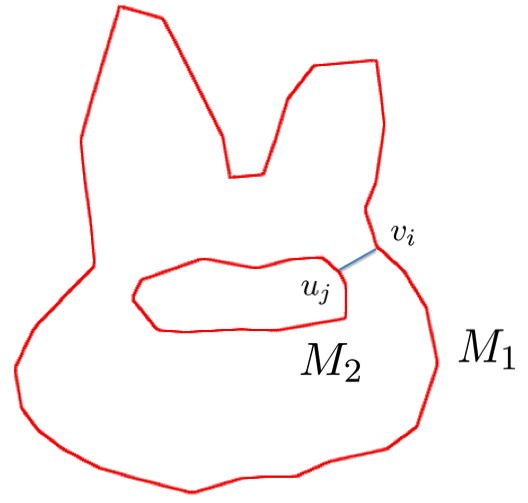}
\\ (b) simply-connected domain
\end{minipage}\\
\caption{Convert a multiply-connected domain into simply-connected
  domain. }
\label{fig:multiplyConnect}
\end{figure}

\subsection{Laplacian smoothing method}

\label{sec:topology1}

Since the quad-meshing quality significantly affects the
final parameterization results,  we adapt an iterative
Laplacian smoothing method to improve the quality of the quad mesh.
In every iteration, each internal mesh vertex is moved to the centroid of its
neighbor vertices, i.e.
\begin{equation}
x_i^k=\frac{\displaystyle{\sum_{j=1}^{N_i}}x_{j}^{k-1}}{N_i}, \qquad  y_i^k=\frac{\displaystyle{\sum_{j=1}^{N_i}}y_{j}^{k-1}}{N_i}
\end{equation}
in which $N_i$ is the number of neighbor vertices of the internal mesh
vertex with location $(x,y)$, and the superscript $k$ is the iteration
counter.  This iteration is terminated according to the following
termination rules:
\begin{equation}
\frac{[\displaystyle{\sum_{i=1}^m}[(x_i^k-x_i^{k-1})^2+(y_i^k-y_i^{k-1})^2]]^{1/2}}
{[\displaystyle{\sum_{i=1}^m}[(x_i^{k-1})^2+(y_i^{k-1})^2]]^{1/2}} <
\delta,
\end{equation}
where $m$ is the total number of mesh vertices, and $\delta$ is the
specified tolerance value (0.001).

By the Laplacian smoothing method, we can improve the quad-mesh
quality towards equilibrating the element size
globally as shown in Fig. \ref{fig:quadgeneration}(c).

\subsection{Construction of segmentation curves between patches }

\label{sec:topology2}

After constructing the quad mesh $Q(V,E)$ of the discrete computational domain,
we construct the segmentation curves between patches
corresponding to each quad. The segmentation curves interpolate
two vertices on the quad mesh $Q(V,E)$ as shown in Fig. \ref{fig:segmentation}(a). Furthermore, for a planar computational domain, the shape of the segmentation
curves has a great effect on the uniformity of the patch-size. Hence, we propose a global optimization method
to construct the optimal shape of the segmentation curves.

\begin{figure}
\centering
\begin{minipage}[t]{2.5in}
\centering
\includegraphics[width=2.3in]{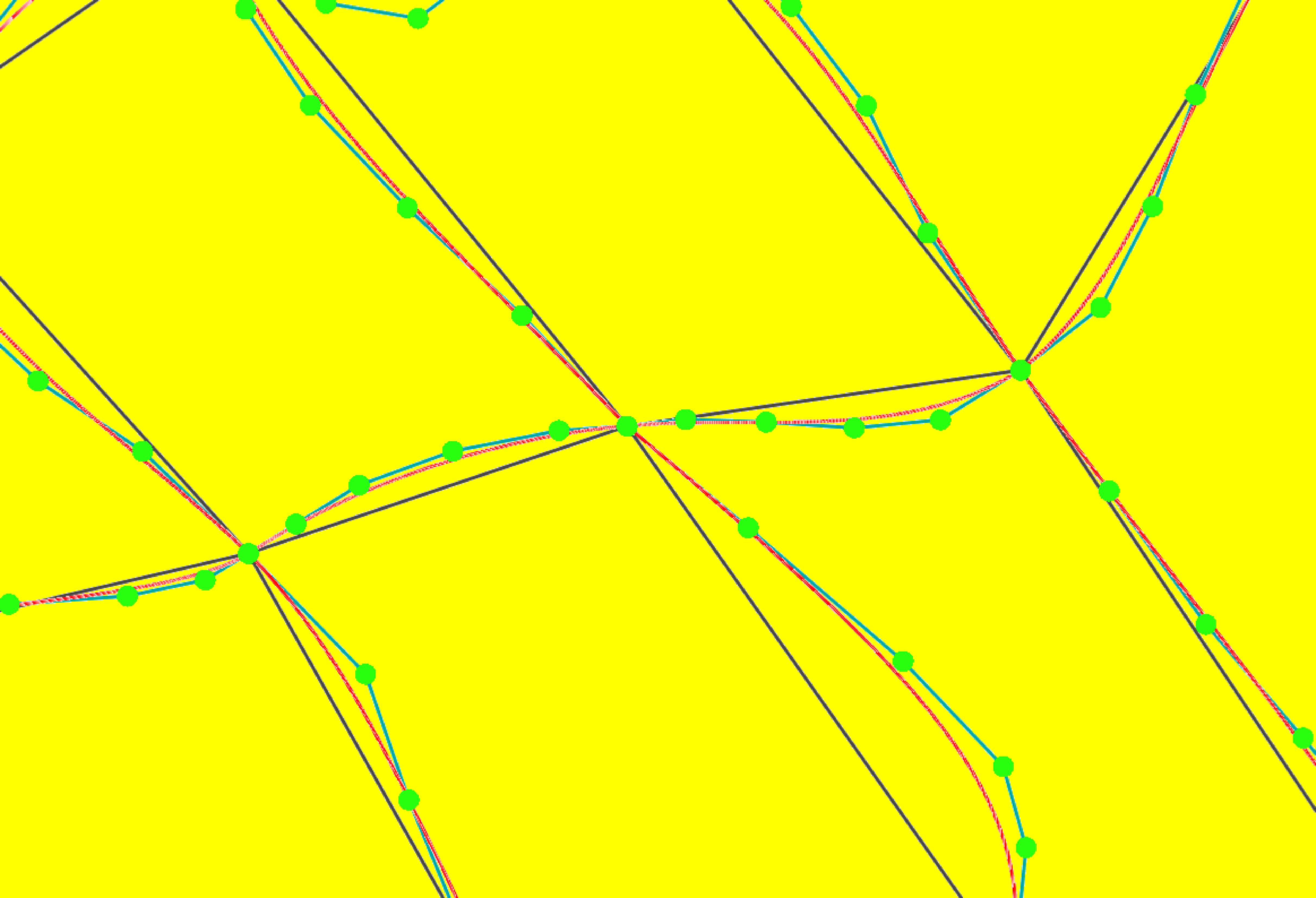}
\\ (a) segmentation curves(red) w.r.t quad edges( black)
\end{minipage}
\begin{minipage}[t]{2.5in}
\centering
\includegraphics[width=2.2in]{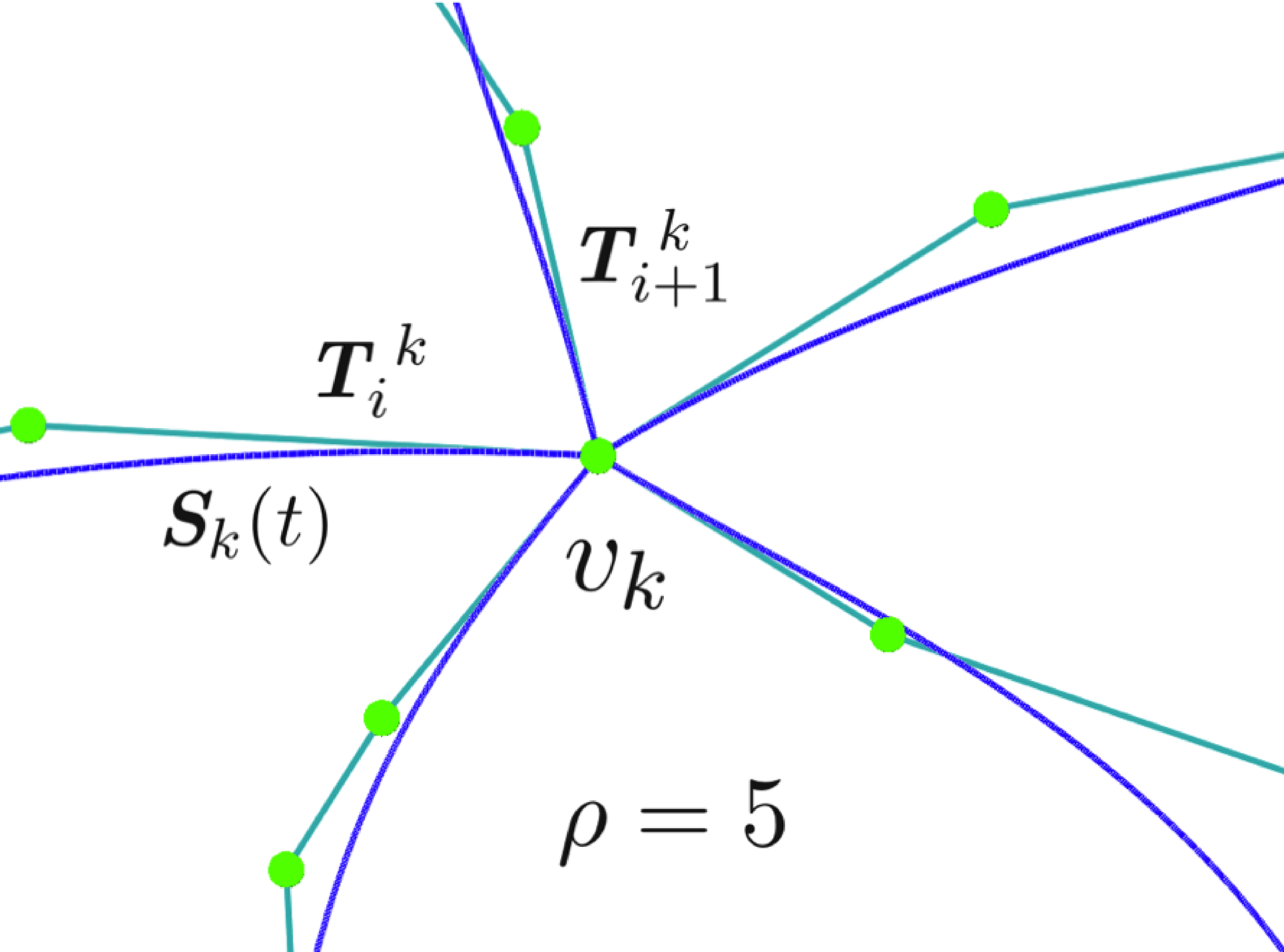}
\\ (b) segmentation curve at a singular vertex with $\rho=5$
\end{minipage}\\
\caption{Construction of segmentation curves. }
\label{fig:segmentation}
\end{figure}

There are three objective functions: The first term is related to the uniformity of the patch-size; the second term is related to the shape quality of the
segmentation curves; the third term is related to the tangent
constraints of the segmentation curves at quad-mesh vertices.

For the uniformity of the patch-size, we
propose a new metric related to the patch-area. It is obvious that the
area of a planar region only depends on the shape of the boundary
curves. In fact, for computing the area of the planar region bounded by \Bezier curves, we have the following
proposition.

\begin{proposition}
For the planar region $\Omega$ bounded by $N$ \Bezier curves
\[
\textit{\textbf{S}}_k(t)=(S_k^x(t),S_k^y(t))=\displaystyle{\sum_{i=1}^n} (s_{i,k}^x,s_{i,k}^y)
B_i^n(t),
\]
in which $(s_{i,k}^x,s_{i,k}^y) $ is the control point, $t\in
[0,1]$, $k=1,\cdots,N$, the area $A(\Omega)$ of the planar region $\Omega$ is the sum of
integrals as follows:
\begin{equation}\label{eqn:Ai}
A(\Omega)=  \frac{1}{4n} ~ \displaystyle{\sum_{k=1}^N}\displaystyle{\sum_{j=0}^{2n-1}} (c_j^k-d_j^k),
\end{equation}
in which
\begin{equation}\label{eqn:cj}
c_j^k= \sum_{r=\text{max}(0,j-n)}^{\text{min}(j,n-1)} \frac{{n \choose
    r}{n-1 \choose j-r}}{{2n-1 \choose j}}s_{r,k}^x \dot (s_{j-r+1,k}^y-s_{j-r,k}^y)
\end{equation}
\begin{equation}\label{eqn:dj}
d_j^k= \sum_{r=\text{max}(0,j-n)}^{\text{min}(j,n-1)} \frac{{n \choose
    r}{n-1 \choose j-r}}{{2n-1 \choose j}}s_{r,k}^y \dot (s_{j-r+1,k}^x-s_{j-r,k}^x)
\end{equation}\label{prop:area}
\end{proposition}

Proposition \ref{prop:area} can be proved by Green's formula \cite{Carmo:1976} and the
properties of Bernstein polynomials, which are described in subsection
\ref{subsec:prebernstein}.

From Proposition \ref{prop:area},  we define the uniformity
metric for the \Bezier patches constructed from the quad mesh $Q$.
The uniform patch structure  over the quad mesh $Q$ means that each
patch defined on each quad element in $Q$ has the same area. In probability theory
and statistics, the \emph{variance} measures how far a set of numbers is
spread out. A variance of zero indicates that all the values are
identical. Hence, the uniformity of the
patch structure requires the variance between each patch to be as small as possible. Suppose that $A_i$ is the area of the $i$-th
patch bounded by four \Bezier curves,  and $A_{ave}$ is the average patch area in the patch structure $\Omega$ over
the quad mesh $Q$, i.e.
\[
A_{ave}=\frac{A(\Omega)}{L},
\]
then the uniformity metric $F_{uniform}$ is defined as the variance of
$A_i$,
\begin{equation}\label{eq:disvar}
F_{\text{uniform}} = \frac{1}{L}\sum_{i=0}^L {(A_i-A_{ave})^2}
\end{equation}
in which $L$ is the number of \Bezier patches in the patch structure $\Omega$,
$A(\Omega)$ is the area of the patch structure $\Omega$ bounded by the
given set of B-spline curves. $A_i$ and $A(\Omega)$ can be computed
according to the area formula presented in Proposition
\ref{prop:area}.

For the term related to the shape quality of the segmentation curves,
the corresponding objective function is defined as a combination
of the stretch energy and the strain energy,
\begin{equation}\label{eq:shape}
F_{\text{shape}}= \sum_{k=0}^{N}\int_0^1 \sigma_1
\|\textit{\textbf{S}}_k^{~'}(t)\|^2 + \sigma_2
\|\textit{\textbf{S}}_k^{~''}(t)\|^2  dt,
\end{equation}
in which $\sigma_1$ and $\sigma_2$ are positive
weights.   If $\sigma_1 >
 \sigma_2$, then the resulting segmentation curves have smaller
 stretch energy, which measures the length of a curve.
 If $\sigma_2 > \sigma_1$, we obtain
segmentation curves with smaller strain energy, which is a measures of  the curve's bending.

For the term related to the tangent vector requirements of segmentation curves
at the quad-mesh vertex $v_k$ with valence $\rho$, we define the
following objective function,
\begin{equation}\label{eq:tangent}
F_{\text{tangent}}= \sum_{k=0}^{N} \sum_{i=1}^{\rho}
(\frac{\textit{\textbf{T}}_i^{~k} \cdot
  \textit{\textbf{T}}_{i+1}^{~k}}{\| \textit{\textbf{T}}_i^{~k}\|
  \|\textit{\textbf{T}}_{i+1}^ {~k}\|}-\cos {\frac{2\pi}{\rho}})^2
\end{equation}
where $\textit{\textbf{T}}_{i}^ {~k}$ is the tangent vector of the
segmentation curve $\textit{\textbf{S}}_k (t)$ at the quad-mesh vertex
$v_k$, and   $\textit{\textbf{T}}_{\rho+1}^ {~k} =
\textit{\textbf{T}}_{1}^ {~k}$ as shown in Fig. \ref{fig:segmentation}(b). At regular
vertices with valence $\rho = 4$, the minimization of
$F_{\text{tangent}}$ will achieve quasi-$C^1$ and quasi-orthogonal segmentation curves,
which is a basic requirement for analysis-suitable parameterization.

Combining the optimization terms defined in (\ref{eq:disvar}), (\ref{eq:shape}) and
(\ref{eq:tangent}), we construct an objective function $F$ as follows,
\begin{equation}\label{eq:curveobj}
F =\omega_1 F_{\text{uniform}}+ \omega_2
F_{\text{shape}}+ \omega_3 F_{\text{tangent}},
\end{equation}
in which $\omega_1$, $\omega_2$ and $\omega_3$ are positive
weights for the balance between the shape quality metric,
uniformity metric and the metric for tangent constraints. In order to
achieve quasi-$C^1$ and quasi-orthogonal segmentation curves,
$\omega_3$ is usually much larger than  $\omega_1$ and
$\omega_2$ .

The segmentation curves are obtained by solving the non-linear optimization problem
\begin{equation}\label{eq:curveopt}
\underset{\textit{\textbf{s}}_{i,k}}{\operatorname{arg\,min}}
\quad  F \quad,
\end{equation}
in which the design variables are the control points of the
segmentation curves $\textit{\textbf{S}}_k (t)$.

The scale of the optimization problem in \eqref{eq:curveopt} depends
on the number and the degree of the \Bezier segmentation curves.
The L-BFGS method is adopted to obtain the optimal solution,
which is a quasi-Newton method to solve unconstrained nonlinear minimization
problems. In the L-BFGS method, we approximate the inverse Hessian
matrix of the objective function in \eqref{eq:curveopt} by a sequence of gradient
vectors from previous iterations. For more details, the reader can refer to
\cite{LBFGS}.

\begin{figure}
\centering
\begin{minipage}[t]{3.in}
\centering
\includegraphics[width=3.in]{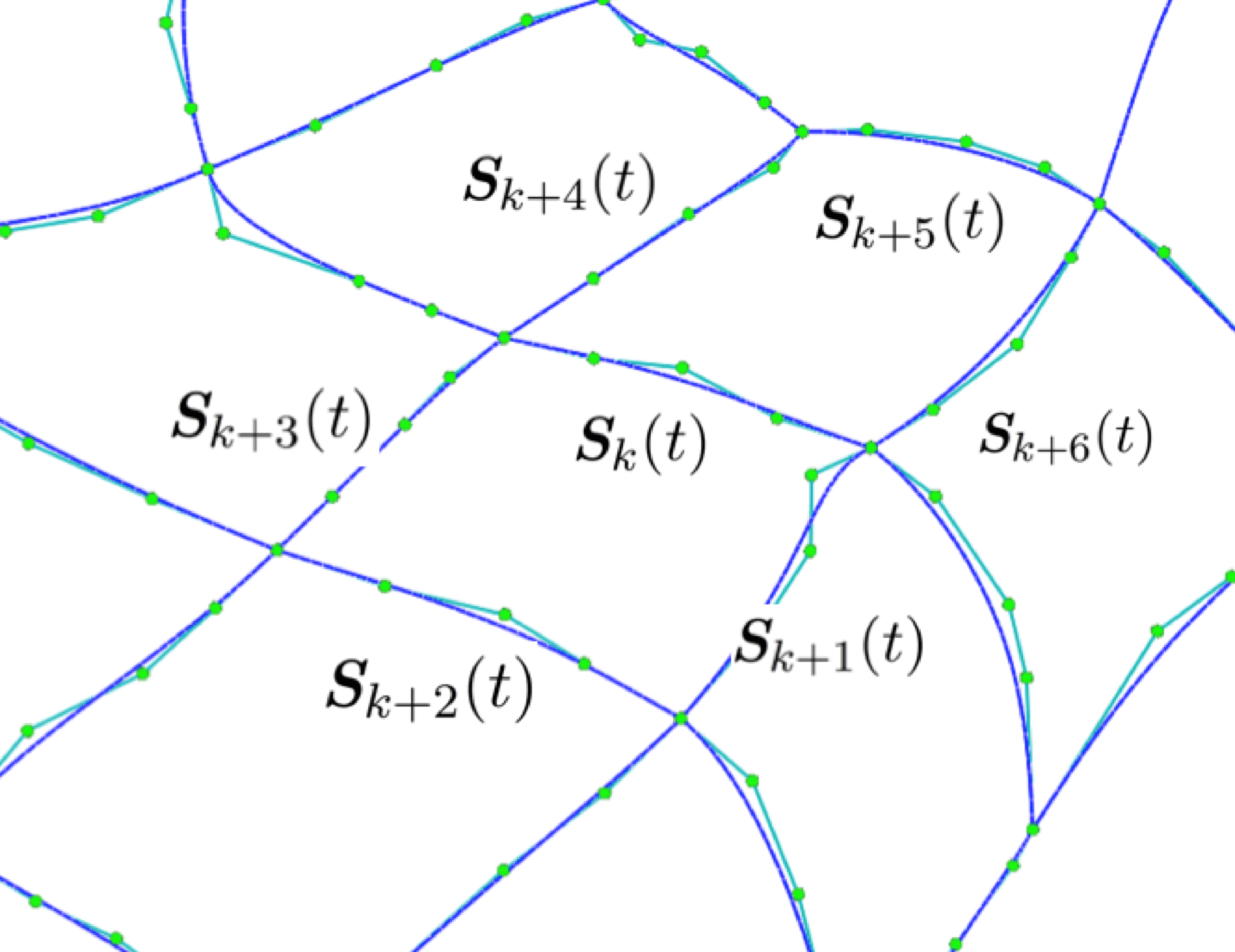}
\end{minipage}\\
\caption{Interior segmentation curve $\textit{\textbf{S}}_k(t)$ and its neighboring segmentation curves. }
\label{fig:constraintSeg}
\end{figure}

\newpage

\noindent \textbf{Remark 5.1. }  The interior segmentation curves are called \emph{valid} if there are no intersections between them except the connecting control points. That is, for each interior segmentation curve $\textit{\textbf{S}}_k(t)$, $k=1,\cdots,N$, it has no intersections with its neighboring segmentation curves $\textit{\textbf{S}}_{k+1}(t)$, $\textit{\textbf{S}}_{k+2}(t)$, $\textit{\textbf{S}}_{k+3}(t)$, $\textit{\textbf{S}}_{k+4}(t),\textit{\textbf{S}}_{k+5}(t)$ and $\textit{\textbf{S}}_{k+6}(t)$ as shown in Fig. \ref{fig:constraintSeg}. In particular, we can derive the following sufficient conditions such that $\textit{\textbf{S}}_k(t)$ has no intersections with $\textit{\textbf{S}}_{k+j}(t)$:
$$
\triangle \textit{\textbf{s}}_{i-1}^k\textit{\textbf{s}}_{i}^k\textit{\textbf{s}}_{i+1}^k \cap \triangle \textit{\textbf{s}}_{l-1}^{k+j}\textit{\textbf{s}}_{l}^{k+j}\textit{\textbf{s}}_{l+1}^{k+j} =\emptyset, \quad i=1,\cdots,n, \quad l=1,\cdots,n, \quad j=1,\cdots,6.
$$
in which $\textit{\textbf{s}}_{i}^k$ is the control point of segmentation curve $\textit{\textbf{S}}_k(t)$. That is, if the triangles $\triangle \textit{\textbf{s}}_{i-1}^k\textit{\textbf{s}}_{i}^k\textit{\textbf{s}}_{i+1}^k $ formed by the control points of $\textit{\textbf{S}}_k(t)$ do not intersect with the other triangles  $\triangle \textit{\textbf{s}}_{l-1}^{k+j}\textit{\textbf{s}}_{l}^{k+j}\textit{\textbf{s}}_{l+1}^{k+j} $ formed by the control points $\textit{\textbf{S}}_{k+j}(t)$, then
$\textit{\textbf{S}}_k(t)$ has no intersections with $\textit{\textbf{S}}_{k+j}(t)$. These sufficient conditions can be integrated with the optimization problem (\ref{eq:curveopt}) to generate valid interior segmentation curves.

\section{High-quality patch parameterization by local constrained optimization }{}

\label{sec:local}

After the four-sided patch partition is generated, we next
construct the inner control points for each \Bezier patch. One basic
requirement is that the resulting domain parameterization
satisfies the continuity constraints.  In this paper, we propose an
efficient patch-wise local optimization method to construct
high-quality patch parameterizations with continuity constraints. The main framework for each local patch construction is summarized as follows:
\begin{enumerate}[{Step}~1:]
\item Construct the boundary control points on the second layer of the
  control mesh for each patch by
  orthogonality optimization and impose continuity constraints to
  achieve a near-orthogonal iso-parametric structure at common
  segmentation curves as described in subsection \ref{sec:twolayer};
\item Construct interior $(n-3)\times(n-3)$ control points for each patch by solving a
  linear system related to efficient  $C^1$ energy-minimizing scheme as presented in subsection \ref{sec:energy};
\item Find the invalid patches on the parameterization and
  recover the patch validity by repositioning  the control points
by a local optimization method as described in subsection
\ref{sec:valid}.
\end{enumerate}

\subsection{Construction of boundary second-layer control points with
  orthogonality optimization and continuity constraints}
\label{sec:twolayer}

In this section, we will construct the boundary control
points on the second layer of the control mesh, i.e ,
$\textit{\textbf{P}}_{n-1,j}$, $\textit{\textbf{P}}_{1,j}$,
$\textit{\textbf{P}}_{i,1}$, and $\textit{\textbf{P}}_{i,n-1}$,  to
satisfy the orthogonality and continuity requirements.  The construction procedures involves:
\begin{enumerate}[{Step}~1:]
\item \emph{Initial construction}: Construct the initial
  control points on the second layer of the
  control mesh for each patch by orthogonality optimization as
  described in subsection \ref{subsection:orthoptimize};
\item \emph{Imposition of $C^1$-continuity}: In the regular region, adjust the boundary control points on the
  second layer of the control mesh for each patch to achieve
  $C^1$-continuity. Therefore, the Lagrange-multiplier method  is employed as proposed in subsection \ref{subsection:c1optimize} ;
\item \emph{Imposition of $G^1$-continuity}: In the irregular region, adjust
the 1-neighbor control points around the irregular vertex to satisfy the general $G^1$-continuity
constraints by solving linear system as described in subsection \ref{sec:g1} .
\end{enumerate}

In the following subsections, some details will be described for each
step.

\subsubsection{Initial construction by orthogonality optimization  }
\label{subsection:orthoptimize}

Firstly, we will describe the initial construction by orthogonality
optimization. Without loss of generality, we firstly construct the initial position
$\textit{\textbf{P}}_{n-1,j}^{0}$ of $\textit{\textbf{P}}_{n-1,j}$ for
the \Bezier patch $\emph{\textbf{r}}(u,v)$:
\[
\textit{\textbf{P}}_{n-1,j}^{0}= \textit{\textbf{P}}_{n,j} + \frac{(\textit{\textbf{P}}_{0,j}-\textit{\textbf{P}}_{n,j})}{n},
\]
In order to achieve a near-orthogonal isoparametric structure on the
boundary $\emph{\textbf{r}}(1,v)$,  the final position of
$\textit{\textbf{P}}_{n-1,j}$ is obtained by solving the following optimization problem
\begin{equation}\label{eqn:orthopt}
\underset{\textit{\textbf{P}}_{n-1,j}}{\operatorname{arg\,min}}  \int_0^1 (< \emph{\textbf{r}}_{1,u}(1,v),\emph{\textbf{r}}_{1,v}(1,v)>)^2 dv
\end{equation}
in which
\begin{equation}
\emph{\textbf{r}}_{1,u}(1,v)=n\sum\limits_{j=0}^{n}{B_{l}^{n}(v)\Delta^{1,0}\textit{\textbf{P}}_{n-1,l}},
\end{equation}
\begin{equation}
\emph{\textbf{r}}_{1,v}(1,v)=n\sum\limits_{j=0}^{n-1}{B_{l}^{n-1}(v)\Delta^{0,1}\textit{\textbf{P}}_{n,l}},\label{eq:rv}
\end{equation}

Solving the similar optimization problems as
(\ref{eqn:orthopt}) yields the position of
$\textit{\textbf{P}}_{1,j}$,
$\textit{\textbf{P}}_{i,1}$ and $\textit{\textbf{P}}_{i,n-1}$ for each
patch, which requires adjustment to achieve $C^1$ continuity.

\subsubsection{$C^1$ construction in the regular region}

\label{subsection:c1optimize}

In order to ensure that the joint two \Bezier patches $\textit{\textbf{r}}^{k}_1
(u,v)$ and $\textit{\textbf{r}}^{k}_2 (u,v)$ in the regular region  satisfy $C^1$ continuity,
the control mesh near the common segmentation curve $\textit{\textbf{S}}_{k}(t)$   with control
points $\textit{\textbf{s}}_{j}^k$ should satisfy the following $C^1$ conditions:
\begin{equation}\label{c1}
\textit{\textbf{s}}_{j}^k-\textit{\textbf{P}}_{j}^k=\textit{\textbf{Q}}_{j}^k-\textit{\textbf{s}}_{j}^k,
\qquad k=0, \cdots, N, \quad  j= 0,\cdots, n
\end{equation}
in which $N$ is the number of segmentation curves,
$\textit{\textbf{P}}_j^k$ and $\textit{\textbf{Q}}_j^k$ are the
one-neighbor control points of the joint \Bezier patches along the
segmentation curve $\textit{\textbf{S}}_k(t)$ (see
Fig. \ref{fig:localoptimize}).

The problem to impose $C^1$-continuity can be
formulated as follows: Minimize the change of related control points
along the segmentation curves (except the one-neighbor control points
around the irregular vertex)  such that they satisfy the
$C^1$-constraints. In the following, we will solve this constrained
optimization problem using the Lagrange multiplier method.

In order to minimize the change of all the related control points along the
segmentation curves,
the optimization term is defined as
\begin{equation}\label{eq:optterm1}
\emph{\emph{Min}} \sum_{k=1}^{N}\sum_{j=0}^{n} (\|
\emph{\textbf{P}}_j^k -\bar{\emph{\textbf{P}}}_j^k\|^2+  \|
\emph{\textbf{Q}}_j^k -\bar{\emph{\textbf{Q}}}_j^k\|^2)
\end{equation}
in which $N$ is the number of segmentation curves,
$\bar{\emph{\textbf{P}}}_j^k$ and $\bar{\emph{\textbf{Q}}}_j^k$ are
the initial control points constructed by the approach in subsection
\ref{subsection:orthoptimize}.

Combining (\ref{c1}) with (\ref{eq:optterm1}), the
Lagrange function is defined as
\begin{equation}\label{eq:lag1}
L=  \sum_{i=0}^{N}\sum_{j=0}^{n}  (\|
\emph{\textbf{P}}_j^k -\bar{\emph{\textbf{P}}}_j^k\|^2+  \|
\emph{\textbf{Q}}_j^k -\bar{\emph{\textbf{Q}}}_j^k\|^2)+ \sum_{i=0}^{N}\sum_{j=0}^{n}
\lambda_{k,j}(2\textit{\textbf{s}}_{j}^k- \emph{\textbf{P}}_j^k-\textit{\textbf{Q}}_{j}^k )
\end{equation}
where
$\lambda_{k,j}=[\lambda_{k,j}^x,\lambda_{k,j}^y]^\emph{\emph{T}}$
are Lagrange multipliers.

The unknown variables in the constrained optimization problem
(\ref{eq:lag1}) are the control points
$\emph{\textbf{P}}_j^k=[p_{j}^{k,x},p_{j}^{k,y}]$
and
$\emph{\textbf{Q}}_j^k=[q_{j}^{k,x},q_{j}^{k,y}]$. A necessary condition for
 $\emph{\textbf{P}}_j^k$,   $\emph{\textbf{Q}}_j^k$ and
$\lambda_{k,j}$ to be a solution of (\ref{eq:lag1}) is that the
corresponding partial derivatives vanish, that is,
\begin{equation}
\label{c1linear}
\left\{\begin{array}{ll}
\frac{\D\partial L}{\D\partial
\lambda_{k,j}}=
2\textit{\textbf{s}}_{j}^k-\textit{\textbf{P}}_{j}^k-\textit{\textbf{Q}}_{j}^k=0,
\quad   k=0, \cdots, N, \quad  j= 0,\cdots, n;   \\
\frac{\D
\partial L}{\D
\partial p_{j}^{k,r}}= 0, \quad  r=x,y;\\
\frac{\D\partial L}{\D\partial q_{j}^{k,r}}= 0, \quad r=x,y
\end{array}
\right.
\end{equation}

Methods such as
Gauss elimination can be employed to solve the linear system (\ref{c1linear}).

The proposed least-square scheme leads to the boundary
control points on the second layer of control mesh for each \Bezier
patch, which satisfies the $C^1$-continuity
requirements.

 \begin{figure*}[t]
\centering
\begin{minipage}[t]{3.5in}
\centering
\includegraphics[width=3.5in,height=2.4in]{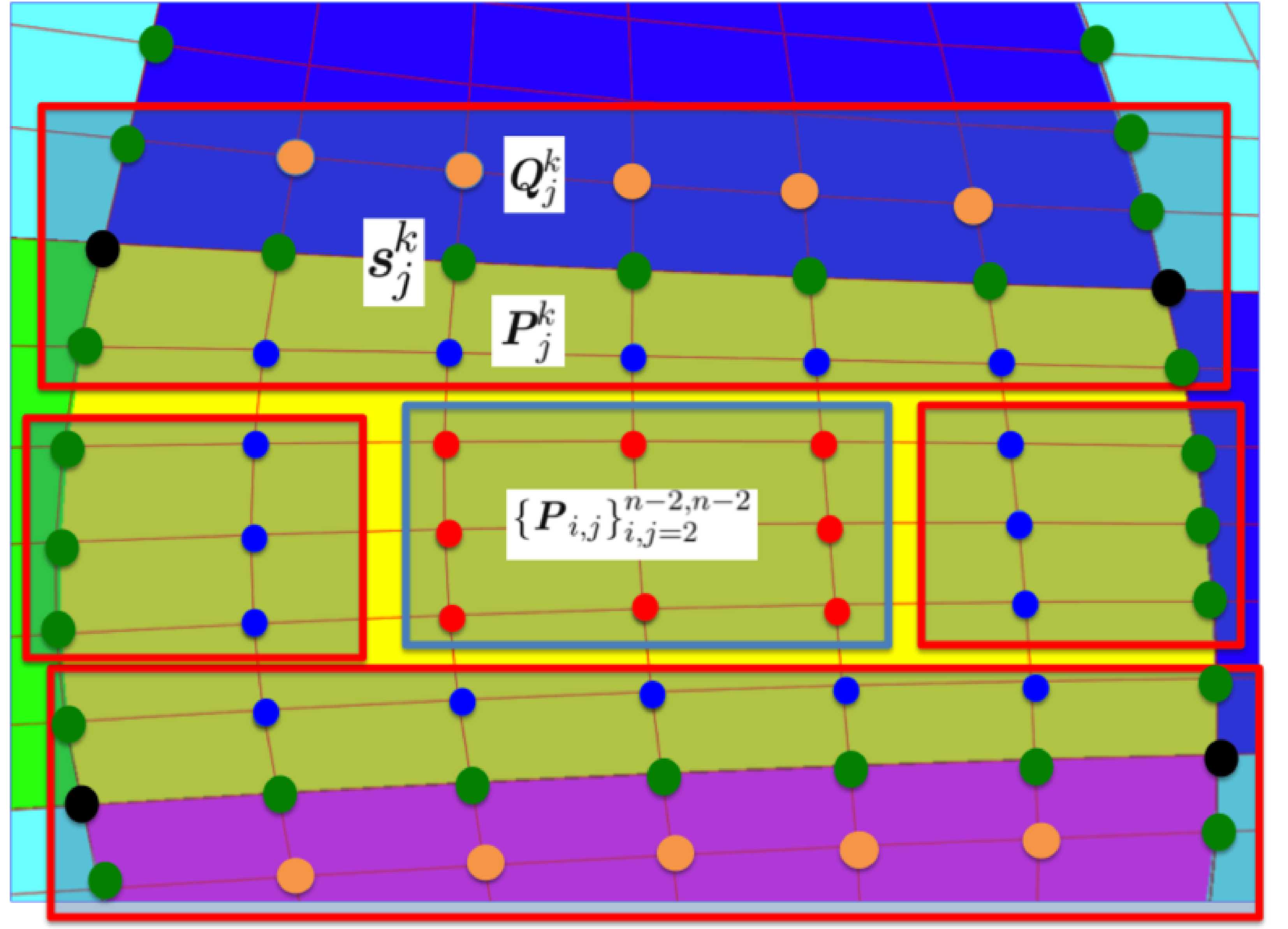}
\end{minipage}
 \\
\caption{Control points involved in the local optimization
  process: The blue  control points $\textit{\textbf{P}}_{j}^k$
  and the orange control points $\textit{\textbf{Q}}_{j}^k$ are
  determined by the $C^1$-continuity constraints with the green
  control points $\textit{\textbf{s}}_{j}^k$ of segmentation curves $\textit{\textbf{S}}_{k}(t)$ as
  proposed in subsection \ref{subsection:c1optimize} ;
   the red control points $\{\textbf{\emph{P}}_{i,j}\}_{i,j=2}^{n-2,n-2}$ are determined by the local
  linear-energy-minimizing method as described in subsection \ref{sec:energy}.  }
\label{fig:localoptimize}
\end{figure*}

 \begin{figure*}[t]
\centering
\begin{minipage}[t]{3.2in}
\centering
\includegraphics[width=3.2in ]{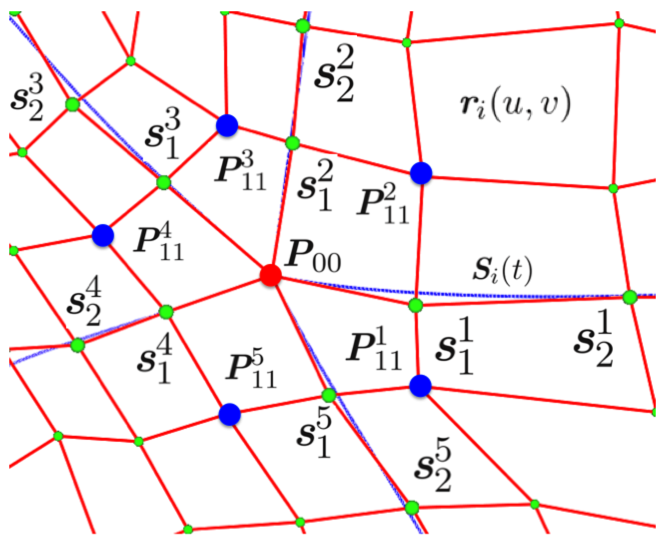}
\end{minipage}\\
\caption {Given the control points $\textit{\textbf{s}}^i_{j}$ of
  segmentation curves $\textit{\textbf{S}}_i(t)$,   construct 1-neighbor control points
  $\textit{\textbf{P}}_{11}^i$,
  around irregular vertex $\textit{\textbf{P}}_{00}$  to achieve
  $G^1$-continuity on planar parameterization, $i=1,2,\cdots,5$,\quad $j=0,1,\cdots, n $. }
\label{fig:localoptimize2}
\end{figure*}

\subsubsection{$G^1$ construction  around irregular vertex}

\label{sec:g1}

For the irregular vertices $\textit{\textbf{P}}_{00}$ with non-four-valence
on the quad mesh, some special treatments should be done in
order to achieve $G^1$ continuity at $\textit{\textbf{P}}_{00}$;  $G^1$
continuity means $C^1$-continuity of the parameterization maps through
the composition with transition maps across the common edges of the
patches. It differs from the $G^1$ continuity of the (planar) surfaces,
which are the image of these parameterizations.

According to the pattern-based construction method of topology information
described in Section \ref{sec:topology} , there are only irregular vertices of valence $3$ or
$5$. Without loss of
generality as shown in Fig. \ref{fig:localoptimize2} ,
suppose that there are $M$ \Bezier patches $\textit{\textbf{r}}_i(u,v)$ of degree $n \times n$
meeting at the common control points $\textit{\textbf{P}}_{00}$,
$i=1,2,\cdots,M$ ($M=3$ or $M=5$).  With the \Bezier segmentation curves
$\textit{\textbf{S}}_i(t)$   constructed
in Section \ref{sec:topology2}, we construct the
control points $\textit{\textbf{P}}_{11}^i$, which are nearest to the irregular vertex
$\textit{\textbf{P}}_{00}$ among the inner control points of
$\textit{\textbf{r}}_i(u,v)$, to satisfy the general $G^1$ continuity
constraints.
Specifying the transition map in the $G^1$ conditions proposed in
\cite{Mourrain:CAGD2016},   the $G^1$
continuity constraints around the irregular vertex
$\textit{\textbf{P}}_{00}$ can be obtained as follows,
\begin{eqnarray}\label{eqn:Ri}
 && (\textit{\textbf{s}}^i_{1}-\textit{\textbf{P}}_{00})=\alpha_i
(\textit{\textbf{s}}^{i+1}_1-\textit{\textbf{P}}_{00}) + \beta_i
(\textit{\textbf{s}}^{i-1}_1-\textit{\textbf{P}}_{00}),
\label{eqn:angle}\\
&&\mathbf{0}= n\alpha_i (\textit{\textbf{P}}_{11}^i-\textit{\textbf{s}}^i_{1}) +n \beta_i
(\textit{\textbf{P}}_{11}^{i-1}-\textit{\textbf{s}}^{i}_{1}) -  (n-1)
(\textit{\textbf{s}}^i_{2}-\textit{\textbf{s}}^{i}_1)+
(\textit{\textbf{s}}^i_{1}-\textit{\textbf{P}}_{00})  ,
  i=1,2, \dots, M  \label{eqn:systemG1}
 \end{eqnarray}
Eq.\eqref{eqn:Ri} and Eq. \eqref{eqn:systemG1} can be deduced from Eq. (13) and the system after Eq. (22) in \cite{Mourrain:CAGD2016}, where the derivatives for the functions are expressed in terms of the control points of the \Bezier patches. The similar $G^1$-continuity constraints can be found in \cite{Jorg:CAGD2017,Jorg:CAD2016}.

The control points
$\textit{\textbf{s}}^i_{j}$ of the segmentation curves
are  fixed, $i=1,2,\cdots,M$, $j=0,1,\cdots, n $. In order
to satisfy the first condition \eqref{eqn:angle} of $G^1$ continuity,
we firstly determine the values of $\alpha_i$ and $\beta_i$ by
solving a $2\times2$ linear system  \eqref{eqn:angle} for a specified $i$. Then, the unknown
control points $\textit{\textbf{P}}_{11}^i$ are determined according to
the second condition \eqref{eqn:systemG1}, which can be rewritten as a
linear system,
\begin{equation}\label{eqn:g1system}
\left( {\begin{array}{*{20}{c}}
\alpha_1&0&  \ldots &0&\beta_1\\
\beta_2&\alpha_2 & \ldots &0&0\\
0&\beta_3&   \ldots &0&0\\
 \vdots & \vdots   & \ddots & \vdots & \vdots  \\
0&0 &   \ldots &\alpha_{M-1}&0\\
0&0 &  \ldots &\beta_M&\alpha_M
\end{array}} \right) \left( {\begin{array}{*{20}{c}}
\textit{\textbf{P}}_{11}^1\\
\textit{\textbf{P}}_{11}^2\\
\textit{\textbf{P}}_{11}^3\\
\vdots\\
\textit{\textbf{P}}_{11}^{M-1}\\
\textit{\textbf{P}}_{11}^M \end{array}} \right) =\left( {\begin{array}{*{20}{c}}
\textit{\textbf{H}}_1 \\
\textit{\textbf{H}}_2\\
\textit{\textbf{H}}_3\\
\vdots\\
\textit{\textbf{H}}_{M-1}\\
\textit{\textbf{H}}_{M}  \end{array}} \right)
\end{equation}
in which $\textit{\textbf{H}}_{i}=(\alpha_i+\beta_i-1)
\textit{\textbf{s}}^{i}_1  +(1-\frac{\D 1}{\D
  n} )  \textit{\textbf{s}}^{i}_2  +\frac{\D 1}{\D
  n}\textit{\textbf{P}}_{00}$, \quad $i=1,2,\cdots,M$.

\begin{proposition}\label{prop:invertible}
The coefficient matrix involved in the linear system
\eqref{eqn:g1system} is
invertible for $M=3$ and $M=5$ , i.e, there exists unique solution of the linear system
\eqref{eqn:g1system}.
\end{proposition}

\begin{proof}
The determinant of   $M \times M$  matrix in the linear system
\eqref{eqn:g1system} for $M=3$ and $M=5$ is
\begin{equation}\label{eqn:deter}
\prod\limits_{i=1}^{M} \alpha_i +
\prod\limits_{i=1}^{M} \beta_i.
\end{equation}
in which $\alpha_i$ and $\beta_i$ are determined from (\ref{eqn:angle}).

From the vertex enclosure formulation with $G^1$-continuity \cite{Jorg:CAD2011}, we have
$$
\prod\limits_{i=1}^{M} \alpha_i =
\prod\limits_{i=1}^{M} \beta_i.
$$
From (\ref{eqn:deter}),
the $M \times M$  matrix  in the linear system
\eqref{eqn:g1system}  is invertible, i.e, there exists unique solution of the linear system\eqref{eqn:g1system}.
\end{proof}

\subsection{Local $C^1$ linear-energy-minimizing method for constructing inner
  control points }\label{sec:energy}

In this section, we will propose a local $C^1$ linear-energy-minimizing method for
constructing interior $(n-3)\times(m-3)$  control points of each patch with prescribed
boundary two-layer control points as shown in Fig.\ref{fig:localoptimize}. This
problem can be stated as follows: Given the boundary control points on
the first two layers of the control mesh, i.e. $\textit{\textbf{P}}_{0,j}$,$\textit{\textbf{P}}_{1,j}$,$\textit{\textbf{P}}_{n-1,j}$
$\textit{\textbf{P}}_{n,j}$, $\textit{\textbf{P}}_{i,0}$,
$\textit{\textbf{P}}_{i,1}$, $\textit{\textbf{P}}_{i,n-1}$   and
$\textit{\textbf{P}}_{i,n}$ of a tensor \Bezier patch
$\textbf{\emph{r}}(u,v)$, find the remaining interior control points
$\{\textbf{\emph{P}}_{i,j}\}_{i,j=2}^{n-2,n-2}$, such that the
following energy function $E(\textbf{\emph{r}})$ is minimal,
\begin{equation}\label{eq:quasifunc}
E(\emph{\textbf{r}}) = \int_\Omega \tau_1 (\|\textit{\textbf{r}}_{u}\|^2 +\|\textit{\textbf{r}}_{v} \|^2)
+  \tau_2 (\|\textit{\textbf{r}}_{uu}\|^2 + 2\|\textit{\textbf{r}}_{uv}\|^2
+ \|
\textit{\textbf{r}}_{vv} \|^2)dudv.
\end{equation}
The energy function $E(\textbf{\emph{r}})$ is related to the
orthogonality and uniformity of the iso-parametric structure on
the \Bezier surface \cite{xu:jcp2013}. $\tau_1$ and $\tau_2$ are
positive weights to control the parameterization results: if $\tau_1$ is big, then the
resulting iso-parametric structure has better orthogonality; if $\tau_2$
has a small value, then we can obtain an iso-parametric grid with
better uniformity.

Different from the non-linear optimization method
\cite{xu:jcp2013}, with prescribed boundary two-layer control points, we will give the sufficient and necessary
condition for the interior control points of the \Bezier patch with minimal energy
$E(\textbf{\emph{r}})$ in the following proposition.
\begin{proposition}
\label{th:sncondition1} Given the boundary
control points $\textit{\textbf{P}}_{0,j}$,$\textit{\textbf{P}}_{1,j}$,$\textit{\textbf{P}}_{n-1,j}$
$\textit{\textbf{P}}_{n,j}$, $\textit{\textbf{P}}_{i,0}$,
$\textit{\textbf{P}}_{i,1}$, $\textit{\textbf{P}}_{i,n-1}$   and
$\textit{\textbf{P}}_{i,n}$  of a tensor product B\'{e}zier surface
$\textbf{r}(u,v)$, then
$\textbf{r}(u,v)$ has minimal energy $E(\textbf{r})$
if and only if the remaining inner control points
$\{\textbf{\emph{\emph{P}}}_{i,j}\}_{i,j=2}^{n-2,n-2}$ satisfy
\begin{eqnarray}\label{eqn:result}
\textbf{0}&=&\frac{\tau_1}{4(n-1)}\left(\sum\limits_{k=0}^{n-1}\sum\limits_{l=0}^{n}{\frac{ {n \choose
 l}}{  {2n \choose
 l+j}}C_{n,i}^{k}\Delta^{1,0}\textit{\textbf{P}}_{kl}}+\sum\limits_{k=0}^{n}\sum\limits_{l=0}^{n-1}{\frac{ {n \choose
 k}}{  {2n \choose
 i+k}}C_{n,j}^{l}\Delta^{0,1}\textit{\textbf{P}}_{kl}}\right)
\nonumber \\
& &+\frac{2 \tau_2}{(2n-1)^2}\sum\limits_{k=0}^{n-1}\sum\limits_{l=0}^{n-1}{\frac{ {n-1 \choose
 k}{n-1 \choose
 l}}{ {2n-2 \choose
 i+k-1}{2n-2 \choose
 l+j-1}}B_{n,i}^{k}B_{n,j}^{l}\Delta^{1,1}\textit{\textbf{P}}_{kl}}\\
&&+\frac{\tau_2}{(2n-3)(2n+1)}\left(\sum\limits_{k=0}^{n-2}\sum\limits_{l=0}^{n}{\frac{ {n-2 \choose
 k}{n \choose
 l}}{ {2n-4 \choose
 i+k-2}{2n \choose
 l+j}}A_{n,i}^{k}\Delta^{2,0}\textit{\textbf{P}}_{kl}} +\sum\limits_{k=0}^{n}\sum\limits_{l=0}^{n-2}{\frac{ {n \choose
 k}{n-2 \choose
 l}}{ {2n \choose
 i+k}{2n-4 \choose
 l+j-2}}A_{n,j}^{l}\Delta^{0,2}\textit{\textbf{P}}_{kl}}\right) \nonumber
\end{eqnarray}
where
\begin{eqnarray*}
A_{n,i}^{k}&=&\frac{nk(n-1)(k-1)+i^2(n-2)(n-3)-i(n-3)(2kn+n-2k-2)}{ (2n-i-k-2)(2n-i-k-3)}\\
B_{n,i}^{k}&=&\frac{[n(k-i)-k][n(k-i)-n+2i]-ni(2n-i-k-1)}{ (2n-i-k)(2n-i-k-1)}\\
C_{n,i}^{k}&=&\frac{ni-nk-i}{ 2n-i-k-1}\frac{ {n-1 \choose
 k}}{  {2n-2 \choose
 i+k-1}}
\end{eqnarray*}
\end{proposition}

A proof of Proposition \ref{th:sncondition1} is given in Appendix II.

Since the energy functional $E(\textbf{r})$ has a lower bound, the
corresponding energy-minimizing \Bezier surface always exists. From
Proposition  \ref{th:sncondition1}, for each unknown interior  control point, an
equation is obtained. Hence, we can get a linear system with
$(n-3)\times(n-3)$ equations and $(n-3)\times(m-3)$ variables as follows,
\[
\textbf{M} \textbf{P} = \textbf{B}
\]
in which $\textbf{M}$ is the coefficient matrix determined by
Eq.(\ref{eqn:result}),
$\textbf{P}=\{\textbf{\emph{P}}_{i,j}\}_{i,j=2}^{n-2,n-2}$
is the set of unknown interior control points, $\textbf{B}$ is
the right-hand side related to the specified control points on the
first two layers of control mesh.

Solving this linear system, the unknown interior control points
 $\{\textbf{\emph{P}}_{i,j}\}_{i,j=2}^{n-2,n-2}$ can be
represented as a linear combination of the specified control
points on the first two layers of the control mesh. Furthermore, as each
\Bezier patch in the computational domain has  the same
coefficient matrix  $\textbf{M}$, we only compute the
inverse of $\textbf{M}$ once and reuse it for the remaining \Bezier patches in
the computational domain.

\subsection{Local optimization approach for  injective parameterization}
\label{sec:valid}

A parameterization is valid if it has no self-intersections, meaning the mapping from the
parametric domain to the computational domain is injective.
After the local energy-minimizing patch construction described in
Section \ref{sec:energy} , usually most of the \Bezier patches are valid. However, for a few patches the resulting
parameterization may still have self-intersections. In this section,
we will firstly identify these invalid patches and subsequently reposition their control points to recover validity.

A parameterization is injective if its Jacobian is
positive everywhere \cite{gravessen:parameterization, xu:cmame2011}. For a planar \Bezier surface $\textit{\textbf{r}}(u,v)=
\sum\limits_{i=0}^{n}\sum\limits_{j=0}^{n} \pb_{i,j}
B_i^{n}(u) B_j^{n}(v)$, its Jacobian
can be represented as a high-order Bernstein polynomial
\cite{gravessen:parameterization, xu:cmame2011},
\begin{eqnarray}
J (u,v) &=&
\sum_{i=0}^{2n-1}\sum_{j=0}^{2n-1}  \alpha_{ij}
B_i^{2n-1}(u) B_j^{2n-1}(v).
\label{eq:J}
\end{eqnarray}
Hence, the Jacobian $J(u,v)$ is bounded by
\begin{equation}
 \min_{ 0 \leq i,j \leq 2n-1
 } \alpha_{ij} \leq J(u,v) \leq
 \max_{ 0 \leq i,j \leq 2n-1 } \alpha_{ij}
\end{equation}

Our invalid-patch-finding method can be described as follows: For each
\Bezier patch $\textit{\textbf{r}}(u,v)$ on the computational domain
$\Omega$, we compute its Jacobian coefficients  $\alpha_{ij}$
according to Eq. (\ref{eq:J}); if  $\displaystyle{\min_{ 0 \leq i,j \leq 2n-1
}  \alpha_{ij}}  > 0$, then
$\textit{\textbf{r}}(u,v)$  is called \emph{valid patch}; otherwise,
$\textit{\textbf{r}}(u,v)$ is called  \emph{invalid patch}.

After finding all the invalid patches on the planar parameterization,
we will repair the invalid patch locally by repositioning its
internal control points such that all the Jacobian coefficients
$\sigma_{ij}$ are positive. Gravessen et
al. \cite{gravessen:parameterization}
and Xu et al. \cite{xu:cmame2011} investigated this problem with
 a general non-linear constrained optimization framework,
\begin{eqnarray}\label{eq:alpha}
    && \text{min} \quad E(\emph{\textbf{r}}(u,v))  \quad \text{s.t.} \quad  \alpha_{ij}> 0
\end{eqnarray}
in which $E(\emph{\textbf{r}}(u,v))$ is defined as (\ref{eq:quasifunc}).

In contrast to the proposed optimization method in
\cite{gravessen:parameterization, xu:cmame2011},
we solve this constrained optimization problem using the
classical logarithmic-barrier method \cite{logbarrier, falini:Fiacco,untangling}:
\begin{equation}\label{eqn:optimized}
\underset{\textit{\textbf{P}}_{i,j}}{\operatorname{arg\,min}}
\quad E(\emph{\textbf{r}}(u,v)) - \mu \sum_{i=0}^{2n-1} \sum_{j=0}^{2n-1} \ln(\alpha_{ij})
\end{equation}
where $\mu$  is a positive penalty parameter.

The results from our experiments indicate that the injective
parameterization can be achieved by the log-barrier method, although
without theoretical guarantee of injectivity. Since the corresponding non-linear optimization (\ref{eqn:optimized}) is only performed on very few local invalid \Bezier patches with
limited number of design variables, we can obtain the final optimized
parameterization efficiently.

\noindent \textbf{Remark 6.1. }  For a general framework to obtain an injective parameterization, the optimization problem (\ref{eq:alpha}) can be generalized as follows :
\begin{eqnarray*}
    && \text{min} \quad E(\emph{\textbf{r}}(u,v))  \quad \text{s.t.} \quad  J (u,v)> 0
\end{eqnarray*}
in which $J(u,v)$ is defined in (\ref{eq:J}).

\noindent \textbf{Remark 6.2. }
  Teichm\"{u}ller mapping method proposed in \cite{Nian:cmame2016}, which has theoretical guarantees of injectivity,  can be used to improve the  proposed local optimization scheme .

\noindent \textbf{Remark 6.3.} Due to the $C^1/G^1$-constraints and the local optimization
method in our framework, the minimum degree of the \Bezier patches to
guarantee a solution is 4. If the degree of the input boundary is
smaller than 4, degree-elevation operation should be performed.

\section{Examples and comparison}

\label{sec:example}

Starting from a set of B-spline curves as input, the proposed
framework for analysis-suitable planar parameterization
has been implemented as a plugin in the
AXEL \footnote{http://axel.inria.fr/} platform.  In this section,  five
parameterization examples are presented to show the effectiveness
of the proposed method. The comparison with the skeleton-based
parameterization method \cite{xu:cmame2015} is also performed for three
parameterization examples.

\subsection{Metrics for quality evaluation}

In order to evaluate the quality of the planar parameterization
results,  we use the scaled Jacobians and the condition number of the Jacobian matrices as two
important criteria. The \emph{scaled Jacobian} at $\textit{\textbf{r}}(u,v)$
can be computed as follows,
\begin{equation}
J_{s}(u,v)= \frac{J(u,v)}{\|\textit{\textbf{r}}_u \| \|\textit{\textbf{r}}_v \|}
\end{equation}

A parameterization  is said to be inverted if its Jacobians at some
place is less or equal to zero. The minimal requirement for a planar parameterization to be suitable for
isogeometric analysis is that all scaled Jacobians are as close as
possible to 1.0.

The condition number $\kappa(\mathbf J)$ of a Jacobian matrix $\mathbf J$ is another important evaluation criteria, given by
\begin{equation}
\kappa(\mathbf J) = |\mathbf J| |\mathbf J^{-1}|
\end{equation}
in which
\begin{equation}
\mathbf J=\left({\begin{array}{cc}
       x_u & y_u  \\
       x_v & y_v
 \end{array}} \right),
\end{equation}
The Frobenius norm of $\mathbf J$ is defined as $|\mathbf J| =
(tr(\mathbf J^{T}\mathbf J))^{1/2}$.
It can be easily proven that the minimum of $\kappa(\mathbf J)$
is 2.0.  For the quality evaluation of planar parameterization,
the condition number $\kappa(\mathbf J)$ is smaller, the
optimization is better.

The parameterization quality with respect to the scaled Jacobians and condition number for the
examples presented in this paper are summarized in
Table. \ref{table:rrefine1} and Table. \ref{table:rrefine}.

\subsection{Effect of quad meshing}

In the proposed framework, the quad-meshing step has a great
effect on the final parameterization results.
Furthermore, the quad-meshing results depend on the approximate convex decomposition of the input computational domain. For the approximate convex decomposition method in \cite{acd},
if  different non-concavity tolerances $\epsilon$ are specified,  different number of approximate convex sub-domain can be obtained. Hence, from the same input boundary Fig. \ref{fig:exam5} (a), different quad-meshing results can be generated from different selection of non-concavity tolerance $\epsilon$ as shown in Fig. \ref{fig:exam5}(b) with $\epsilon=0.6$ and Fig. \ref{fig:exam5}(e) with $\epsilon=0.4$.
For a smaller non-concavity tolerances $\epsilon$, we can obtain a quad-meshing result with more quad elements and more irregular vertices. Moreover, the corresponding
scaled Jacobian colormap of the parameterization in Fig. \ref{fig:exam5}
(c) and Fig. \ref{fig:exam5} (g) are depicted in Fig. \ref{fig:exam5}
(d) and Fig. \ref{fig:exam5} (h).

\begin{figure}
\centering
\begin{minipage}[t]{2.in}
\centering
\includegraphics[width=2.in ]{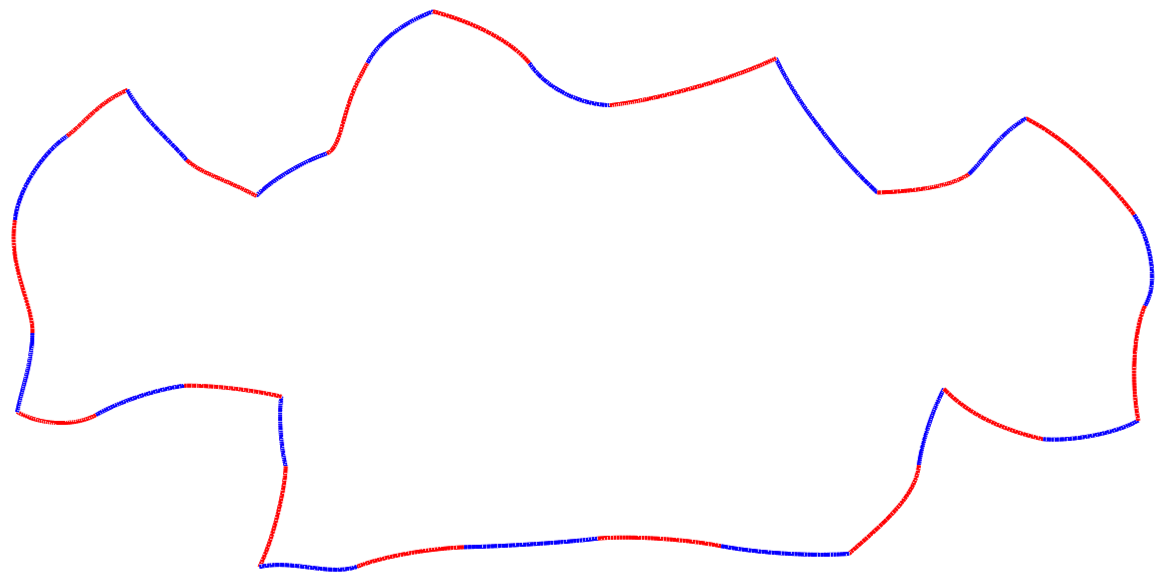}
\\ (a) boundary \Bezier curves
\end{minipage} \qquad
\begin{minipage}[t]{2.in}
\centering
\includegraphics[width=2.in ]{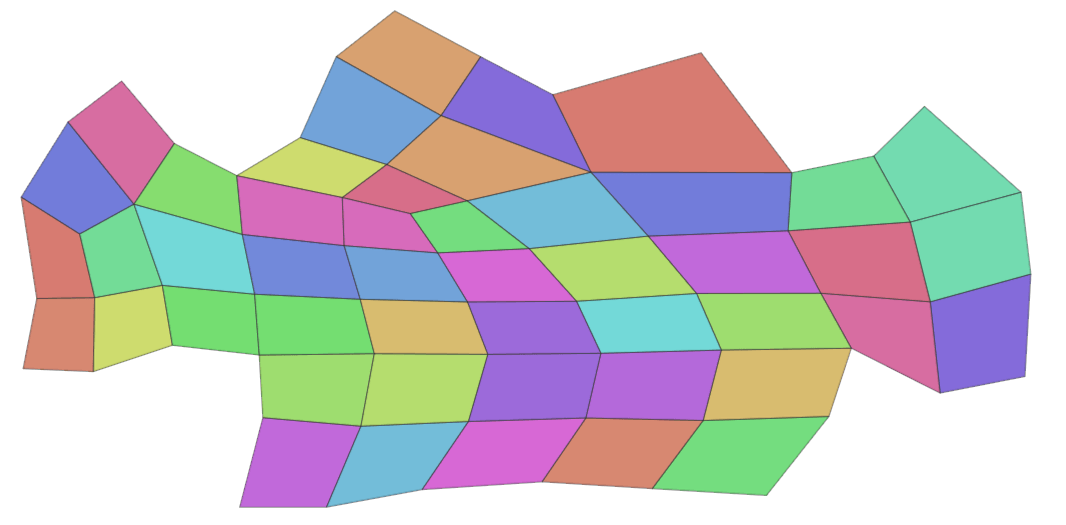}
\\ (b) quad meshing result I with $\epsilon=0.6$
\end{minipage} \\
\begin{minipage}[t]{2.in}
\centering
\includegraphics[width=2.in ]{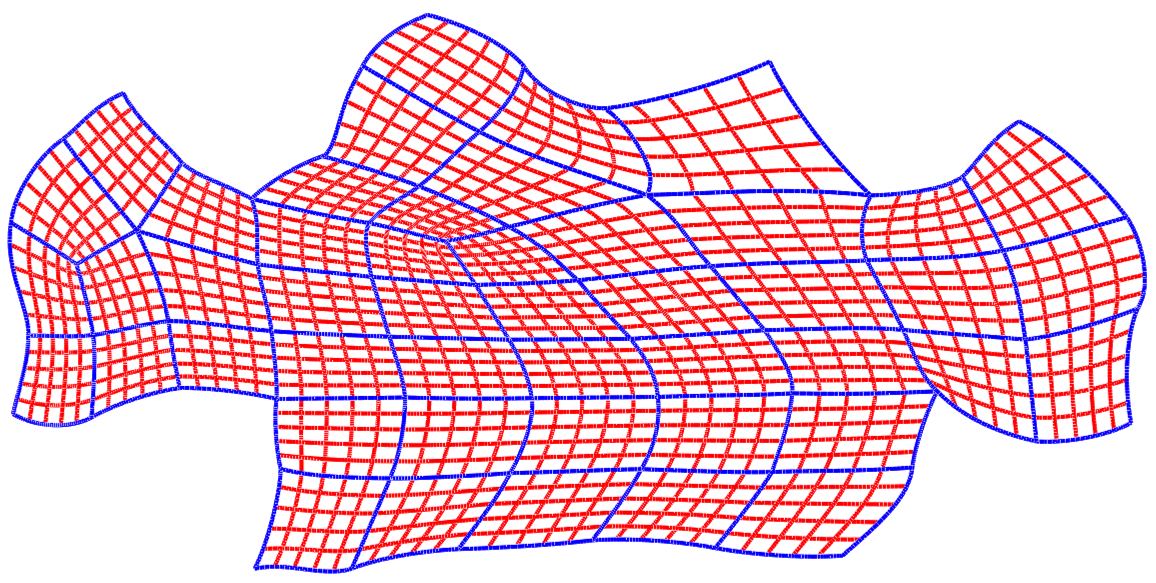}
\\ (c) parameterization result I
\end{minipage}  \quad
\begin{minipage}[t]{2.4in}
\centering
\includegraphics[width=2.4in]{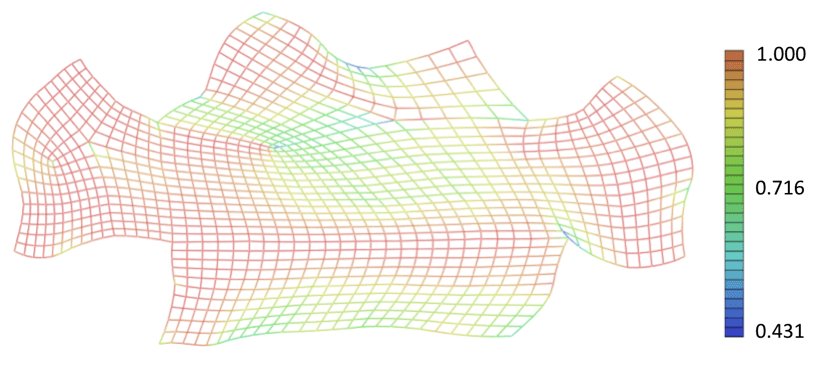}
\\ (d) Jacobian colormap I
\end{minipage}\\
\begin{minipage}[t]{2.in}
\centering
\includegraphics[width=2.in ]{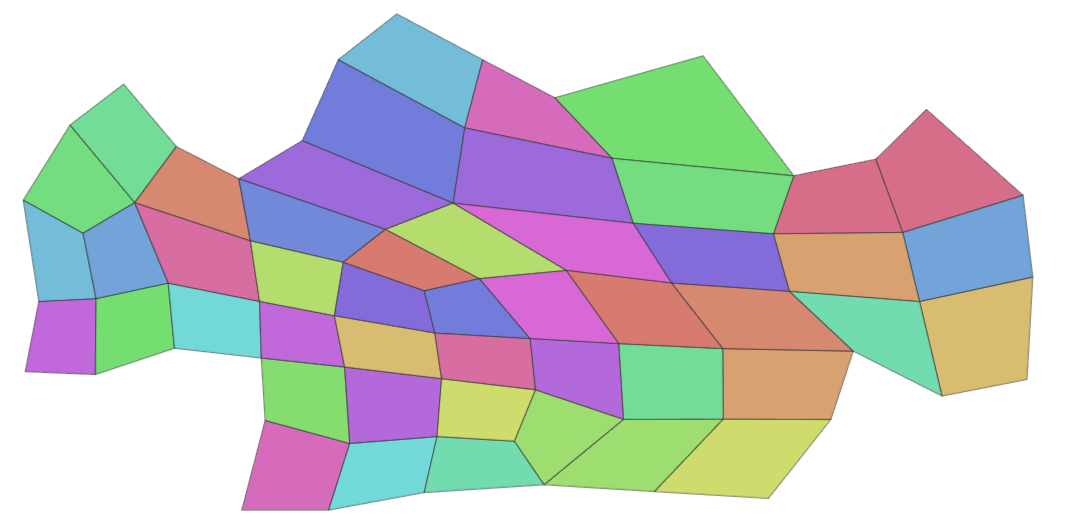}
\\ (e) quad meshing result II with $\epsilon=0.4$
\end{minipage} \qquad
\begin{minipage}[t]{2.in}
\centering
\includegraphics[width=2.in ]{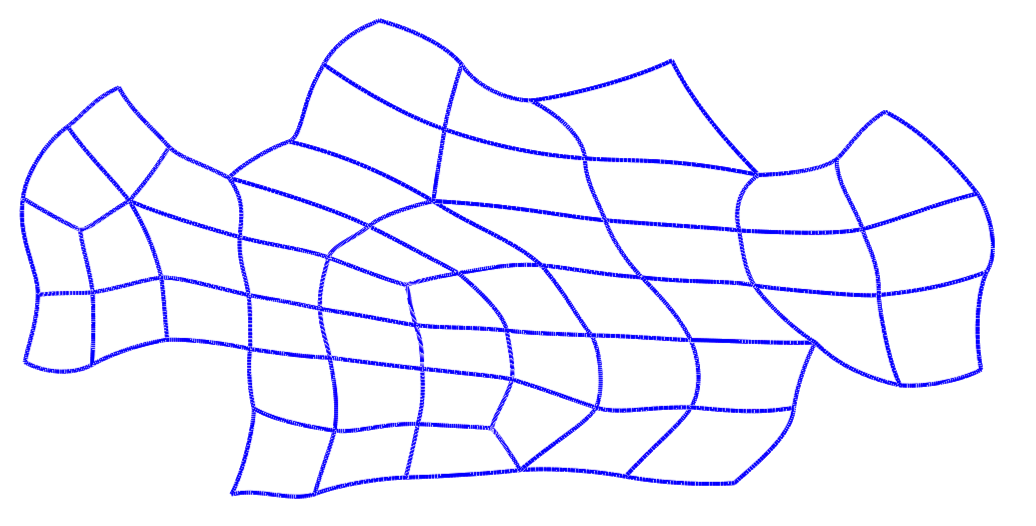}
\\ (f) segmentation curves II
\end{minipage} \\
\begin{minipage}[t]{2.in}
\centering
\includegraphics[width=2.in ]{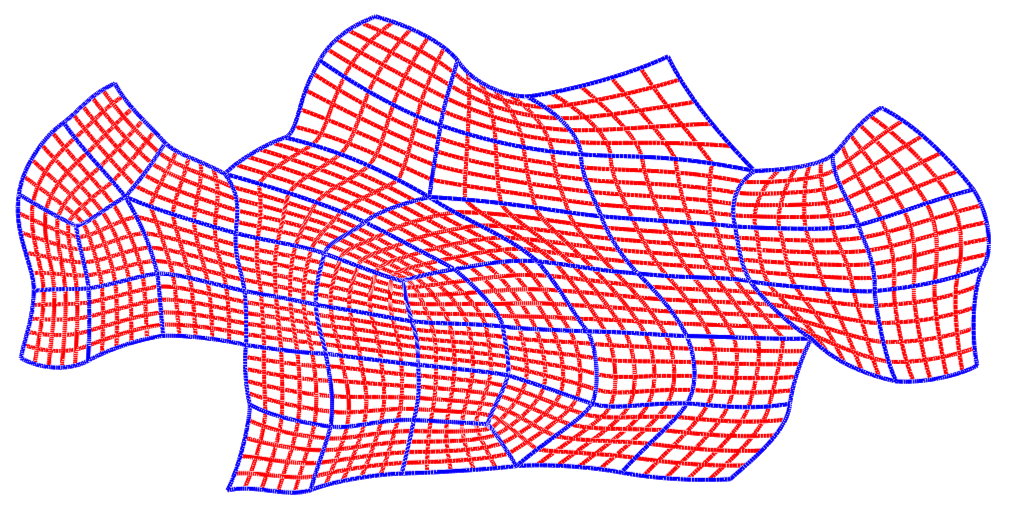}
\\ (g) parameterization result II
\end{minipage}  \quad
\begin{minipage}[t]{2.4in}
\centering
\includegraphics[width=2.4in]{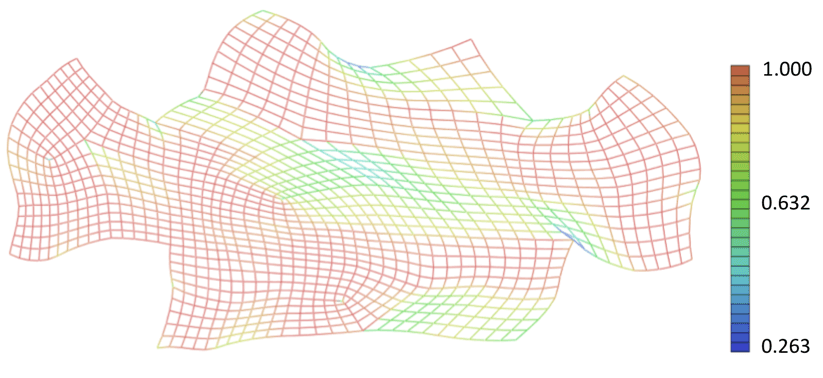}
\\ (h) Jacobian colormap
\end{minipage}\\
\begin{minipage}[t]{2.in}
\centering
\includegraphics[width=2.in ]{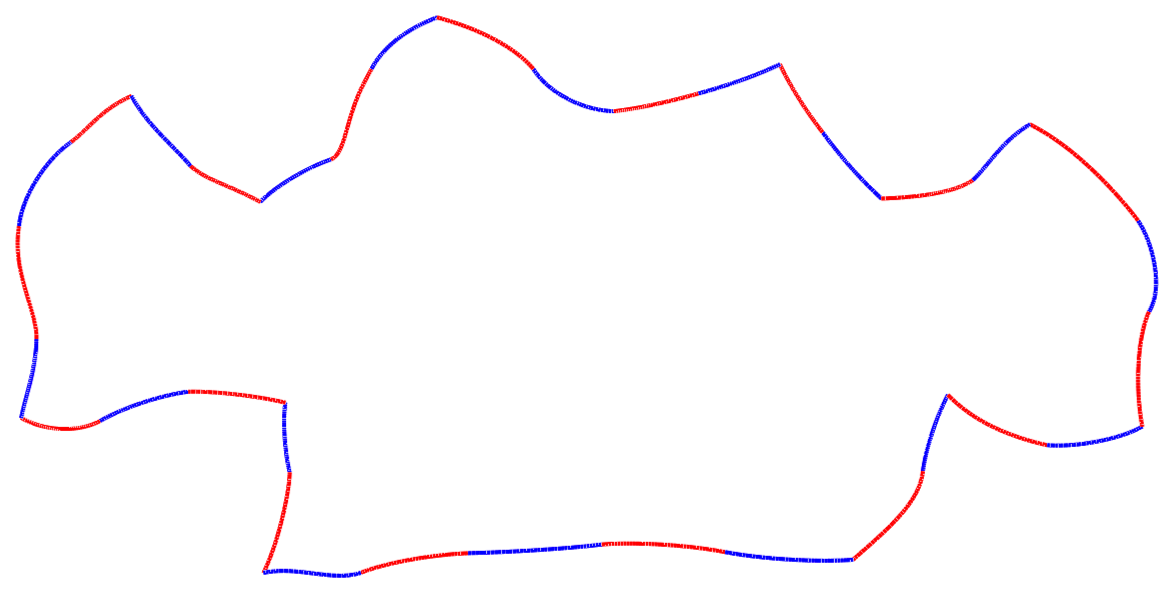}
\\ (i) boundary with more \Bezier curves
\end{minipage} \qquad
\begin{minipage}[t]{2.in}
\centering
\includegraphics[width=2.in ]{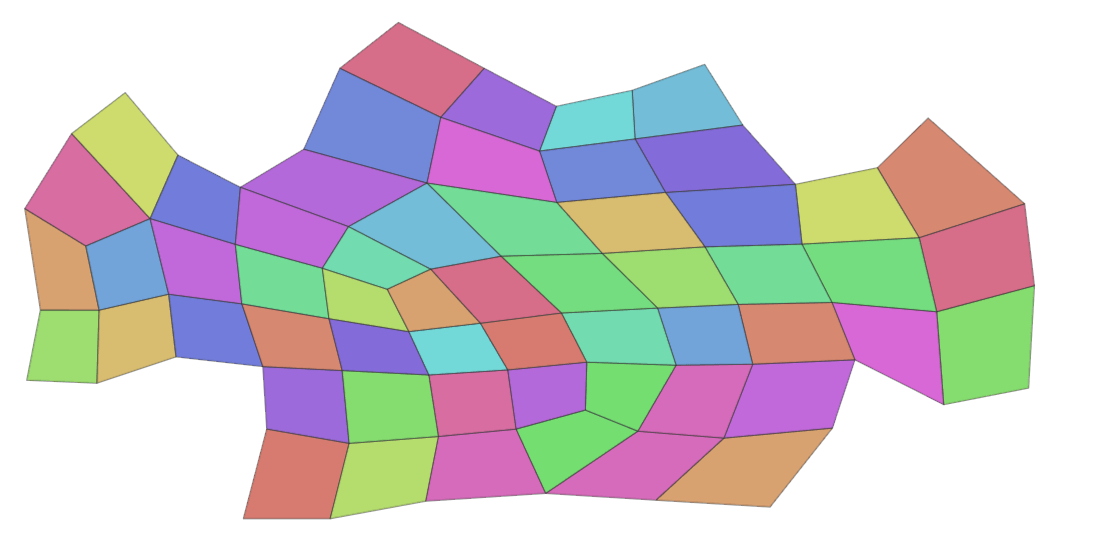}
\\ (j) quad meshing result III
\end{minipage} \\
\begin{minipage}[t]{2.in}
\centering
\includegraphics[width=2.in ]{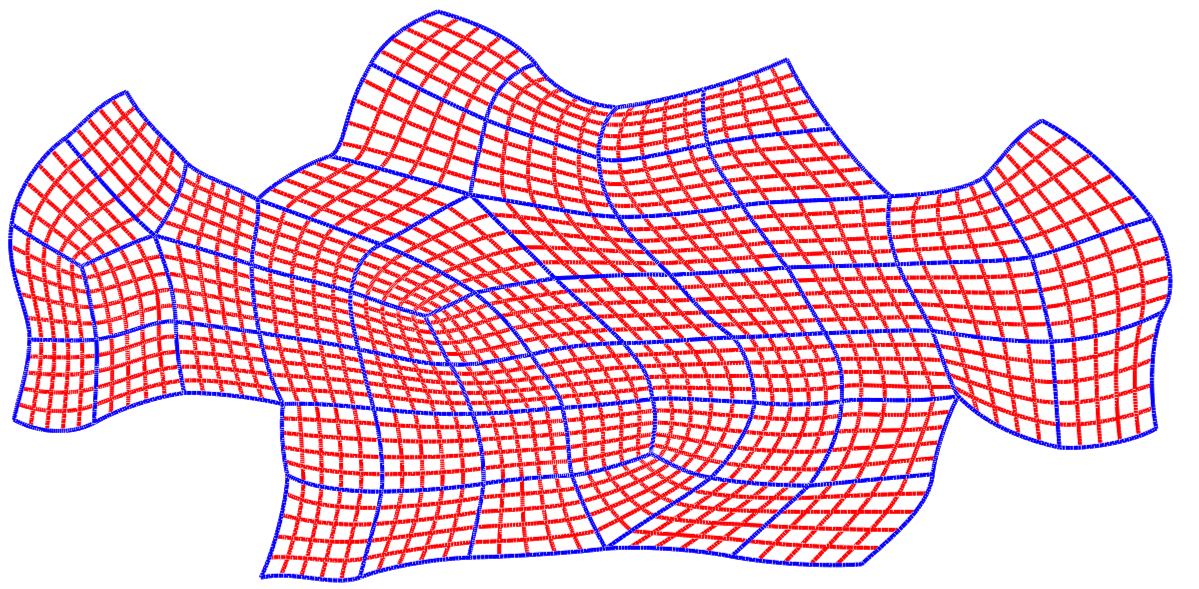}
\\ (k) parameterization result III
\end{minipage} \quad
\begin{minipage}[t]{2.4in}
\centering
\includegraphics[width=2.4in]{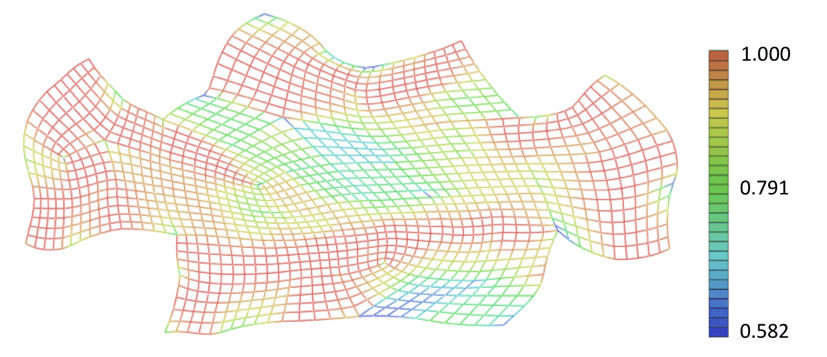}
\\ (l) Jacobian colormap III
\end{minipage}\\
\caption{Example II to show the effect of quad meshing.}
\label{fig:exam5}
\end{figure}

The number of \Bezier segments for the same input
boundary also influences the final parameterization results.
For the same input boundary with different number of \Bezier segments,
different quad-meshing result will be obtained leading to different
parameterizations. For instance, the same
boundary with different number of \Bezier segments is shown in
Fig. \ref{fig:exam5} (i). The corresponding quad-meshing result is
presented in Fig. \ref{fig:exam5} (j), and the final parameterization
and the corresponding scaled Jacobian colormap are
illustrated in Fig. \ref{fig:exam5} (k) and Fig. \ref{fig:exam5} (l)
respectively. Quantitative data of three parameterization results presented in
Fig. \ref{fig:exam5} (c), Fig. \ref{fig:exam5} (g) and
Fig. \ref{fig:exam5} (k) are listed in Table \ref{table:rrefine1}.

\begin{table}[t]
  \caption{Quantitative data for planar parameterization in
   Fig. \ref{fig:exam5} (c),
    Fig. \ref{fig:exam5} (g) and Fig. \ref{fig:exam5} (k).  $p$: degree
    of planar parameterization; \# Con.: number of control points; \#
    Patch.: number of patches. } 
  \centering 
  \begin{tabular}{c | c c c | c c c c c c c} 
     \hline 

\centering
  \multirow{2}{*} {Example} & \multirow{2}{*} {$p$} &
  \multirow{2}{*} {\# Con.} & \multirow{2}{*} {\#Patch}
      & \multicolumn{3}{c}{Scaled Jacobian} &
      &\multicolumn{3}{c}{Condition number }\\
\cline{5-7} \cline{9-11} &&&&    Max & Average & Min &  & Max & Average   & Min   \\ [0.5ex]   
    \hline 
\hline
     Fig. \ref{fig:exam5} (c) & $4$ &  817 & 47   & 1.000 & 0.9054&
     0.431  & &  4.72 & 2.41 & 2.00\\ 
     Fig. \ref{fig:exam5} (g)  & $4$ & 841   & 49  & 1.000 & 0.9024 & 0.263 &
     & 7.62 & 2.53 & 2.00\\ 
     Fig. \ref{fig:exam5} (k) & $4$ & 923 & 57     & 1.000 & 0.9279
     &0.582 &  & 3.86 & 2.36 & 2.00   \\  
    \hline 
  \end{tabular}
  \label{table:rrefine1} 
\end{table}

\begin{figure}
\centering
\begin{minipage}[t]{2.1in}
\centering
\includegraphics[width=2.1in]{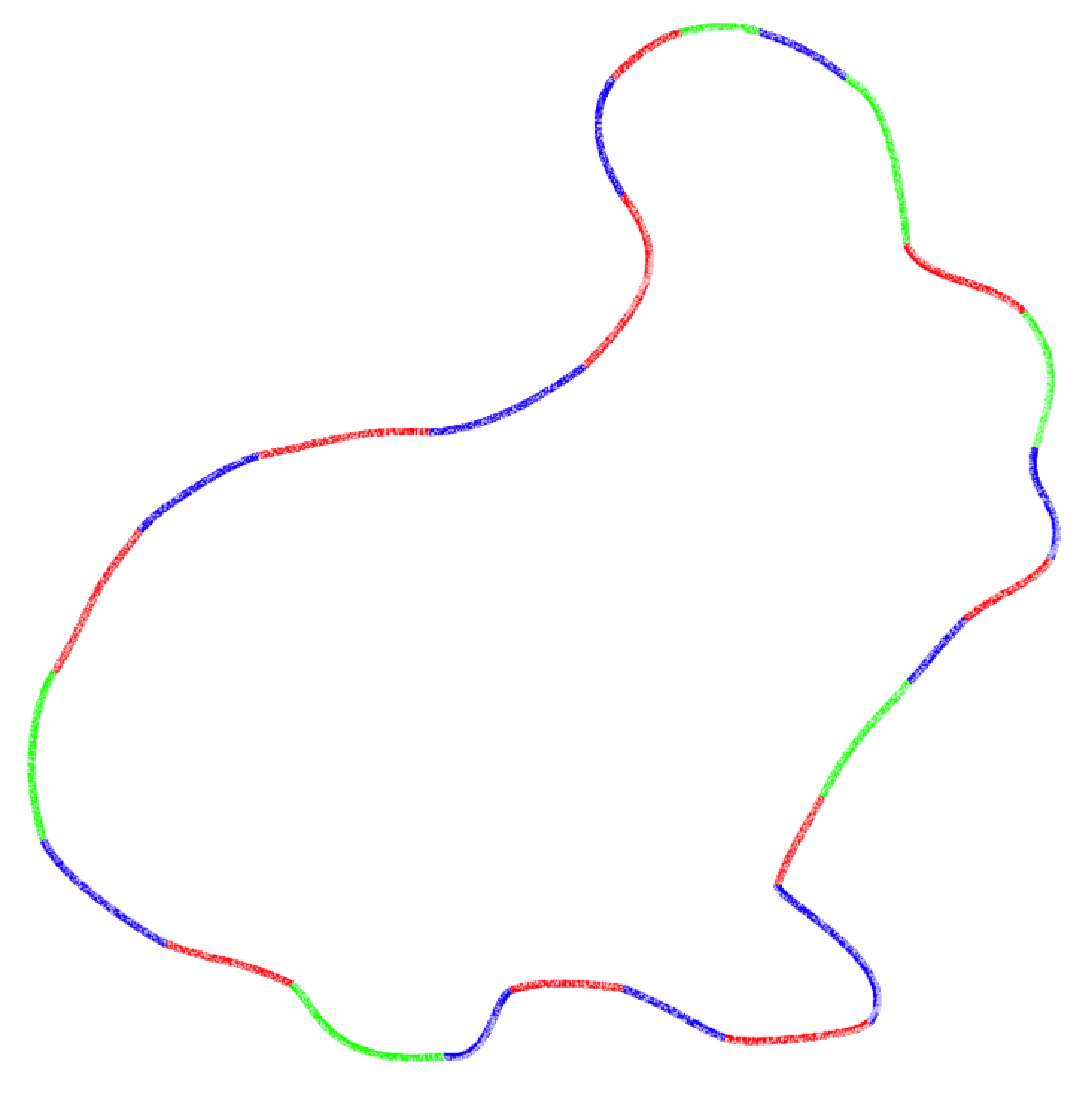}
\\ (a) boundary \Bezier curves
\end{minipage}
\begin{minipage}[t]{2.1in}
\centering
\includegraphics[width=2.3in]{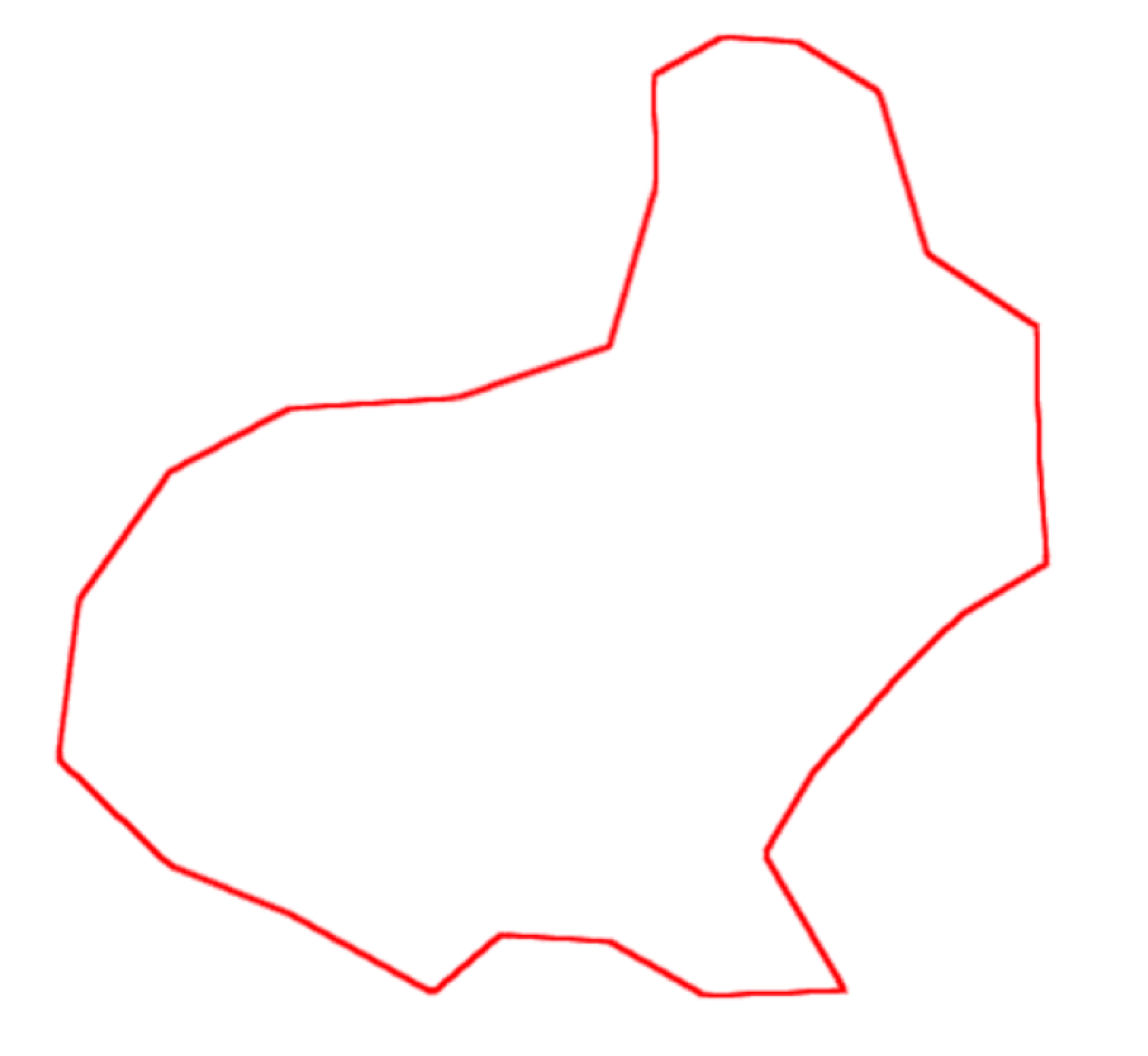}
\\ (b) discrete boundary
\end{minipage}
\begin{minipage}[t]{2.1in}
\centering
\includegraphics[width=2.2in]{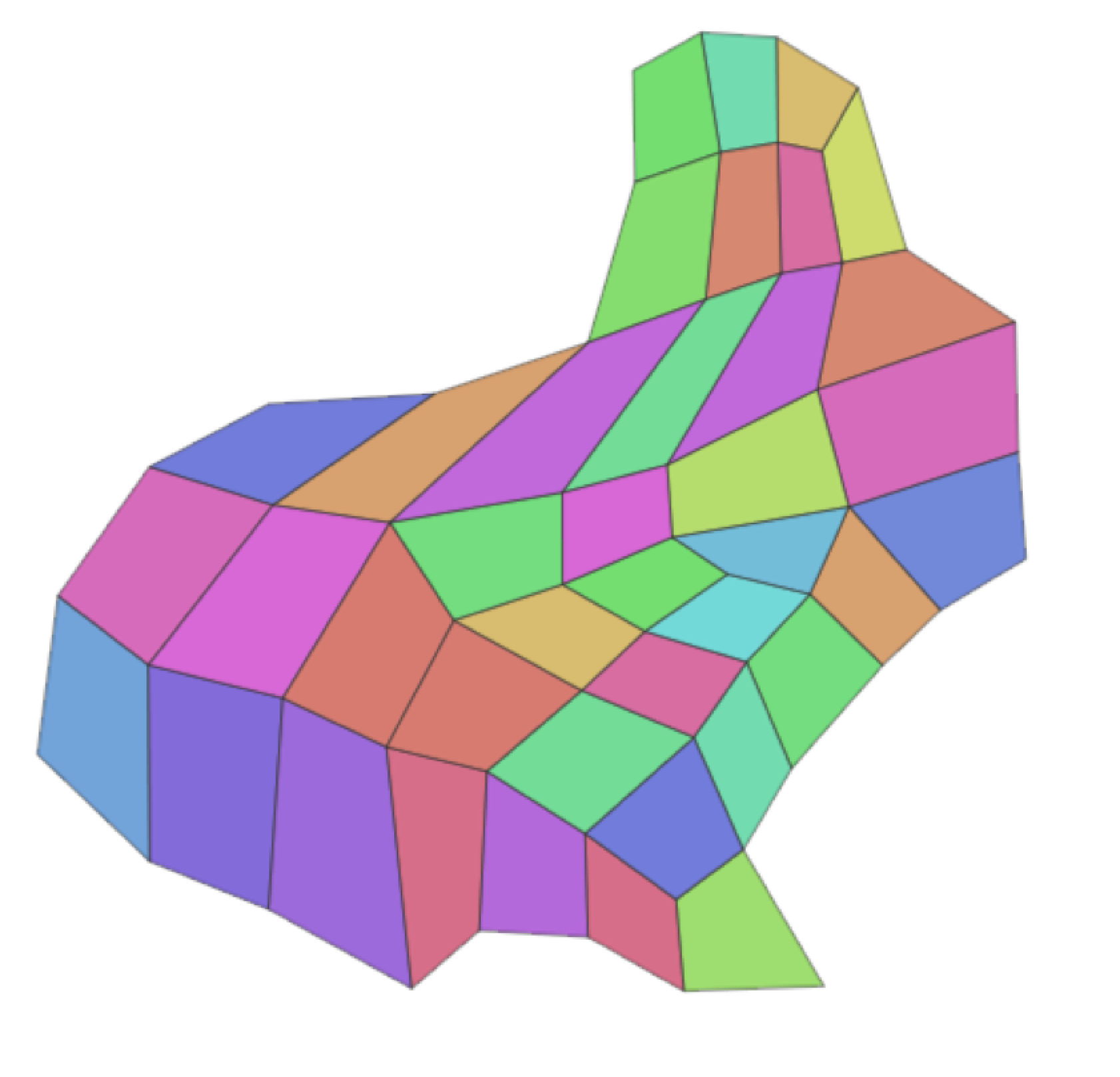}
\\ (c) quad meshing result
\end{minipage}\\
\begin{minipage}[t]{2.1in}
\centering
\includegraphics[width=2.1in]{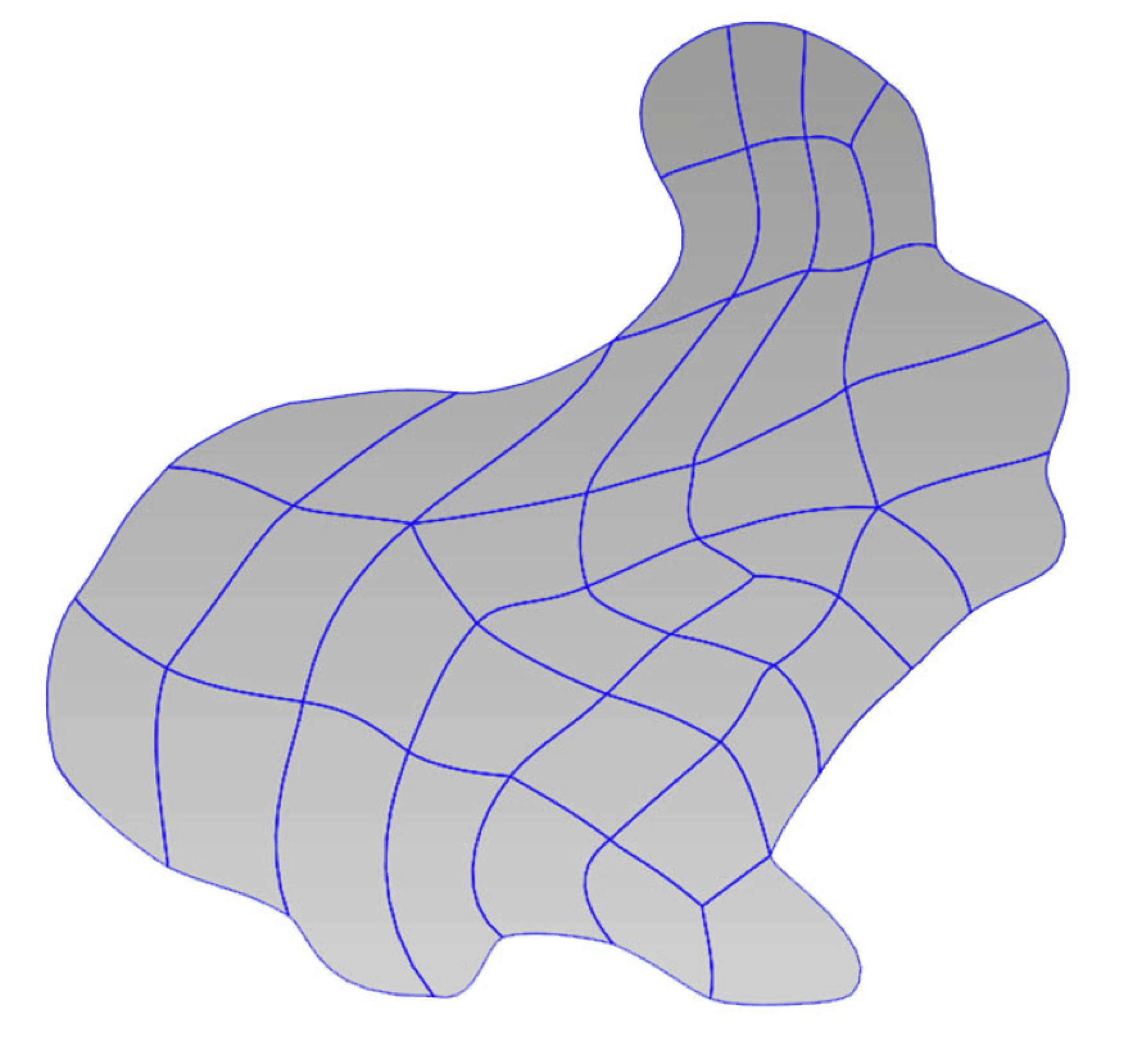}
\\ (d) segmentation curves
\end{minipage}
\begin{minipage}[t]{2.1in}
\centering
\includegraphics[width=2.1in]{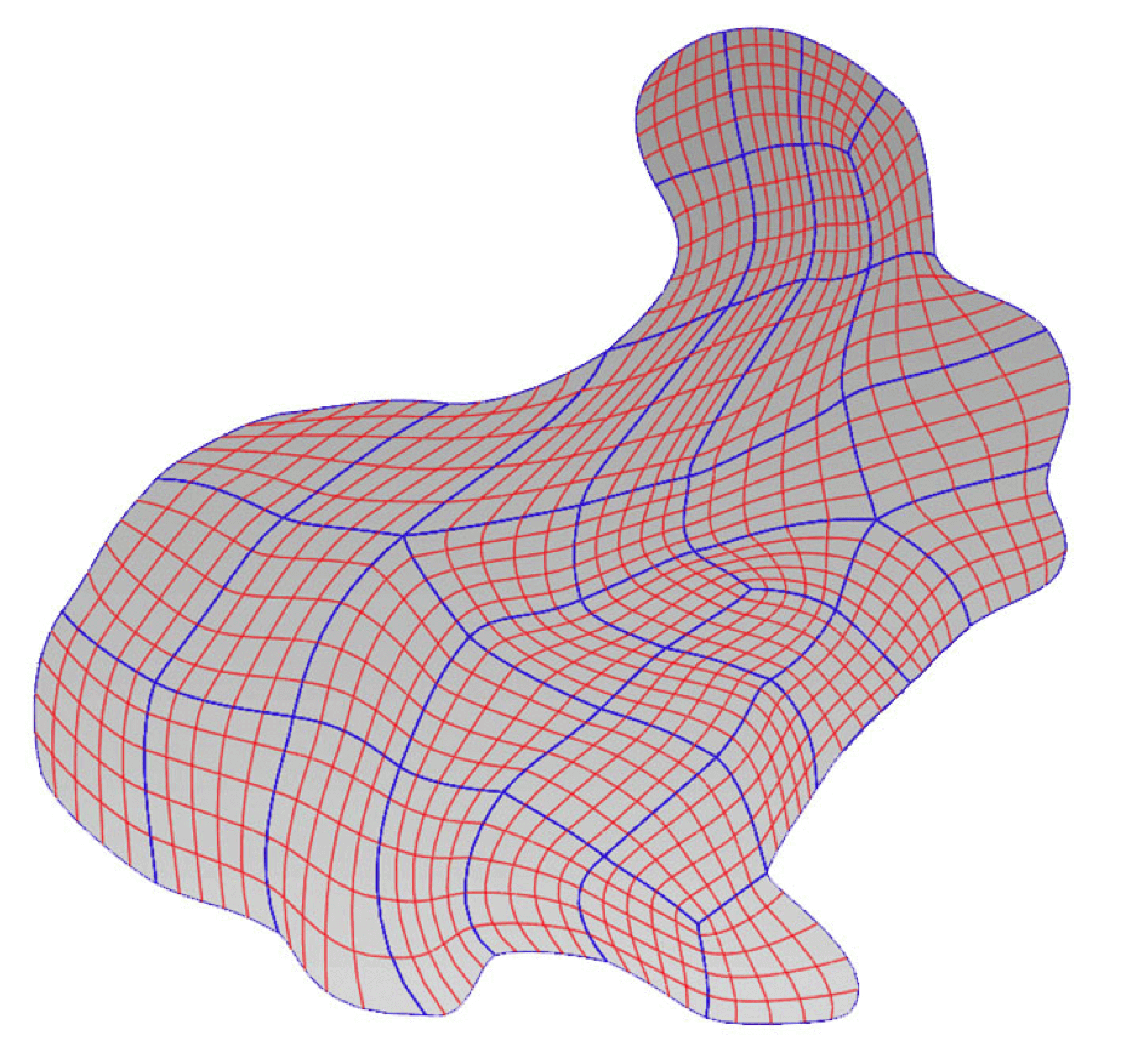}
\\ (e) parameterization result
\end{minipage}
\begin{minipage}[t]{2.1in}
\centering
\includegraphics[width=2.6in]{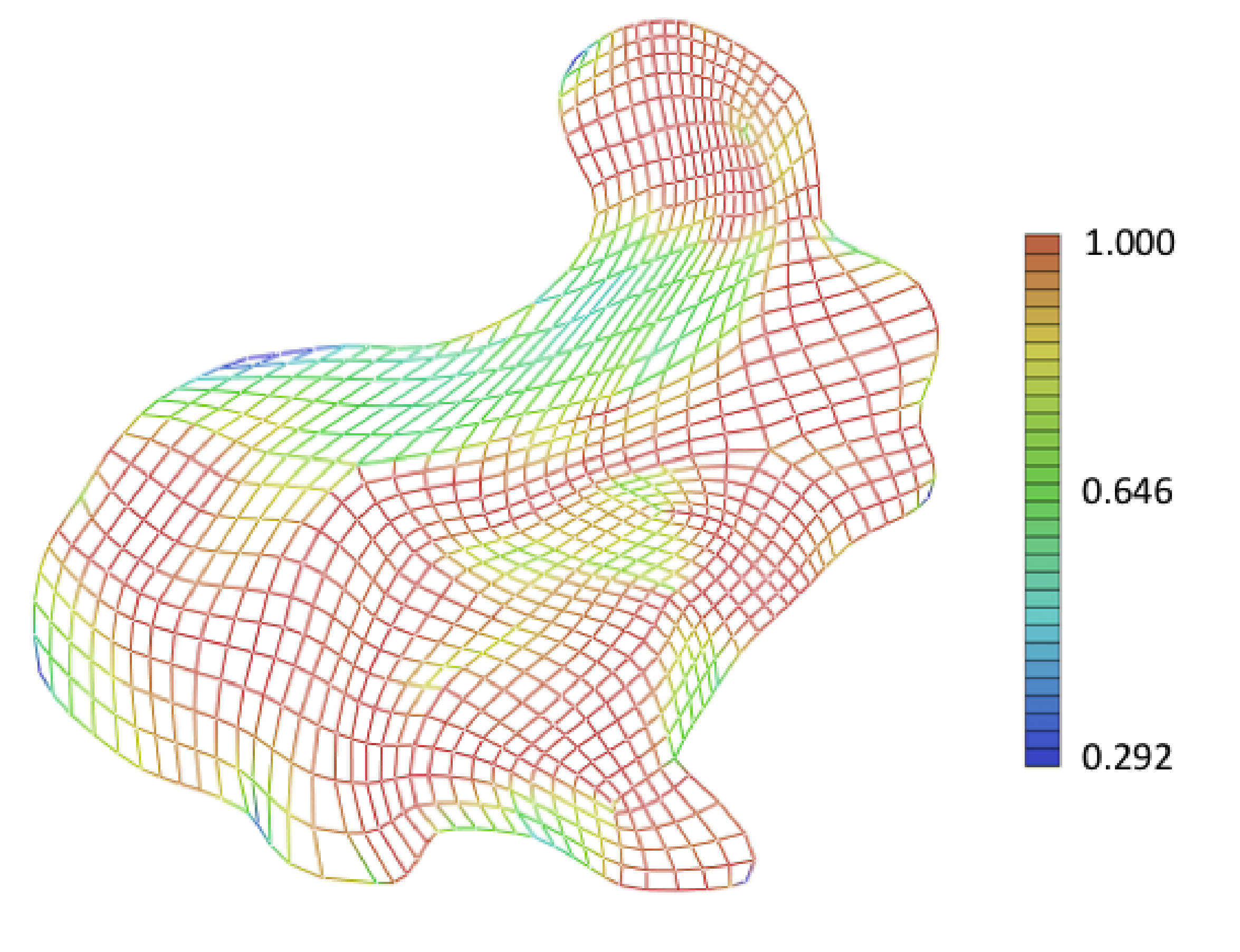}
\\ (f) Jacobian colormap
\end{minipage}\\
\begin{minipage}[t]{2.1in}
\centering
\includegraphics[width=2.05in,height=2.05in]{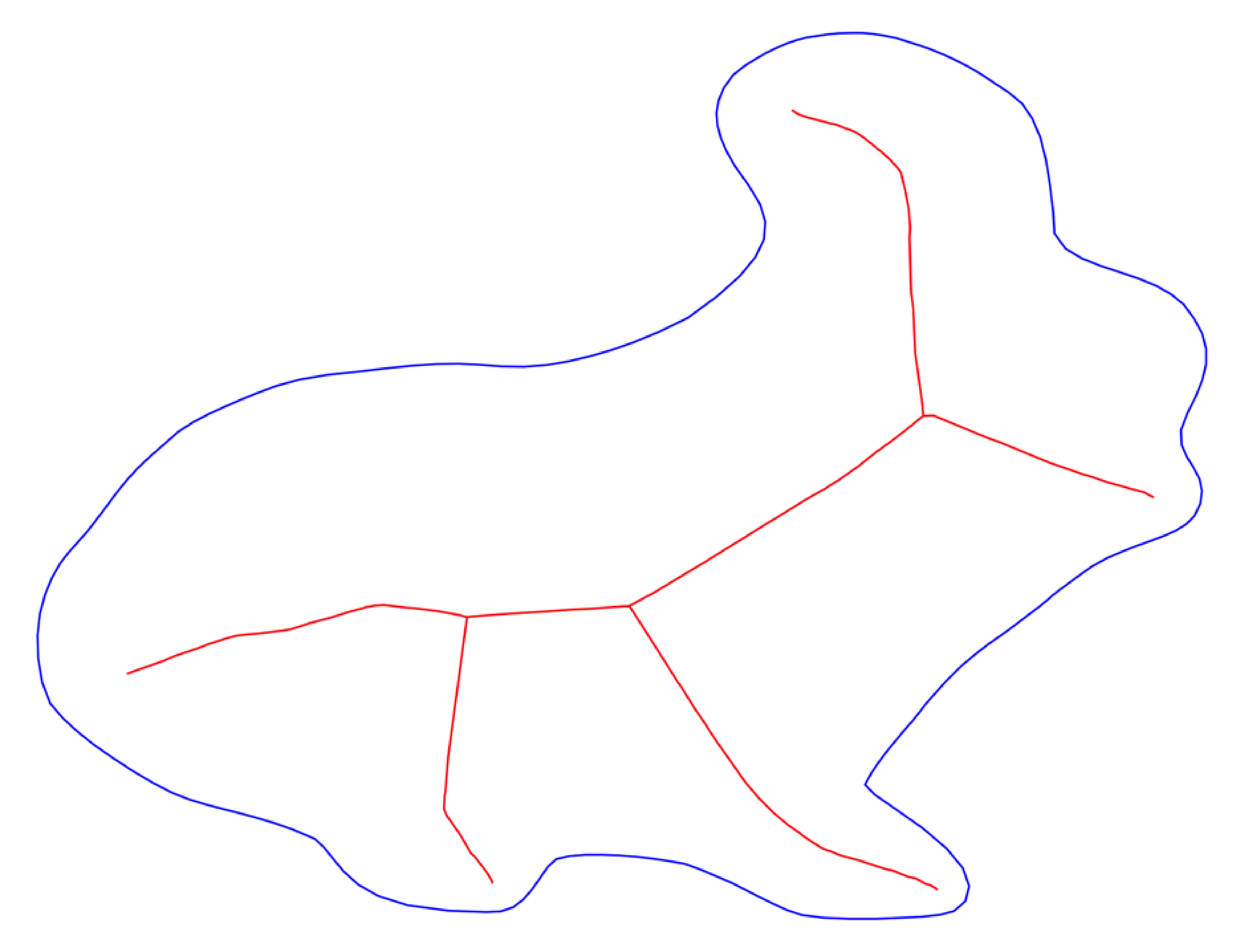}
\\ (g) extracted skeleton \cite{xu:cmame2015}
\end{minipage}
 \begin{minipage}[t]{2.15in}
\centering
\includegraphics[width=2.15in]{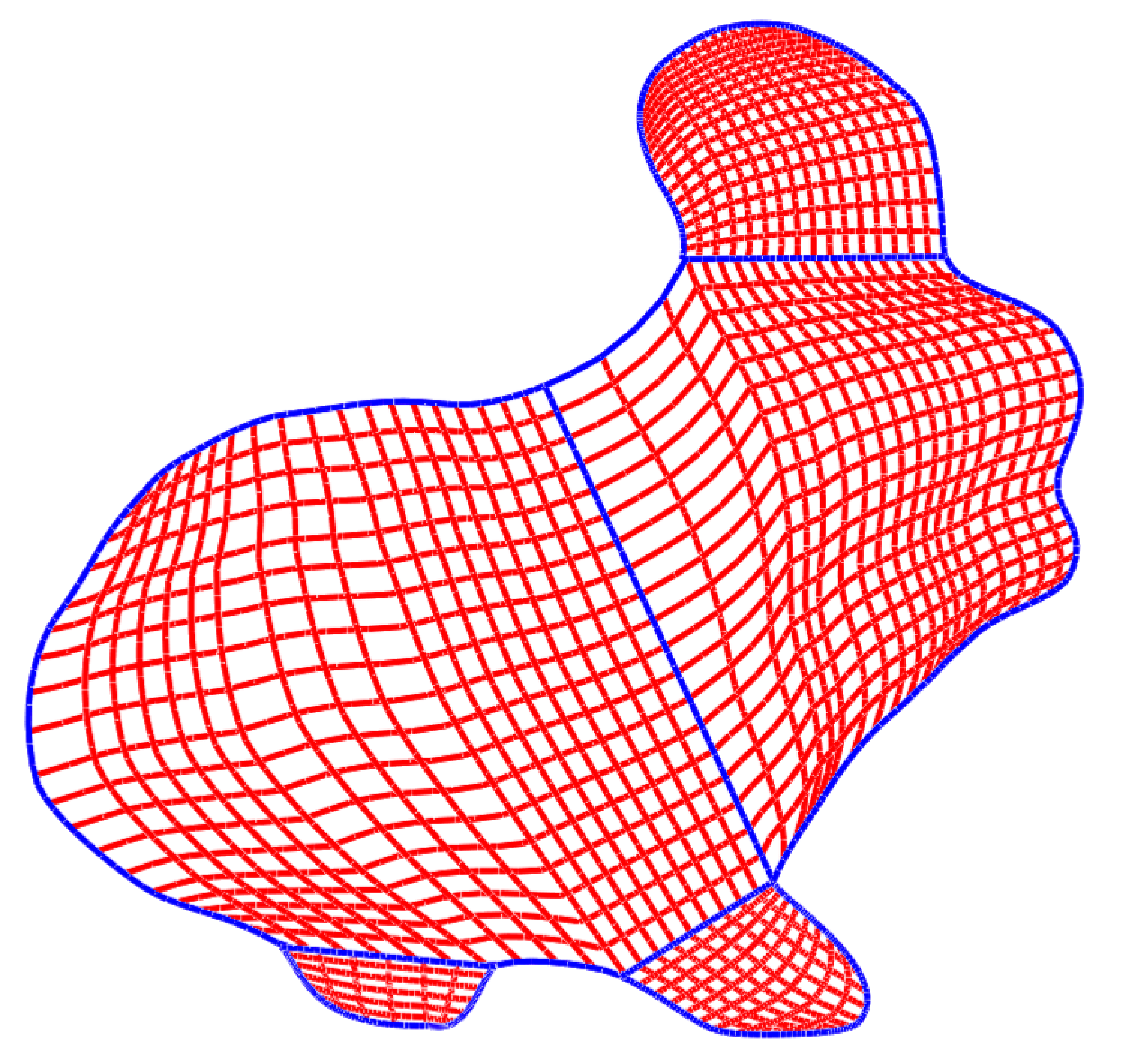}
\\ (h) skeleton-based parameterization \cite{xu:cmame2015}
\end{minipage}
\begin{minipage}[t]{2.1in}
\centering
\includegraphics[width=2.9in]{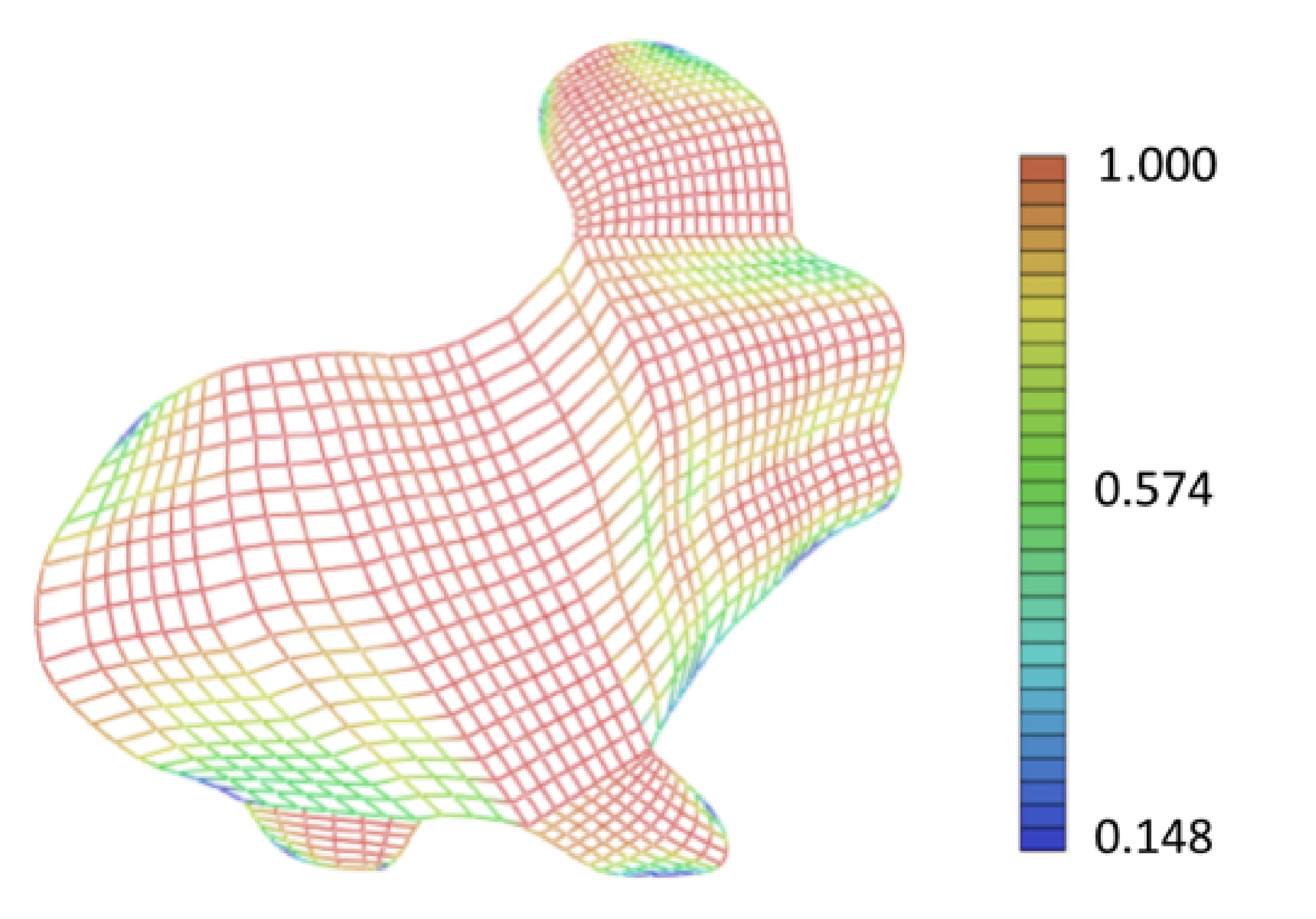}
\\ (i) Jacobian colormap of (h)
\end{minipage}
\caption{Example III.}
\label{fig:example2}
\end{figure}

\begin{figure}
\centering
\begin{minipage}[t]{2.1in}
\centering
\includegraphics[width=2.1in,height=2.3in]{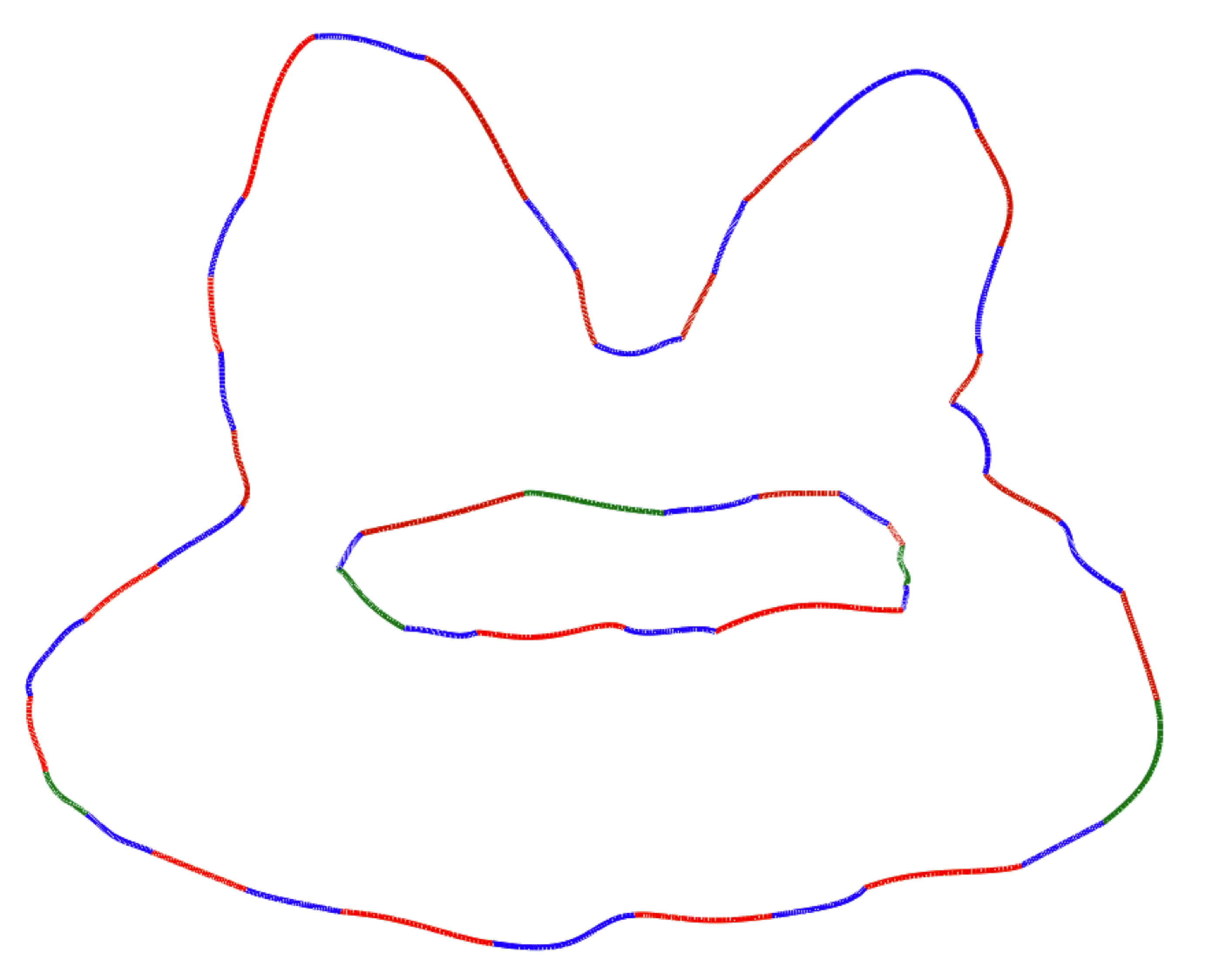}
\\ (a) boundary \Bezier curves
\end{minipage}
\begin{minipage}[t]{2.1in}
\centering
\includegraphics[width=2.1in]{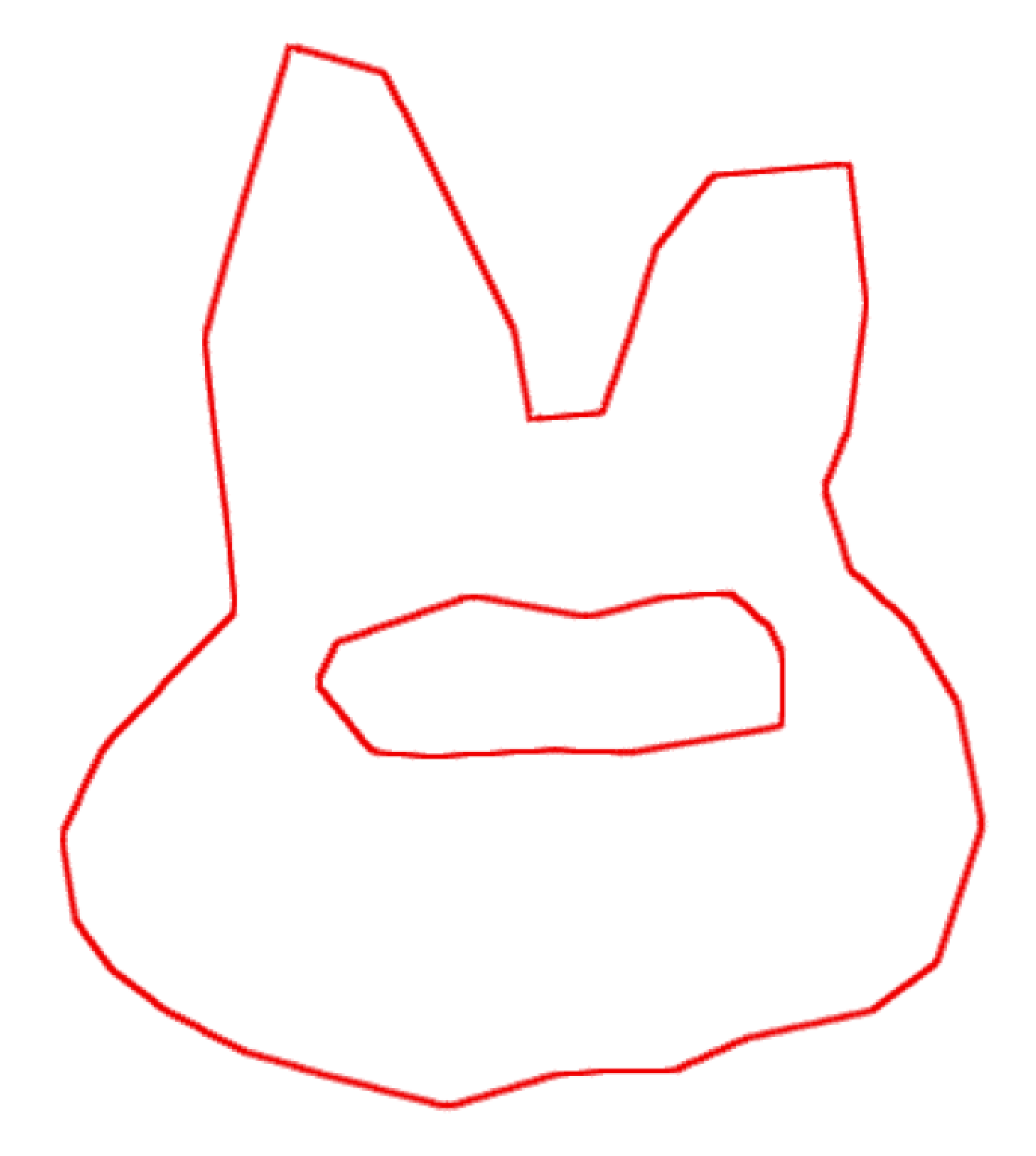}
\\ (b) discrete boundary
\end{minipage}
\begin{minipage}[t]{2.1in}
\centering
\includegraphics[width=2.05in]{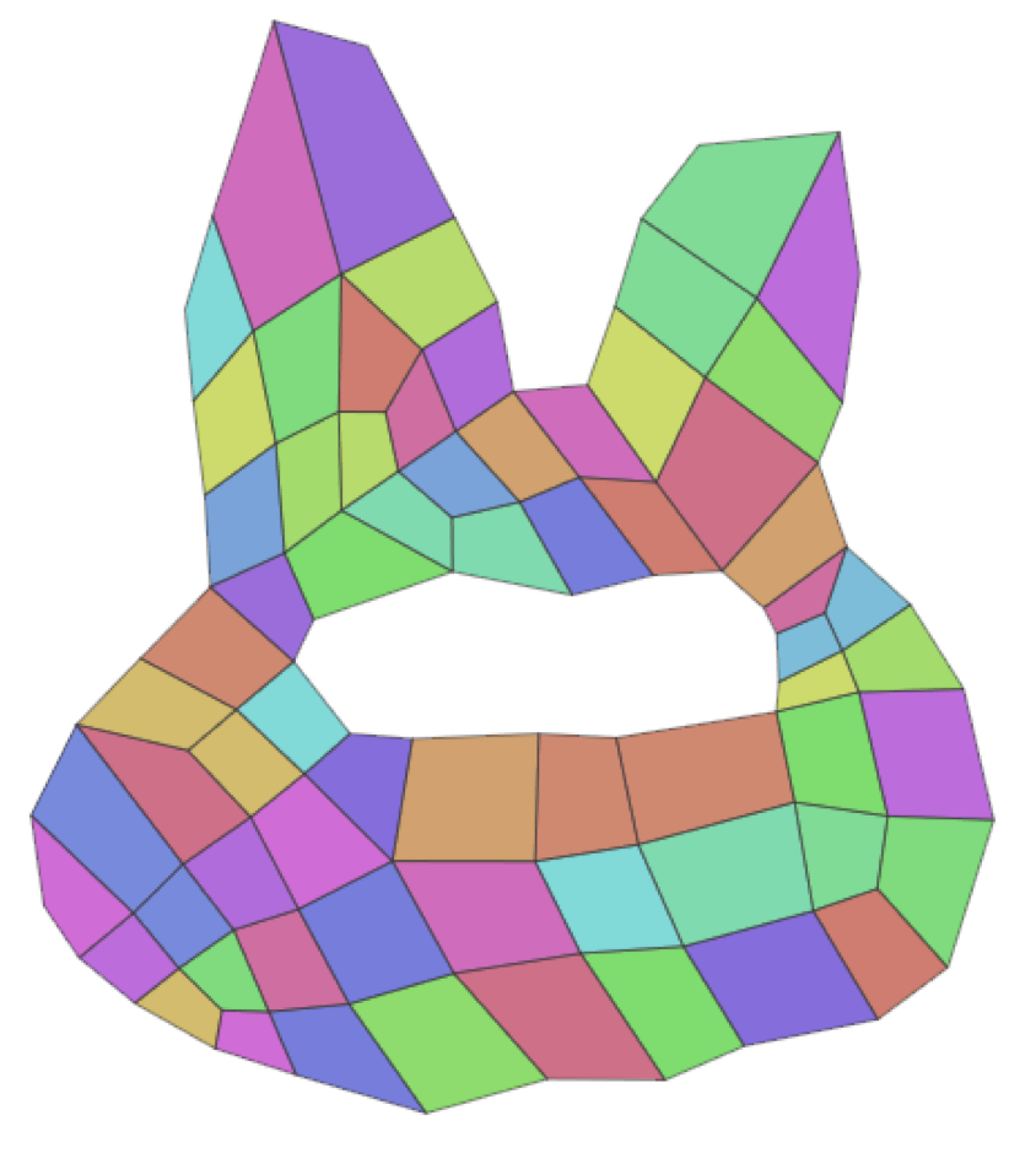}
\\ (c) quad meshing result
\end{minipage}\\
\begin{minipage}[t]{2.1in}
\centering
\includegraphics[width=2.1in]{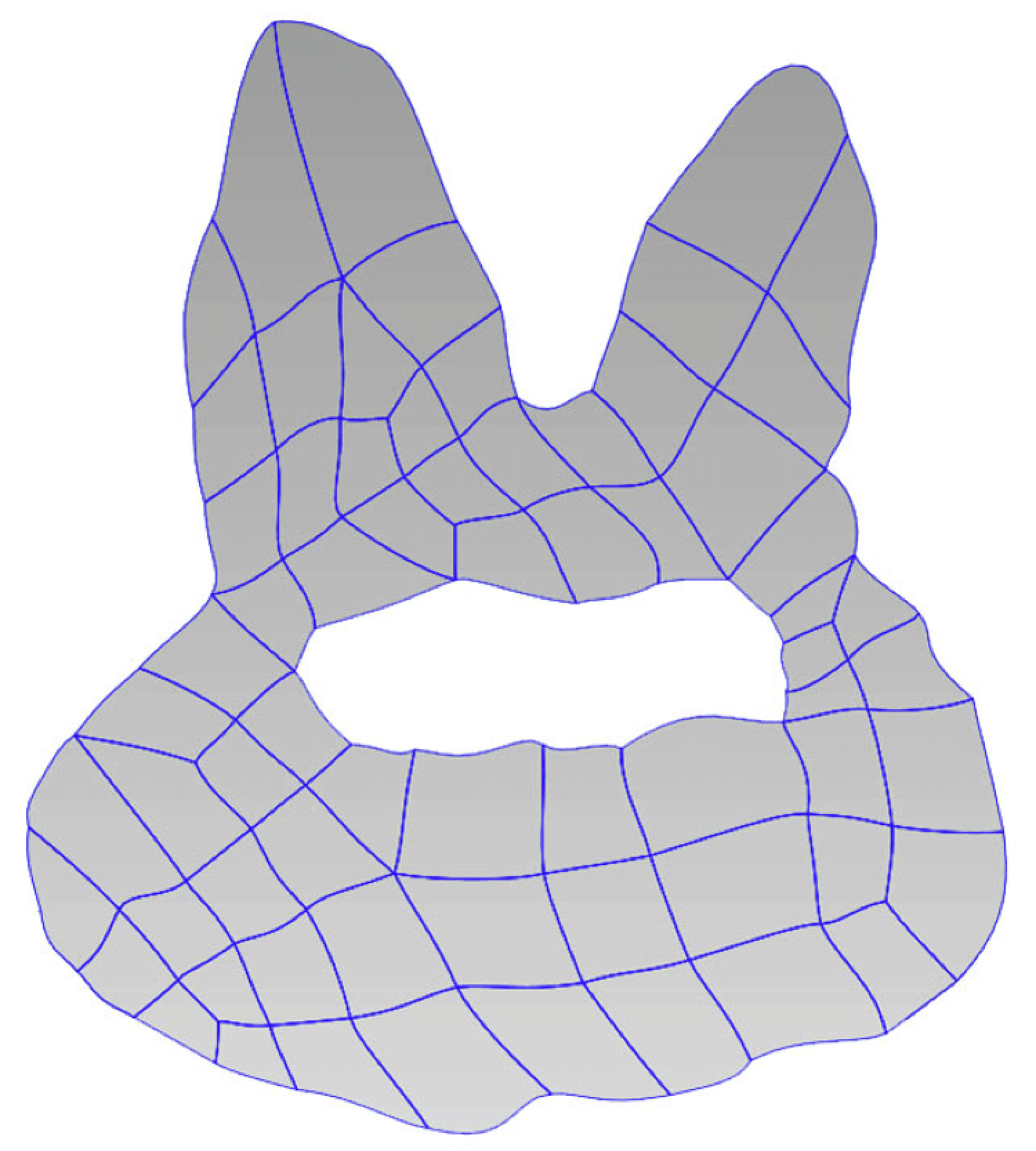}
\\ (d) segmentation curves
\end{minipage}
\begin{minipage}[t]{2.1in}
\centering
\includegraphics[width=2.1in]{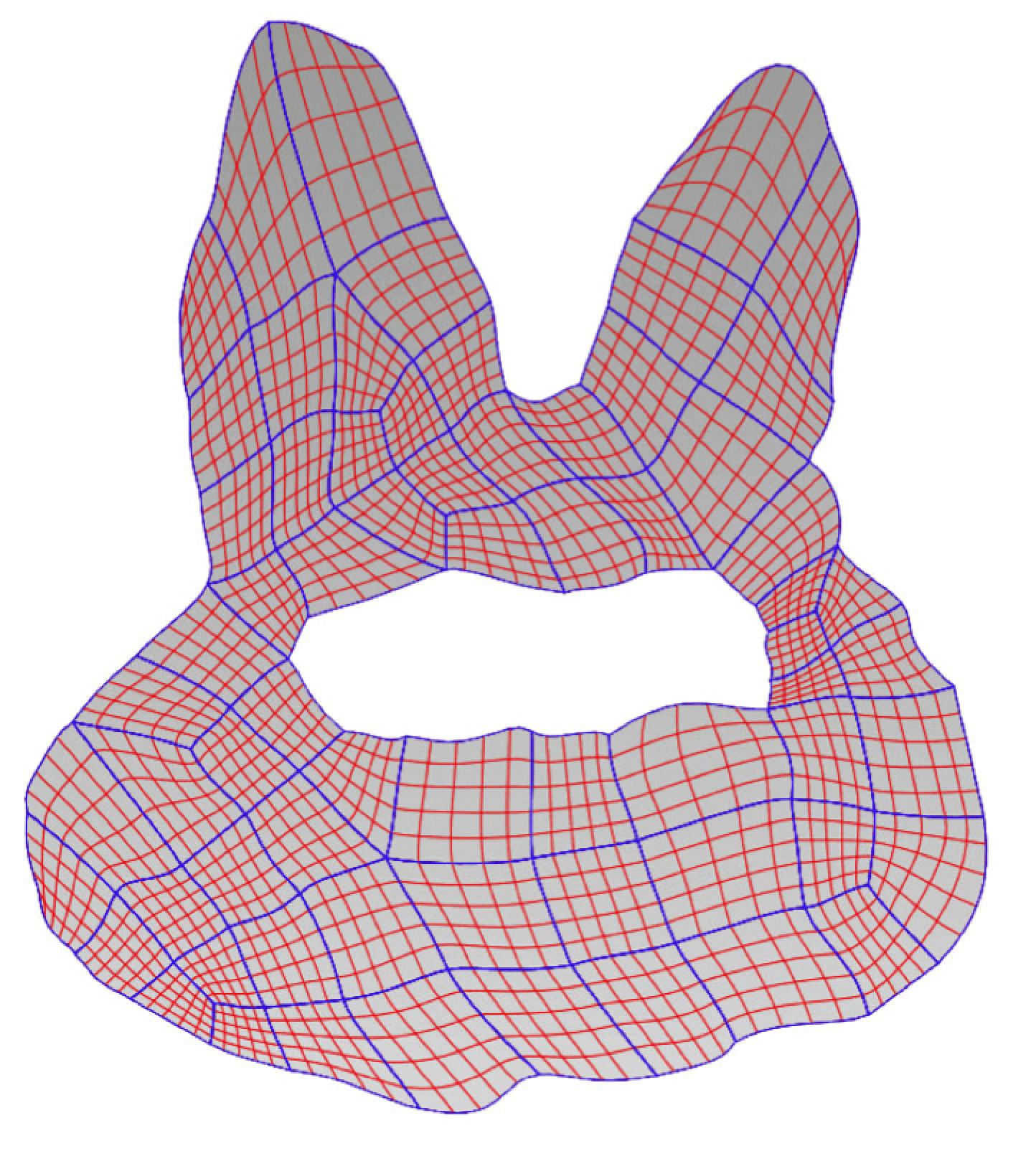}
\\ (e) parameterization result
\end{minipage}
\begin{minipage}[t]{2.1in}
\centering
\includegraphics[width=2.7in]{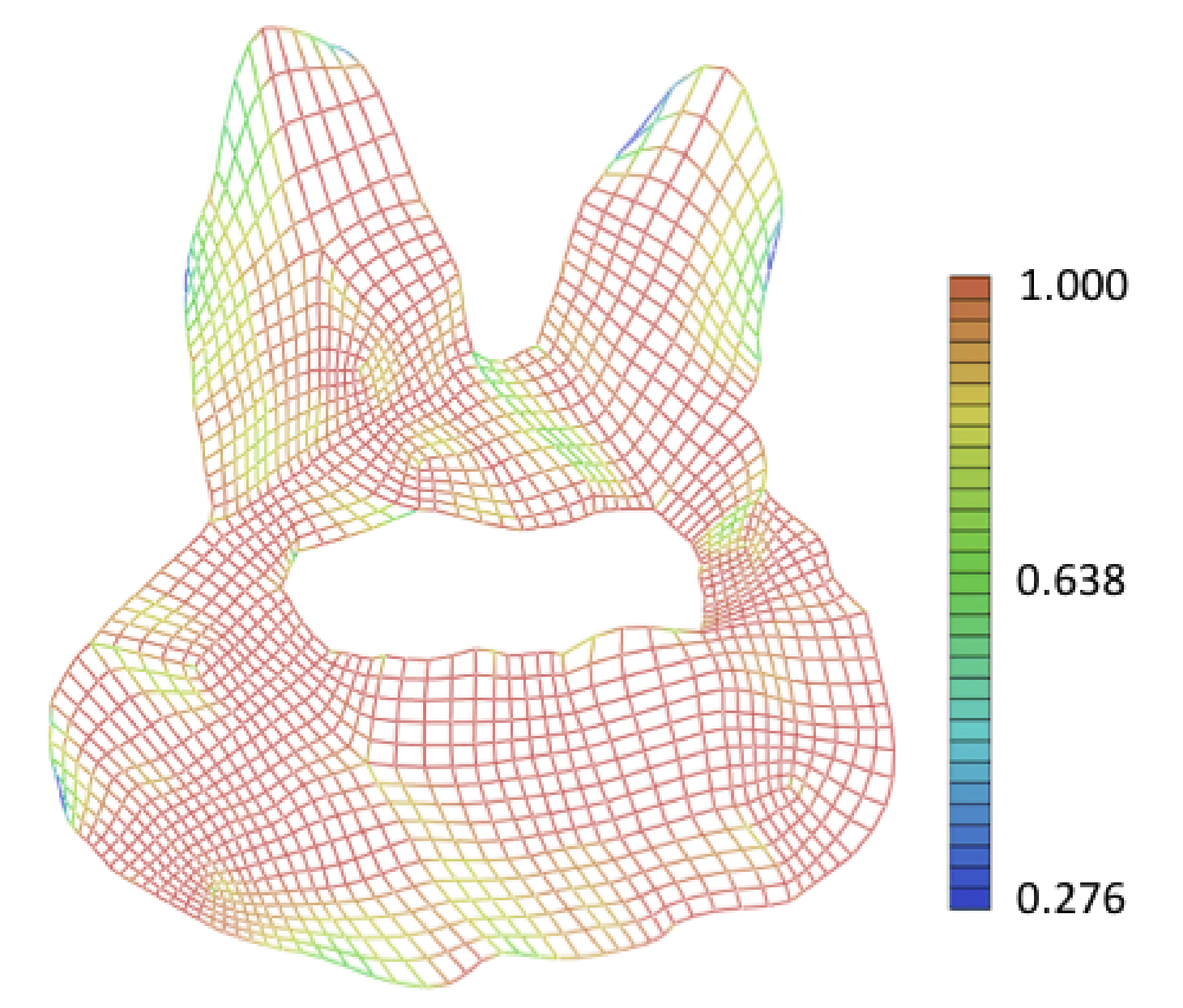}
\\ (f) Jacobian colormap
\end{minipage}\\
\begin{minipage}[t]{2.1in}
\centering
\includegraphics[width=2.05in,height=2.35in]{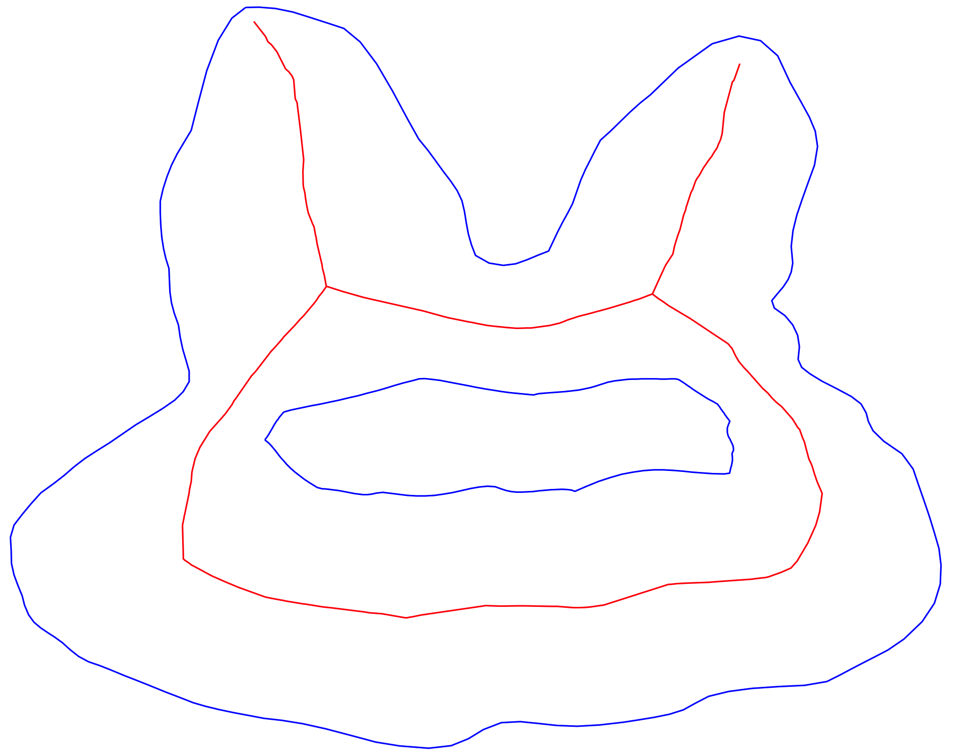}
\\ (g) extracted skeleton \cite{xu:cmame2015}
\end{minipage}
 \begin{minipage}[t]{2.15in}
\centering
\includegraphics[width=2.20in]{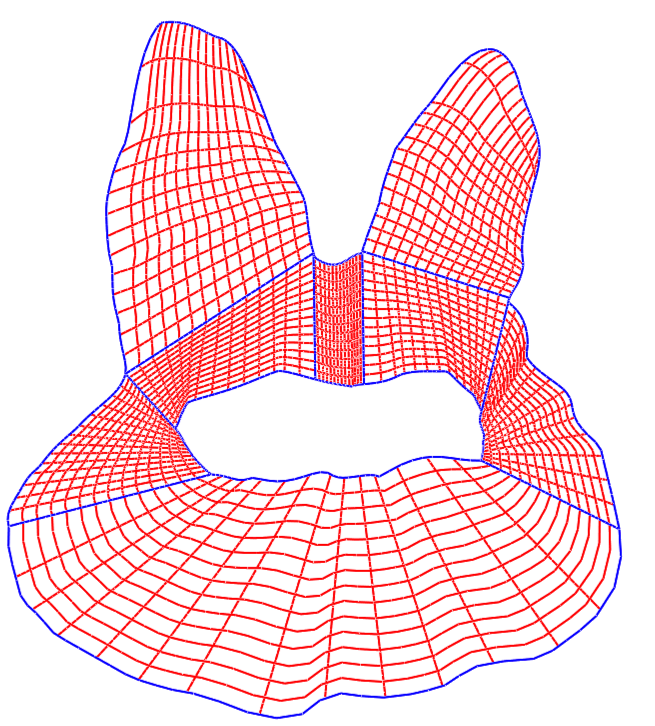}
\\ (h) skeleton-based parameterization \cite{xu:cmame2015}
\end{minipage}
\begin{minipage}[t]{2.1in}
\centering
\includegraphics[width=2.75in]{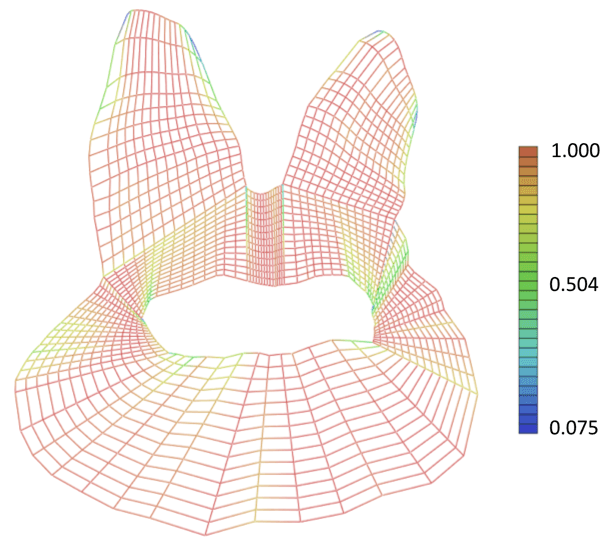}
\\ (i) Jacobian colormap of (h)
\end{minipage}
\caption{Example IV. }
\label{fig:example3}
\end{figure}

\begin{figure}
\centering
\begin{minipage}[t]{2.3in}
\centering
\includegraphics[width=2.3in ]{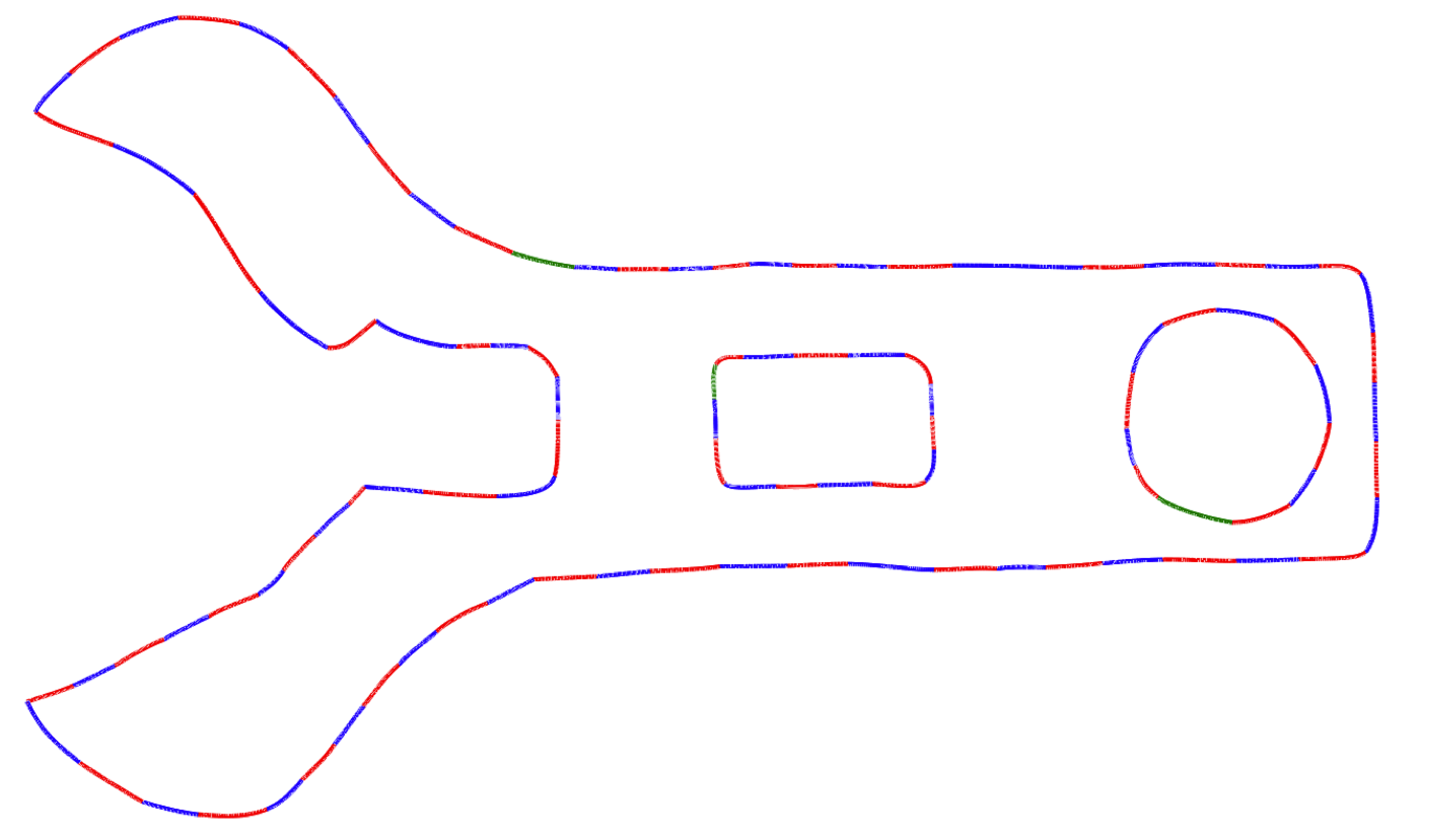}
\\ (a) boundary \Bezier curves
\end{minipage}    \quad
\begin{minipage}[t]{2.25in}
\centering
\includegraphics[width=2.3in ]{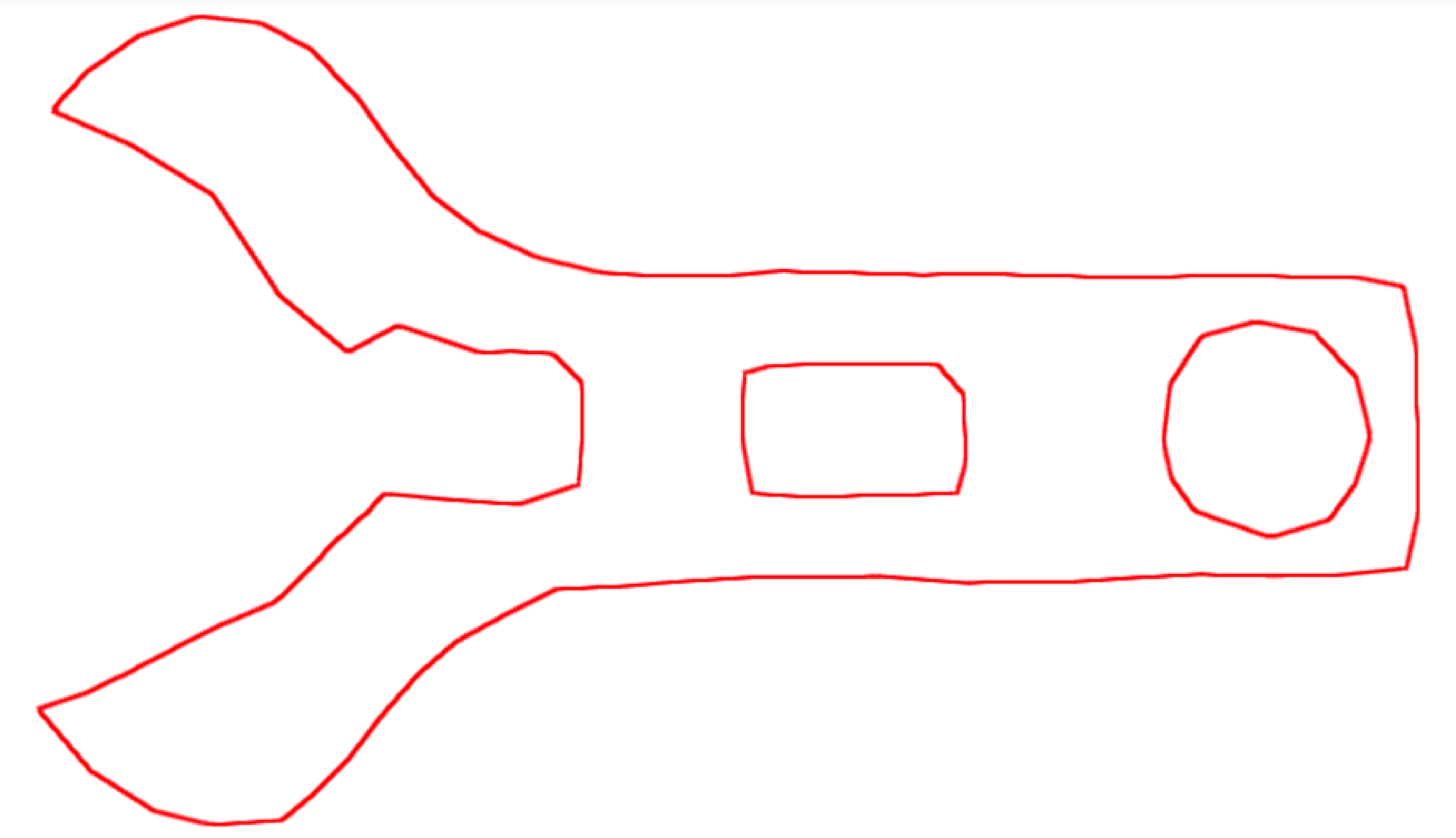}
\\ (b) discrete boundary
\end{minipage} \\
\begin{minipage}[t]{2.3in}
\centering
\includegraphics[width=2.3in ]{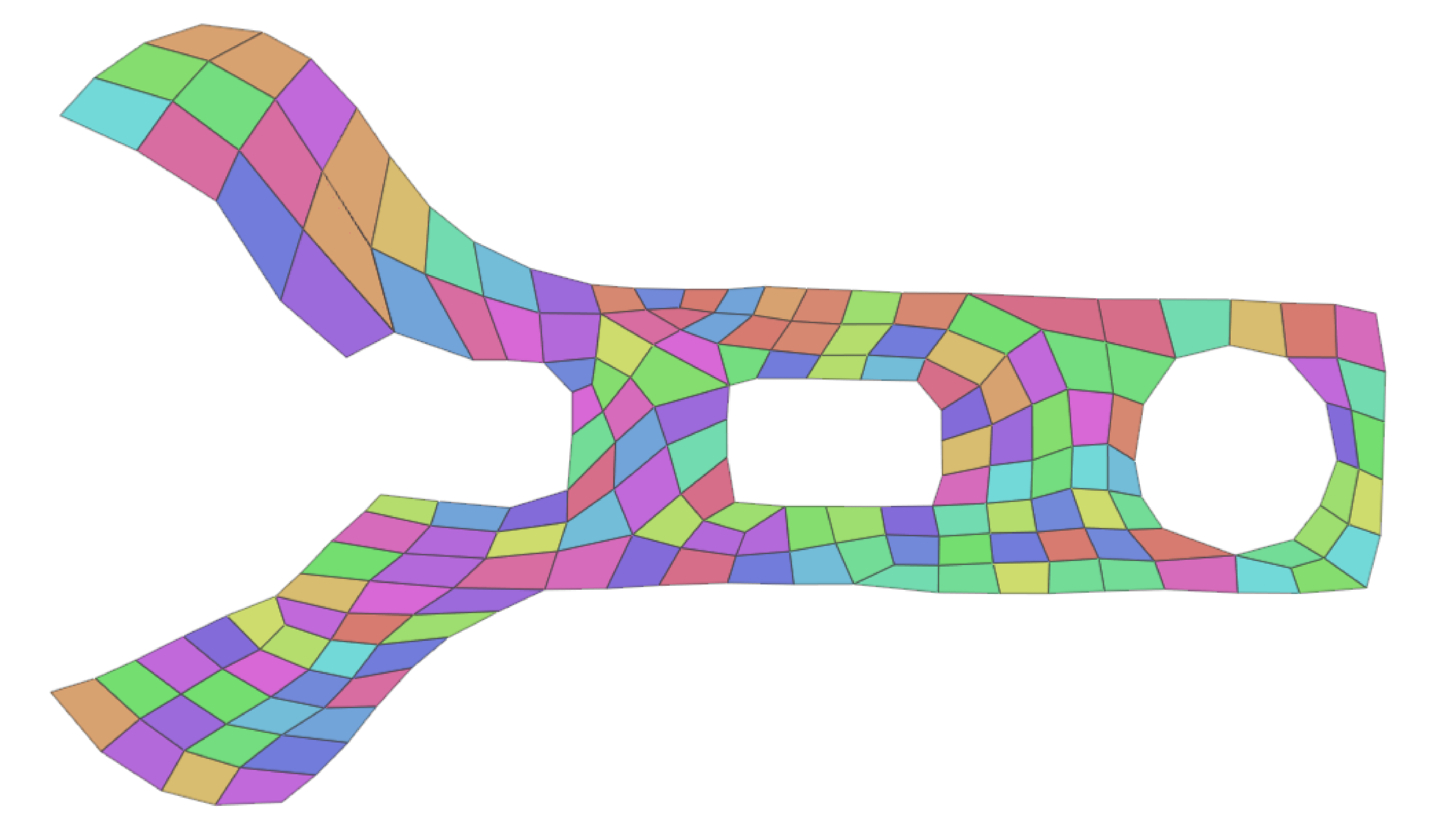}
\\ (c) quad meshing result
\end{minipage} \quad \quad
\begin{minipage}[t]{2.25in}
\centering
\includegraphics[width=2.25in]{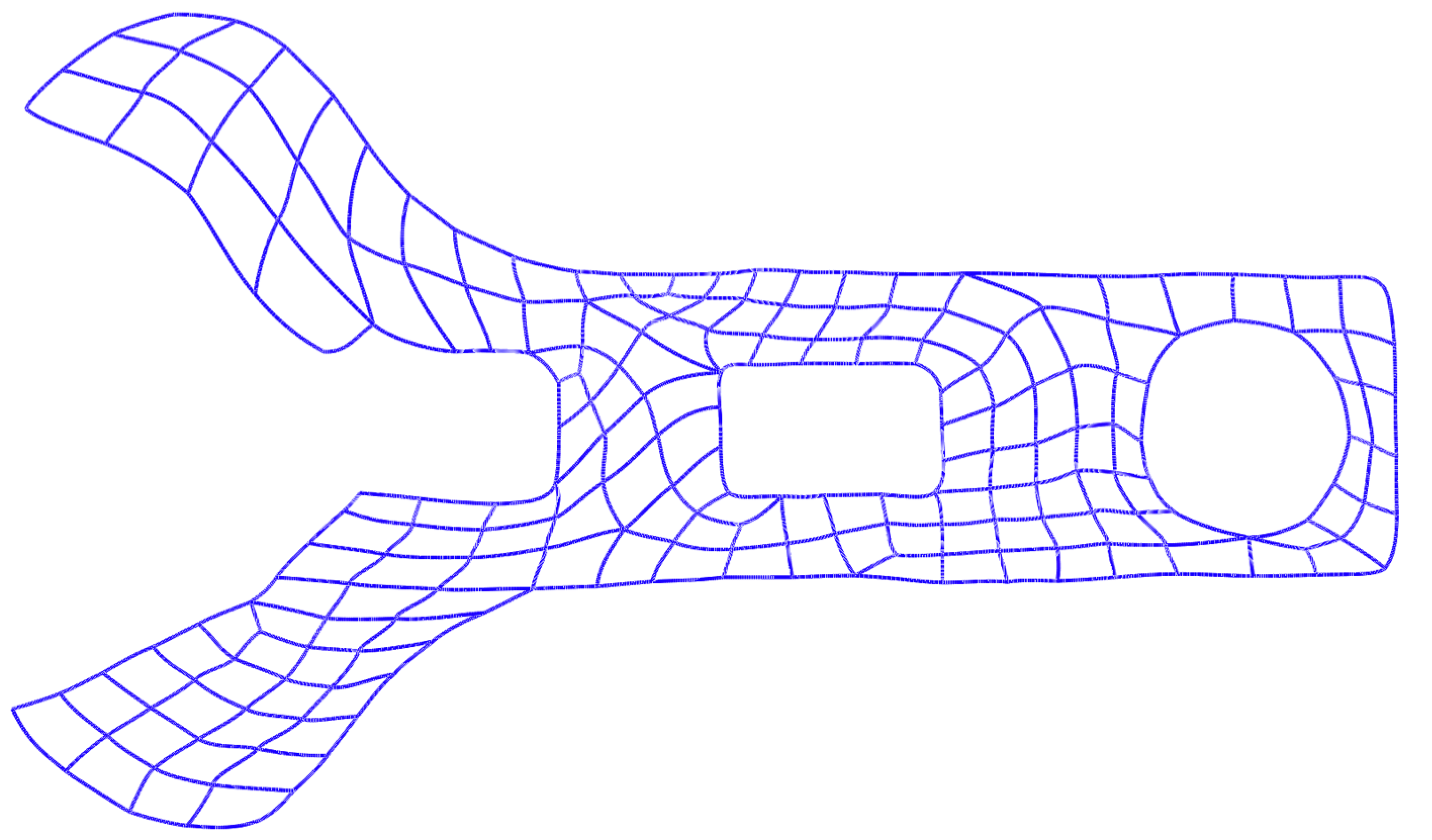}
\\ (d) segmentation curves
\end{minipage}\\
\begin{minipage}[t]{2.2in}
\centering
\includegraphics[width=2.25in]{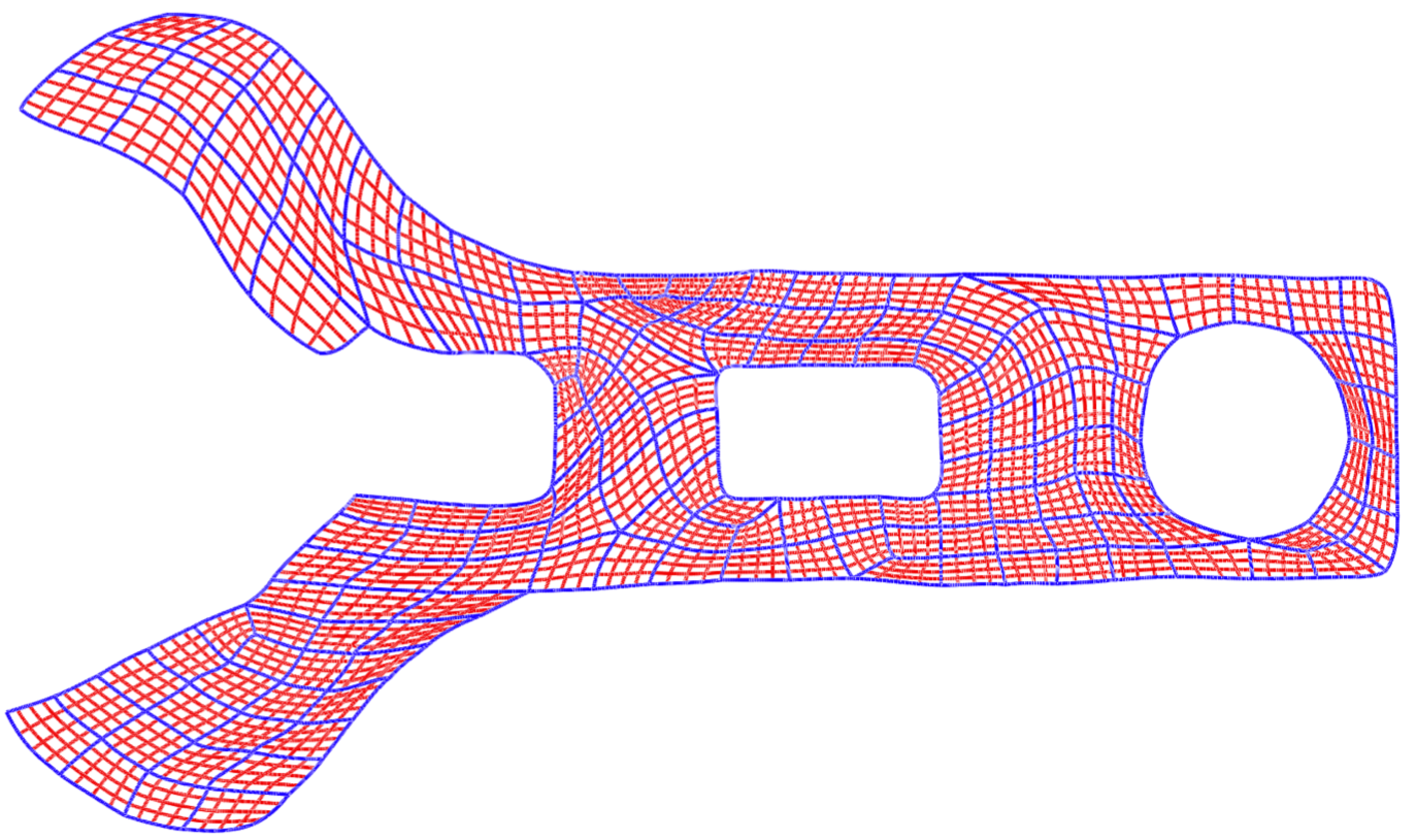}
\\ (e) parameterization result
\end{minipage} \qquad \quad
\begin{minipage}[t]{2.4in}
\centering
\includegraphics[width=2.86in]{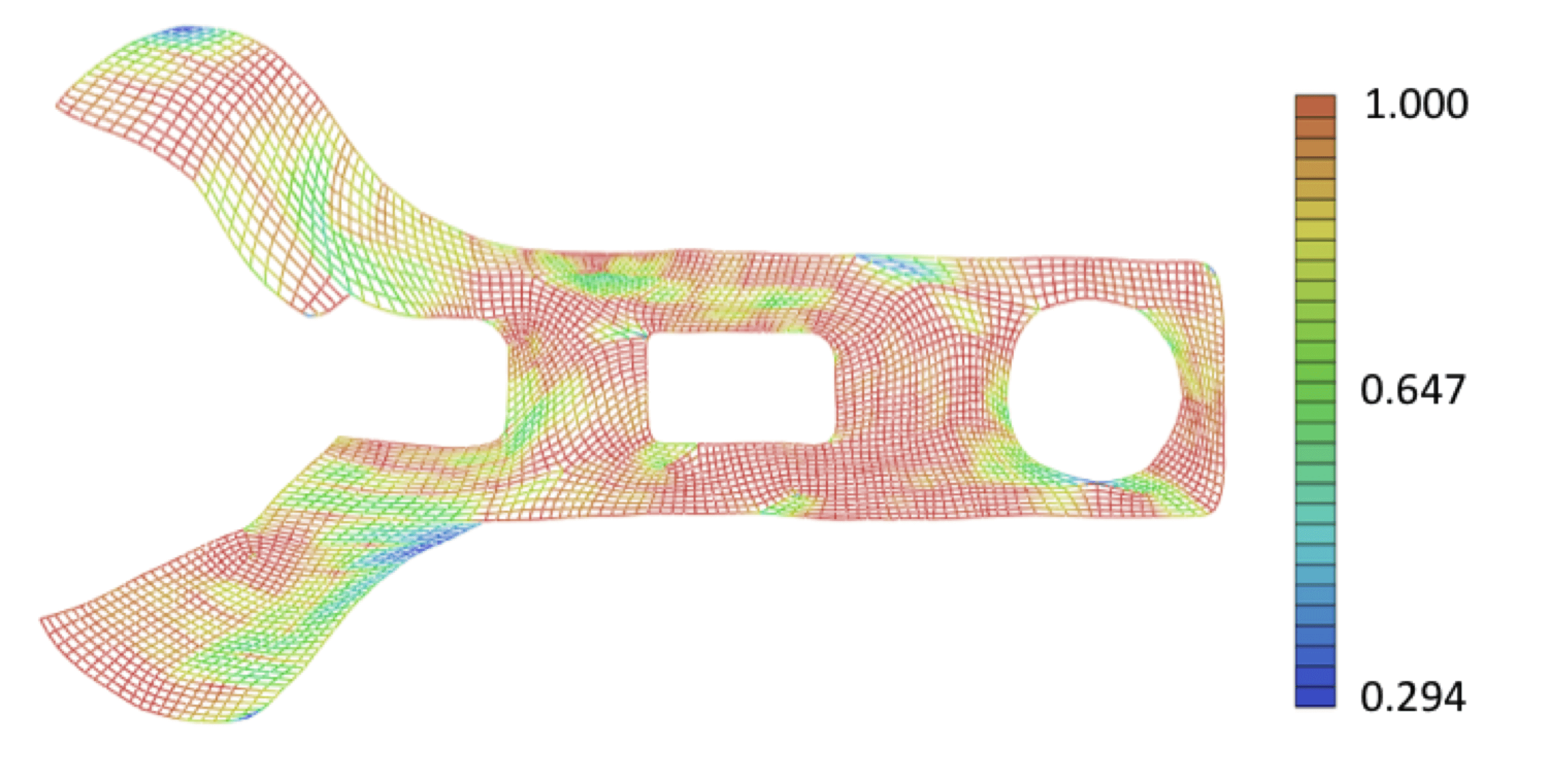}
\\ (f) Jacobian colormap
\end{minipage}\\
\begin{minipage}[t]{2.5in}
\centering
\includegraphics[width=2.48in,height=1.35in]{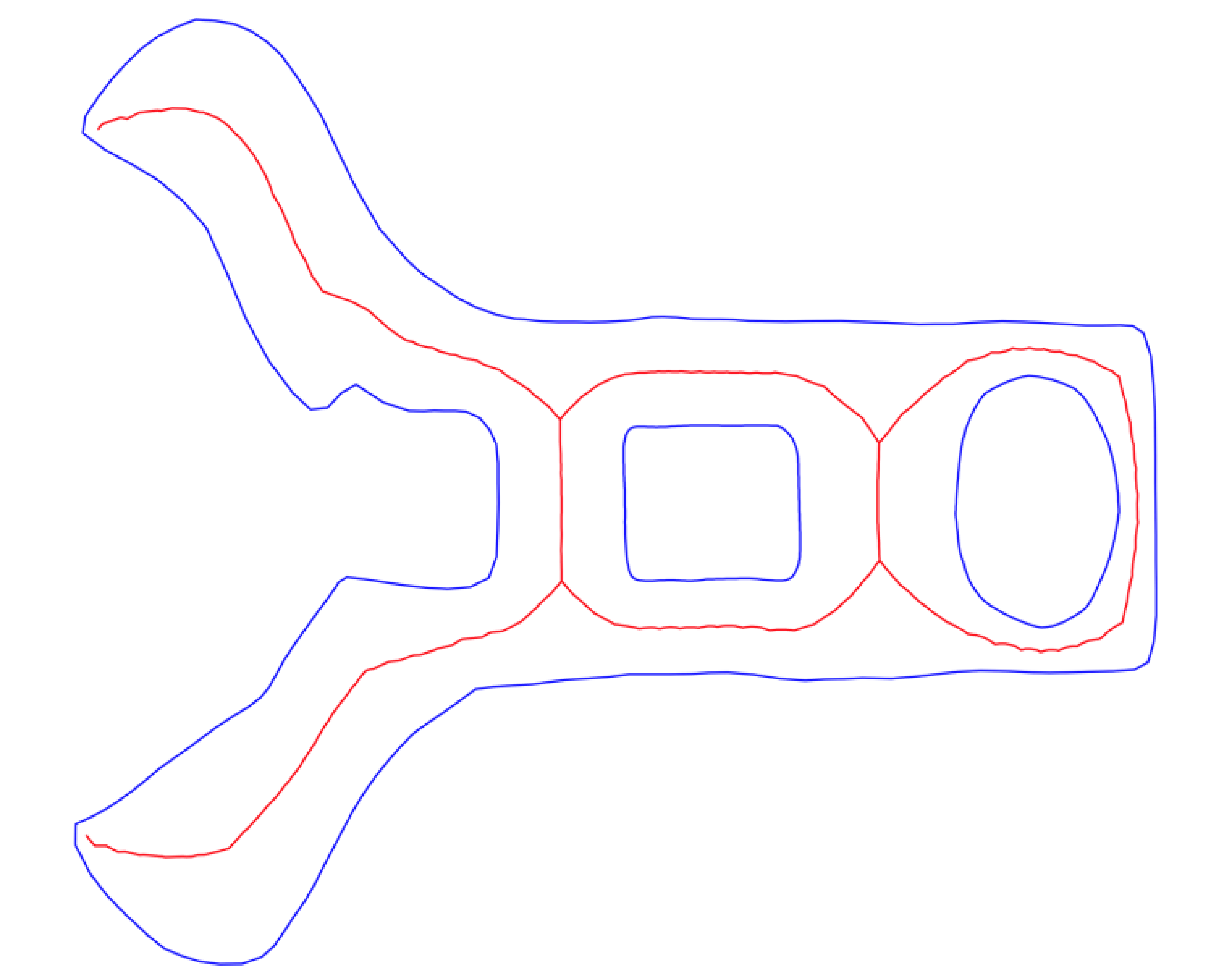}
\\ (g) extracted skeleton \cite{xu:cmame2015}
\end{minipage}
\begin{minipage}[t]{2.7in}
\centering
\includegraphics[width=2.35in,height=1.35in]{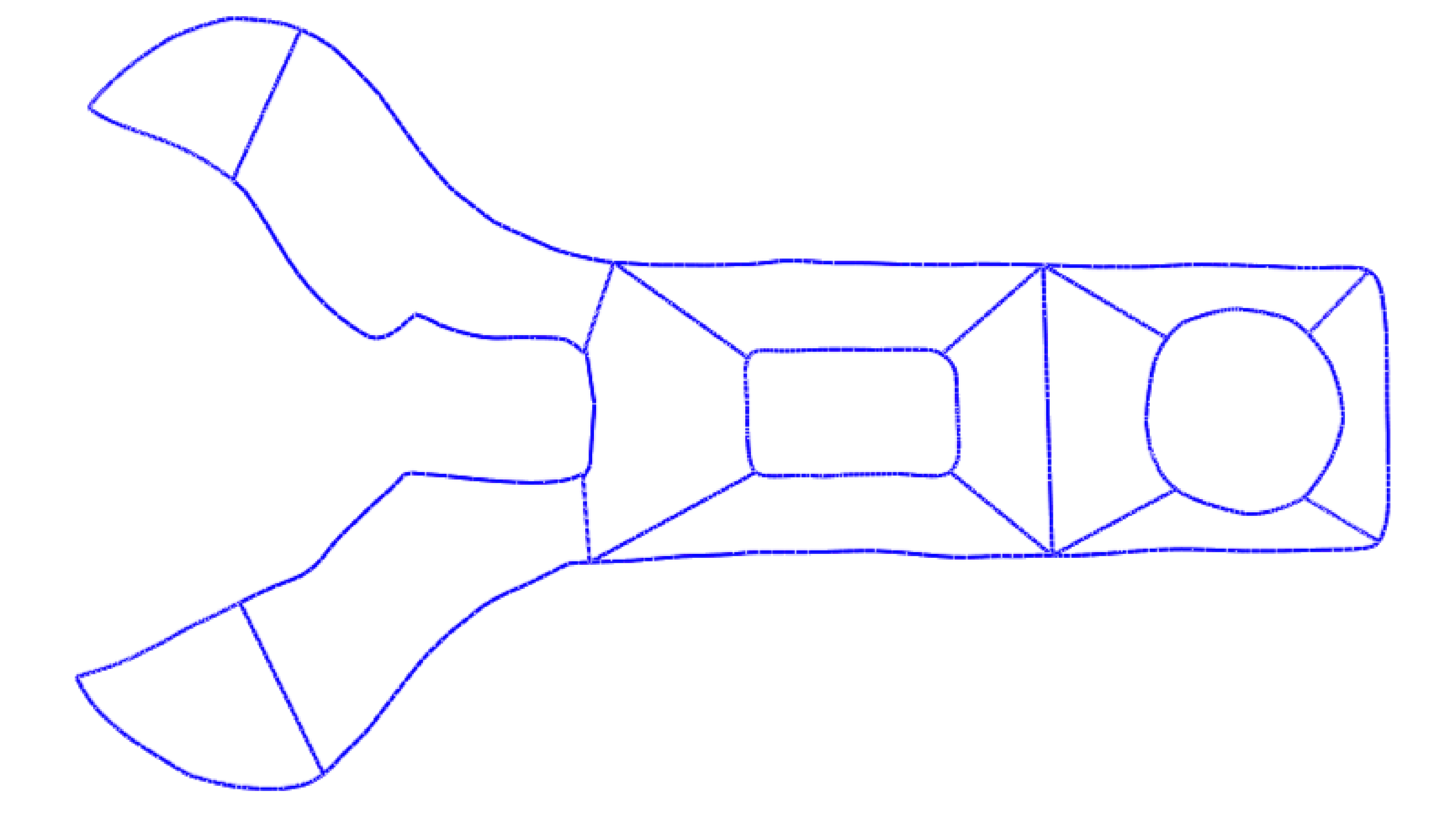}
\\ (h) skeleton-based domain partition\cite{xu:cmame2015}
\end{minipage}\\
 \begin{minipage}[t]{2.7in}
\centering
\includegraphics[width=2.3in]{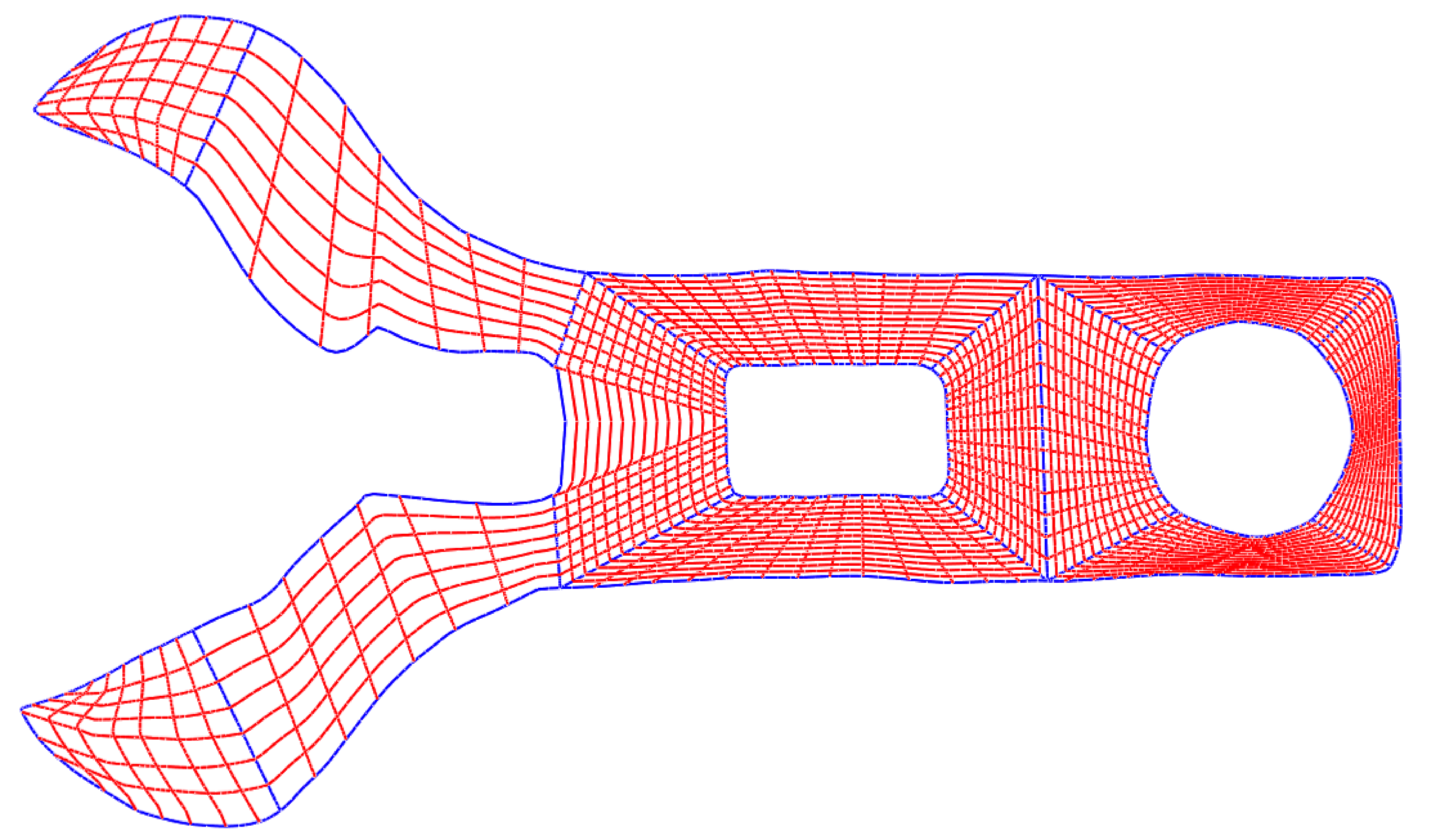}
\\ (i) skeleton-based parameterization \cite{xu:cmame2015}
\end{minipage}
\begin{minipage}[t]{2.4in}
\centering
\includegraphics[width=2.7in]{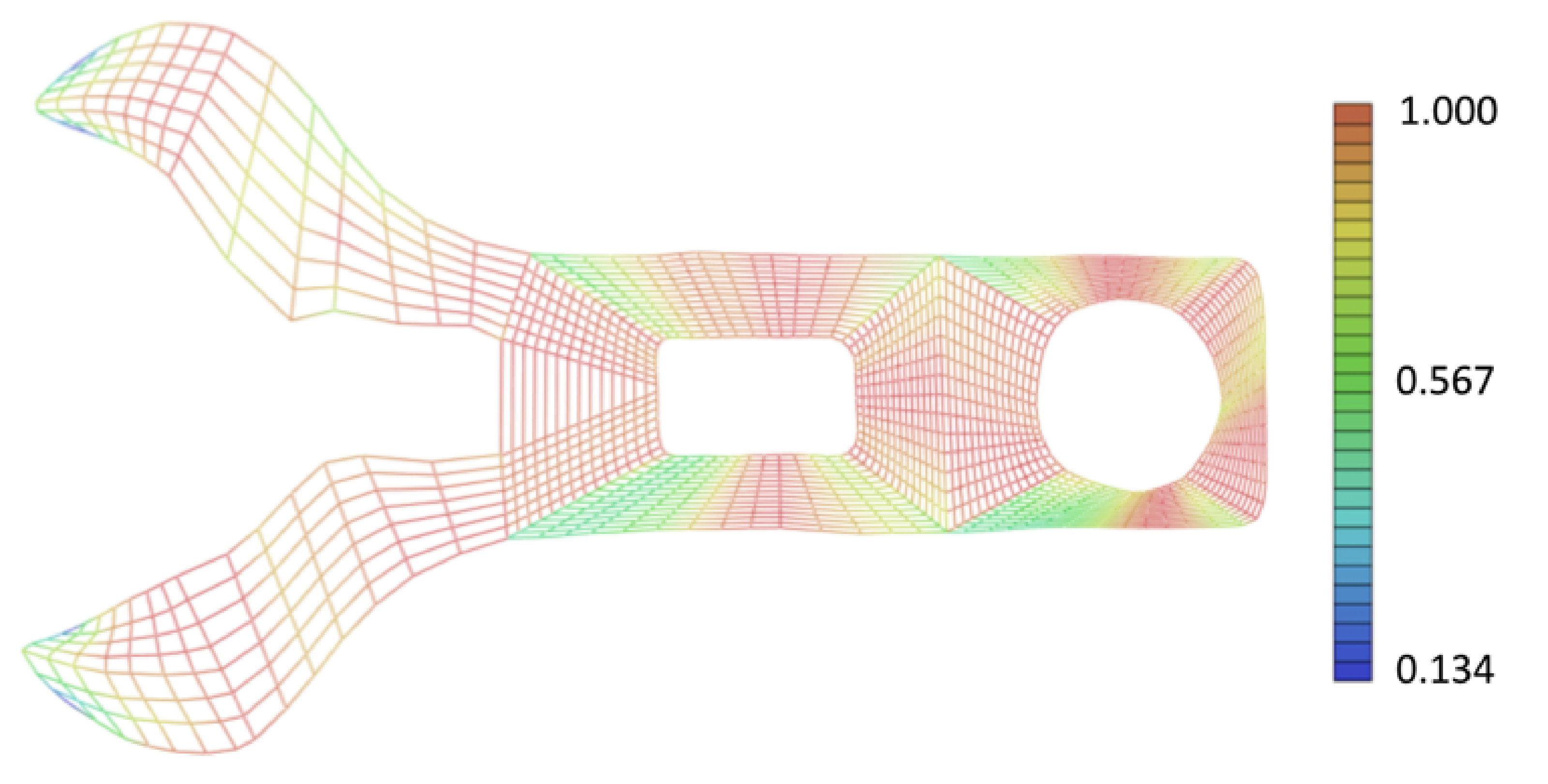}
\\ (j) Jacobian colormap of (i)
\end{minipage}
\caption{Example V.}
\label{fig:spanner}
\end{figure}

\subsection{Comparison with skeleton-based parameterization method}

In order to show the effectiveness of the proposed approach,
three examples are presented to compare the proposed method with the
skeleton-based domain decomposition method \cite{xu:cmame2015}.

Fig. \ref{fig:example2} (a) depicts the boundary \Bezier curves after
the pre-processing of input boundary B-spline curves;
Fig. \ref{fig:example2} (b) presents the discrete boundary obtained by
connecting the ending control points of each \Bezier curve, and the
corresponding quad meshing result can be found in
Fig. \ref{fig:example2} (c). Fig. \ref{fig:example2} (d) shows the
domain partition results with the construction of segmentation curves
by the global optimization method. After the local optimization for
each sub-patch, the final planar parameterization is illustrated in
Fig. \ref{fig:example2} (e). The iso-parametric curves show the quality of the planar parameterization. The corresponding
scaled Jacobian colormap in Fig. \ref{fig:example2}
(f) proves that the parametrization is analysis suitable. Fig. \ref{fig:example2} (g) presents the extracted skeleton from
the input boundary, Fig. \ref{fig:example2} (h) the
skeleton-based parameterization results with blue segmentation curves,
and  Fig. \ref{fig:example2} (i) the corresponding scaled
Jacobian colormap. More comparison examples with complex geometry can be found in
Fig. \ref{fig:example3} and Fig. \ref{fig:spanner}. Since  global/local optimizations with $C^1/G^1$ constraints are applied for the construction of
segmentation curves and parameterization of \Bezier subdomains, our method  achieves smooth parameterizations compared
to the skeleton-based approach \cite{xu:cmame2015}.

\begin{table*}[t]
  \caption{Quantitative data for planar parameterization in
    Fig. \ref{fig:example2}, Fig. \ref{fig:example3} and
    Fig. \ref{fig:spanner}. $p$: degree
    of planar parameterization; \#
    SD: number of subdomains by domain decomposition;\#
    Patch: number of \Bezier patches; \# Con.: number of control
    points.  } 
  \centering 
  \begin{tabular}{c| c | c c c c | c c c c c  } 
     \hline 
\centering
  \multirow{2}{*} {Example} & \multirow{2}{*} {Method}& \multirow{2}{*} {$p$} &
  \multirow{2}{*} {\# SD} & \multirow{2}{*} {\#Patch} &
  \multirow{2}{*} {\# Con.}
      & \multicolumn{2}{c}{Scaled Jacobian} &
      &\multicolumn{2}{c}{Condition number }\\
\cline{7-8} \cline{10-11} &&&& & &  Average & Min &  & Average & Max     \\ [0.5ex]   
    \hline\hline 
    \multirow{2}{3em}{Fig. \ref{fig:example2} } & Our method &$6$ & 39   & 39 & 1467 & 0.8843 & 0.292 &
     & 2.76 &  8.06 \\ 
     & Xu et al.\cite{xu:cmame2015} &$6$ & 5   & 35 & 1309  & 0.5172 & 0.148 &
     & 5.36 & 16.31 \\
     \hline
     \multirow{2}{4em}{Fig. \ref{fig:example3} }  & Our method & $5$ &  66 & 66 &  1768  & 0.9194&
     0.276  & & 2.42 &  10.18 \\ 
    &  Xu et al.\cite{xu:cmame2015}  & $5$ &  8 & 56  &  1507  & 0.7801&
     0.075  & & 4.35 &  18.23 \\
     \hline
     \multirow{2}{4em}{Fig. \ref{fig:spanner} }   & Our method  & $5$ & 155 & 155   & 3720  & 0.9017
     &0.294 &  & 2.57 & 7.86 \\ 
&Xu et al.\cite{xu:cmame2015} & $5$ & 12 & 132   & 3282   & 0.7894
     &0.134 &  & 4.23 & 15.64 \\
    \hline 
  \end{tabular}
  \label{table:rrefine} 
\end{table*}

\begin{table}[t]
  \caption{Parameters of the objective functions and the computing times (in seconds) for the
    examples  in
   Fig. \ref{fig:example1},
    Fig. \ref{fig:exam5},  Fig. \ref{fig:example2}, Fig.
    \ref{fig:example3} and  Fig. \ref{fig:spanner}.
    \# $T_1$: computing time for the global optimization ;   \# $T_2$:
    computing time for the local optimization;
      \# $T$: total computing time. } 
  \centering 
  \begin{tabular}{c |c c c c c c c c c c |c c c} 
     \hline 
\centering
  \multirow{2}{*} {Example}  & \multicolumn{2}{c}{$ F_{\text{shape}}$ in (\ref{eq:shape})} &
      &\multicolumn{3}{c}{$F$ in (\ref{eq:curveobj})} & &\multicolumn{2}{c}{
        $E(\emph{\textbf{r}}) $ in (\ref{eq:quasifunc}) } & &\multirow{2}{*} {\# $T_1$} &\multirow{2}{*} {\# $T_2$}&\multirow{2}{*} {\# $T$}\\
\cline{2-3} \cline{5-7} \cline{9-10} &$\sigma_1$ &$\sigma_2$&&
$\omega_1$  & $\omega_2$ & $\omega_3$ &  & $\tau_1$ & $\tau_2$  & &  & &  \\ [0.5ex]   
\hline
    \hline 
   Fig. \ref{fig:example1} & 2.0 &  1.0 &     & 2.0 & 1.0&
     50.0  & &  2.0 & 1.5 & & 73.22 & 177.86 & 251.08 \\ 
     Fig. \ref{fig:exam5}(c)   & 1.0 &  1.0 &     & 1.0 & 1.0&
     50.0  & &  1.0 & 1.5 & & 22.68 & 26.64 & 49.32\\ 
     Fig. \ref{fig:exam5}(g)& 1.0 &  1.0 &     & 1.0 & 1.0&
     50.0  & &  1.0 & 1.5 & & 22.96 & 27.18 & 50.14 \\
     Fig. \ref{fig:exam5}(k)& 1.0 &  1.0 &     & 1.0 & 1.0&
     50.0  & &  1.0 & 1.5 & & 27.68 & 36.32 &  63.90  \\
 Fig. \ref{fig:example2} & 1.0 &  2.0 &     & 1.0 & 2.0&
     50.0  & &  1.0 & 2.0 & & 32.74 & 53.84 &  86.58 \\
      Fig. \ref{fig:example3} & 2.0 &  1.0 &     & 1.0 & 2.0&
     50.0  & &  2.0 & 1.0 & & 41.02 & 61.36 &  102.38 \\
    Fig. \ref{fig:spanner}& 2.0 &  1.0 &     & 2.0 & 2.0&
     50.0  & &  2.0 & 1.0 & & 75.70 & 209.68 & 285.38\\
    \hline 
  \end{tabular}
  \label{table:time} 
\end{table}

Quantitative data of the comparison examples including the parameterization quality metrics
with respect to scaled Jacobians and condition numbers presented in
Fig. \ref{fig:example2}, Fig. \ref{fig:example3} and Fig. \ref{fig:spanner} are summarized in
Table \ref{table:rrefine}. We can see from Table
\ref{table:rrefine} that the proposed approach
achieves analysis-suitable parameterizations with bigger average scaled
Jacobians and smaller average condition numbers compared to the
skeleton-based method \cite{xu:cmame2015}.  In the comparison examples, the input boundary curves are usually B-spline curves, hence after the domain partition by skeleton-based method \cite{xu:cmame2015}, knot-insertion operations should be performed on each sub-domain to obtain extracted \Bezier patches.  Although the number of subdomains in the skeleton-based method
 \cite{xu:cmame2015} is much smaller than that
in our proposed method, see Table \ref{table:rrefine}, the number of \Bezier patches and the
number of control points in the skeleton-based approach \cite{xu:cmame2015}
are similar to those in the proposed method. These extracted \Bezier
patches are used as the computational elements in IGA \cite{Borden}.

There are several parameters involved in the objective functions
of the optimization framework. Different parameterization
results can be obtained from different choices of these parameters.
The parameters of the objective functions $ F_{\text{shape}}$ in
(\ref{eq:shape}),
$F$ in (\ref{eq:curveobj}) and $E(\emph{\textbf{r}}) $ in
(\ref{eq:quasifunc})   are summarized in Table
\ref{table:time} for the presented examples in
   Fig. \ref{fig:example1},
    Fig. \ref{fig:exam5},  Fig. \ref{fig:example2}, Fig.
    \ref{fig:example3} and  Fig. \ref{fig:spanner}.
Furthermore, the associated CPU times for the global
optimization and local optimization process are listed in Table
\ref{table:time}. The computational cost
of the proposed method depends on the number of patches and the number
of control points in the planar parameterization.

In conclusion, the planar parameterization obtained by the proposed
global/local optimization method has high-quality, and
is suitable for isogeometric applications.

\section{Conclusion and future work}

\label{sec:conclude}

The parameterization of the computational domain with
complex CAD boundary is a key step in IGA.
In this paper, a general framework  for constructing
IGA-suitable planar B-spline parameterizations of the computational
domain with high genus and more complex boundary curves is proposed.  High-quality
patch-partition results with few singularities are achieved by a
global optimization method, and  the \Bezier patch with respect to
each quad in the quadrangulation is obtained by a local
optimization method yielding uniform and orthogonal iso-parametric
structures while keeping the $C^1/G^1$ continuity conditions between
patches. The proposed framework can be considered as a generalized
\Bezier extraction of a planar domain, and the resulting \Bezier
patches can be used as the computational elements for IGA  \cite{Borden}.
The efficiency and robustness of the
proposed approach are demonstrated by several examples.

In the future, we will focus on the following extensions:
\begin{enumerate}
\item Quad mesh generation significantly influences the
  parameterization quality. State-of-the-art quad meshing techniques will
  be investigated for IGA-suitable parameterization problem.
\item Several parameters are involved in the proposed optimization
  framework. Automatic selection of these parameters from the input
data will be addressed.
\item The proposed local optimization scheme can be improved by using
  Teichm\"{u}ller mapping  \cite{Nian:cmame2016} with theoretical guarantees of injectivity,
  and it can be accelerated  by parallelized implementation exploiting GPU/OpenMP.
\item  The proposed method can be extended to NURBS planar
  parameterization directly. Extension to three-dimensional volumetric
  parameterization problem with complex boundary representation will be a part of future work.
\end{enumerate}

\section*{Acknowledgment}

This research was supported by Zhejiang Provincial Natural Science Foundation of China under Grant
Nos. LR16F020003, LQ16F020005, the National Nature Science
Foundation of China under Grant Nos. 61472111, 61602138, 61502130,
and the Open Project Program of the State
Key Lab of CAD\&CG (A1703), Zhejiang University.

St\'ephane Bordas also thanks partial funding for his time provided by
the European Research Council Starting Independent Research Grant (ERC
Stg grant agreement No. 279578) `` RealTCut Towards real time multi-scale
simulation of cutting in non-linear materials with applications to
surgical simulation and computer guided surgery ''. St\'ephane Bordas
is also grateful for the support of the Fonds National de la Recherche
Luxembourg FWO-FNR grant INTER/FWO/15/10318764.

\section*{Appendix I. Conversion of  multiply-connected region}

In this part, an approach for converting multiply-connected region
into simply-connected region will be described.

If the region bounded by the given boundary curves is not a
simply-connected region, that is, it is  a
multiply-connected region, we have to convert
it into a simply-connected region. Suppose that the given
multiply-connected region has one outer boundary $\emph{\textbf{M}}_1=(\emph{\textbf{V}}_1,\emph{\textbf{E}}_1)$ and  one
interior boundary$\emph{\textbf{M}}_2=(\emph{\textbf{V}}_2,\emph{\textbf{E}}_2)$, then we find the
vertex  $\emph{\textbf{v}}_i \in \emph{\textbf{V}}_1=\{\emph{\textbf{v}}_0,\emph{\textbf{v}}_1,\cdots,\emph{\textbf{v}}_{n_1-1} \}$, $\emph{\textbf{u}}_j \in \emph{\textbf{U}}_2=
\{\emph{\textbf{u}}_0,\emph{\textbf{u}}_1,\cdots,\emph{\textbf{u}}_{n_2-1} \}$ such that the distance $\ell_0$ between
$\emph{\textbf{u}}_j$ and $\emph{\textbf{v}}_i$ is minimal among all the distance between two points
from the interior boundary to outer boundary. If we that is ,
\begin{equation}
\ell_0= \min \{\ell(\emph{\textbf{v}}_p, \emph{\textbf{u}}_q) \}, \emph{\textbf{v}}_p \in \emph{\textbf{V}}_1, \emph{\textbf{u}}_q \in \emph{\textbf{V}}_2.
\end{equation}
in which $\ell(\emph{\textbf{v}}_p, \emph{\textbf{u}}_q)$ is the distance between the vertex $\emph{\textbf{v}}_p$ and
$\emph{\textbf{u}}_q$. In order to generate a uniform quadrilateral decomposition,
some extra points should be added between the vertex $\emph{\textbf{u}}_j$ and
$\emph{\textbf{v}}_i$. The number $\nu$ of added vertex is computed as follows:
\begin{equation}
\nu=\lceil \frac{(n_1+n_2)\times
  \ell_0}{\sum\limits_{p=0}^{n_1-2}\ell(\emph{\textbf{v}}_p,\emph{\textbf{v}}_{p+1})+ \sum\limits_{q=0}^{n_2-2}
  \ell(\emph{\textbf{u}}_q,\emph{\textbf{u}}_{q+1})} \rceil
\end{equation}
then the added $\nu$ vertices $\emph{\textbf{z}}_k$ between $\emph{\textbf{v}}_i$ and $\emph{\textbf{u}}_j$ are obtained
by linear interpolation. The final boundary vertex sequence $\emph{\textbf{V}}$  of the simply-connected
region is
\begin{eqnarray}
\emph{\textbf{V}}&=&\{\emph{\textbf{v}}_0,\emph{\textbf{v}}_1,\cdots,\emph{\textbf{v}}_i, \emph{\textbf{z}}_0,\cdots, \emph{\textbf{z}}_{\nu-1}, \emph{\textbf{u}}_i,\cdots, \emph{\textbf{u}}_j,
\emph{\textbf{z}}_{\nu-1}, \cdots, \emph{\textbf{z}}_0, \emph{\textbf{v}}_i, \cdots, \emph{\textbf{v}}_{n_1-1} \}.
\end{eqnarray}

For the domain with multiple components, we can apply above method
iteratively, that is, in every iteration, we insert edges between the
vertices that have the shortest distance between the outer and
interior boundaries.

\section*{Appendix II.   Proof of Proposition \ref{th:sncondition1}}

\begin{proof}
Set $\textit{\textbf{P}}_{ij}=(x_{ij}^1,
x_{ij}^2)$. The gradient of the energy functional
(\ref{eq:quasifunc}) with respect to the coordinates $x_{ij}^\omega$
of an unknown control points, $\omega\in\{1,2\}$, can be computed
as follows,
\begin{eqnarray*}
\frac{\partial E(\emph{\textbf{r}})}{\partial x_{ij}^\omega}& = &
\frac{1}{2}\int_U \tau_1 ( <\textit{\textbf{r}}_{u},
\frac{\partial\textit{\textbf{r}}_{u}}{\partial x_{ij}^\omega}> + <\textit{\textbf{r}}_{v},
\frac{\partial\textit{\textbf{r}}_{v}}{\partial x_{ij}^\omega}>)+\tau_2 <\frac{\partial\textit{\textbf{r}}_{uu}}{\partial
x_{ij}^\omega},\textit{\textbf{r}}_{uu}> \\ & & + 2
\tau_2 <\frac{\partial\textit{\textbf{r}}_{uv}}{\partial
  x_{ij}^\omega},\textit{\textbf{r}}_{uv}>+\tau_2 <\frac{\partial\textit{\textbf{r}}_{vv}}{\partial
x_{ij}^\omega},\textit{\textbf{r}}_{vv}> dudv
\end{eqnarray*}
From the derivatives of the \Bezier surface over rectangular domain,
we have
\begin{eqnarray}
&&\emph{\textbf{r}}_{u}(u,v)=n\sum\limits_{k=0}^{n-1}\sum\limits_{l=0}^{n}{B_{k}^{n-1}(u)B_{l}^{n}(v)\Delta^{1,0}\textit{\textbf{P}}_{kl}},\label{eq:ru}\\
&&
\emph{\textbf{r}}_{v}(u,v)=n\sum\limits_{k=0}^{n}\sum\limits_{l=0}^{n-1}{B_{k}^{n}(u)B_{l}^{n-1}(v)\Delta^{0,1}\textit{\textbf{P}}_{kl}} \label{eq:rv}  \\
&&\emph{\textbf{r}}_{uu}(u,v)=n(n-1)\sum\limits_{k=0}^{n-2}\sum\limits_{l=0}^{n}{B_{k}^{n-2}(u)B_{l}^{n}(v)\Delta^{2,0}\textit{\textbf{P}}_{kl}},\label{eq:ruu}\\
&&\emph{\textbf{r}}_{uv}(u,v)=n^2\sum\limits_{k=0}^{n-1}\sum\limits_{l=0}^{n-1}{B_{k}^{n-1}(u)B_{l}^{n-1}(v)\Delta^{1,1}\textit{\textbf{P}}_{kl}},\label{eq:ruv}\\
&&\emph{\textbf{r}}_{vv}(u,v)=n(n-1)\sum\limits_{k=0}^{n}\sum\limits_{l=0}^{n-2}{B_{k}^{n}(u)B_{l}^{n-2}(v)\Delta^{0,2}\textit{\textbf{P}}_{kl}},\label{eq:rvv}
\end{eqnarray}
Hence,
\begin{eqnarray}
\frac{\partial \emph{\textbf{r}}_{u}}{\partial x_{ij}^\omega}
&=&n\left(B_{i-1}^{n-1}(u)-B_{i}^{n-1}(u)\right)
\textbf{\emph{e}}^\omega, \quad \label{eq:rux}
\frac{\partial \emph{\textbf{r}}_{v}}{\partial x_{ij}^\omega}
 =n\left(B_{j-1}^{n-1}(v)-B_{j}^{n-1}(v)\right) \textbf{\emph{e}}^\omega,\label{eq:rvx} \\
\frac{\partial \emph{\textbf{r}}_{uu}}{\partial x_{ij}^\omega}
&=&n(n-1)\left(B_{i-2}^{n-2}(u)-2B_{i-1}^{n-2}(u)+B_{i}^{n-2}(u)\right)
B_j^{m}(v)\textbf{\emph{e}}^\omega,\label{eq:ruux}\\
\frac{\partial \emph{\textbf{r}}_{uv}}{\partial x_{ij}^\omega} &=&n^2\left(B_{i-1}^{n-1}(u)-B_{i}^{n-1}(u)\right)\left(B_{j-1}^{m-1}(v)-B_{j}^{m-1}(v)\right)\textbf{\emph{e}}^\omega,\label{eq:ruvx}\\
\frac{\partial \emph{\textbf{r}}_{vv}}{\partial
x_{ij}^\omega}&=&n(n-1)B_i^{n}(u)\left(B_{j-2}^{n-2}(v)-2B_{j-1}^{n-2}(v)+B_{j}^{n-2}(v)\right
)\textbf{\emph{e}}^\omega,\label{eq:rvvx}
\end{eqnarray}
where $\textbf{\emph{e}}^1=(1,0),
\textbf{\emph{e}}^2=(0,1)$.

From (\ref{eq:ru}), (\ref{eq:rv}),(\ref{eq:ruu}), (\ref{eq:ruv}), (\ref{eq:rvv}),(\ref{eq:rux}),(\ref{eq:ruux}), (\ref{eq:ruvx}) and
(\ref{eq:rvvx}), it follows that
\begin{eqnarray*}
\frac{\partial E(\emph{\textbf{r}})}{\partial x_{ij}^\omega}&& =
\frac{1}{2} \tau_1 \int_U    n\left(B_{i-1}^{n-1}(u)-B_{i}^{n-1}(u)\right) < \textbf{\emph{e}}^\omega,\textit{\textbf{r}}_{u}> dudv \\
& & +\frac{1}{2} \tau_1 \int_U
n\left(B_{j-1}^{n-1}(v)-B_{j}^{n-1}(v)\right)<\textit{\textbf{r}}_{v},\textbf{\emph{e}}^\omega> dudv\\
& & +\frac{1}{2} \tau_2 \int_U
n(n-1)\left(B_{i-2}^{n-2}(u)-2B_{i-1}^{n-2}(u)+B_{i}^{n-2}(u)\right) B_j^{m}(v) < \textbf{\emph{e}}^\omega,\textit{\textbf{r}}_{uu}> dudv\\
& & +\tau_2 \int_U
n^2\left(B_{i-1}^{n-1}(u)-B_{i}^{n-1}(u)\right)\left(B_{j-1}^{m-1}(v)-B_{j}^{m-1}(v)\right)<\textit{\textbf{r}}_{uv},\textbf{\emph{e}}^\omega> dudv\\
& & +\frac{1}{2}\tau_2 \int_U
n(n-1)B_i^{n}(u)\left(B_{j-2}^{n-2}(v)-2B_{j-1}^{n-2}(v)+B_{j}^{n-2}(v)\right
)<\textit{\textbf{r}}_{vv}, \textbf{\emph{e}}^\omega> dudv
\end{eqnarray*}

The energy functional (\ref{eq:quasifunc}) has an extreme if and only if
$\frac{\D\partial E(\emph{\textbf{r}})}{\D \partial
x_{ij}^\omega}=0$, $\omega\in\{1,2,3\}$.  After some simple
computation by using the product formula of two
Bernstein polynomials in Proposition. \ref{prop:prop} and the integration formula
of Bernstein polynomials in Eq. (\ref{eq:integration}), Eq. (\ref{eqn:result}) can be obtained. Thus the proof is completed.
\end{proof}

\bibliographystyle{abbv}

\end{document}